\documentclass[12pt]{article}
\usepackage[colorlinks = true,
            linkcolor = blue,
            urlcolor  = blue,
            citecolor = blue,
            anchorcolor = blue]{hyperref}
\usepackage[margin=1in]{geometry}
\usepackage{setspace}
\usepackage[utf8]{inputenc}
\usepackage{authblk}
\usepackage{amsfonts, amssymb, amsthm, amsmath, mathrsfs, mathtools}
\usepackage{natbib, enumitem, setspace, multicol}
\usepackage{graphicx, float, caption, subcaption}
\usepackage{bm}
\usepackage{verbatim}
\usepackage{ctable,diagbox}
\usepackage{pdfpages}
\usepackage{fancyvrb}
\usepackage[linesnumbered,ruled]{algorithm2e}
\usepackage{algorithmic}
\usepackage{listings}

\usepackage{natbib}
\setcitestyle{authoryear,open={(},close={)}} 

\usepackage{xr}

\newtheorem{remark}{Remark}
\newtheorem{corollary}{Corollary}
\newtheorem{lemma}{Lemma}
\newtheorem{theorem}{Theorem}
\newtheorem{assm}{Assumption}
\newtheorem{setting}{Setting}

\newtheorem{defn}{Definition}[section]
\newtheorem{prop}{Proposition}[section]

\newtheorem*{result*}{Result}
\newtheorem*{remark*}{Remark}
\newtheorem*{claim*}{Claim}

\usepackage{bm}

\newcommand{\dto}{\stackrel{d}{\longrightarrow}}
\newcommand{\Pto}{\stackrel{p}{\longrightarrow}}
\newcommand{\asto}{\stackrel{a.s.}{\longrightarrow}}

 \newcommand{\ind}{\perp\!\!\!\!\perp} 

\newcommand{\del}{\partial}

\newcommand{\iid}{\stackrel{i.i.d.}{\sim}}
\newcommand{\Tstat}{T_{j}}

\newcommand{\repval}{p^r_{j,n}}

\allowdisplaybreaks
\numberwithin{equation}{section}

\title{The $\ell$-test: leveraging sparsity in the Gaussian linear model for improved inference}
\author{Souhardya Sengupta\footnote{Address for correspondence: Souhardya Sengupta, Department of Statistics, Harvard University, One Oxford Street, Science center, Cambridge, MA 02138, USA. Email: \href{mailto:ssengupta@g.harvard.edu}{ssengupta@g.harvard.edu}}}
\author{Lucas Janson}
\affil{Department of Statistics, Harvard University, Cambridge, MA, USA}
\date{}

\begin{document}
	\maketitle

 \begin{abstract}
We develop novel LASSO-based methods for coefficient testing and confidence interval construction in the Gaussian linear model with $n\ge d$. Our methods' finite-sample \textcolor{black}{validity is} identical to {\color{black}that} of their ubiquitous ordinary-least-squares-$t$-test-based analogues, yet have substantially higher power when the true coefficient vector is sparse. In particular, {\color{black}under sparsity} our coefficient test, which we call the $\ell$-test, performs like the \emph{one-sided} $t$-test (despite not being given any information about the sign), 
and {\color{black} $\ell$-test-based confidence intervals are correspondingly} shorter than the standard $t$-test-based intervals.
The nature of the $\ell$-test directly provides a novel exact adjustment conditional on LASSO selection for post-selection inference, allowing for the construction of post-selection $p$-values and confidence intervals. None of our methods require resampling or Monte Carlo estimation. 
\color{black} We perform a variety of simulations and a real data analysis on an HIV drug resistance data set to demonstrate the benefits of the $\ell$-test. 
\color{black} We additionally show that the $\ell$-test can be applied to a large class of asymptotically Gaussian estimators, dramatically expanding its applicability beyond linear models.
 \end{abstract}

\section{Introduction}
\label{sec:introduction}

\subsection{Motivation}
\label{sec:motivation}

Assume we have data $(\bm y, \bm X)$ from a (homoskedastic Gaussian) linear model: 
\begin{equation}
    \label{eqn:lm}
    \bm y \sim \mathcal{N}(\bm X\bm \beta, \sigma^2 \bm I),
\end{equation}
where $\bm X \in\mathbb{R}^{n\times d}$ is full column-rank (and in particular, assume $n\ge d$) and treated as non-random, and $\bm \beta
\in\mathbb{R}^d$ and $\sigma^2\in\mathbb{R}_{>0}$ are unknown. For a given covariate of interest $X_j$ this paper will consider testing $H_{j}:\beta_j = 0$ and the related problem of constructing a confidence interval for $\beta_j$. It will leverage the LASSO \citep{LASSO} to do so and our method's construction will also make it easy to construct conditionally valid versions, conditioned on LASSO selection. 

The go-to solution for this type of single covariate inference is based on the linear regression $t$-test for $H_j$, which can be efficiently inverted to obtain a $t$-test-based confidence interval for $\beta_j$.
The linear regression $t$-test dates back over a century \citep{fisher1922goodness} and is ubiquitous in introductory statistics courses and methods courses in nearly every domain of science and engineering. As a result, it is hard to overstate how universally widely used it 
is in practice. And it is easy to see why: the $t$-test is intuitive, easy to compute, and comes with strong theoretical guarantees.


The goal of this paper is to allow an analyst to leverage a belief in \emph{sparsity} (of $\bm \beta$) to conduct more informative inference (when sparsity holds) without sacrificing the statistical guarantees of the $t$-test (even when sparsity does not hold). 
Sparsity is a widely held belief throughout applications in 
\color{black}many fields, as evidenced by the widespread adoption of the LASSO \citep{LASSO} and other sparse regression methods; see \cite{rish2014sparse,64524c4c-35f2-309c-807b-12f074681130} for a textbook and review paper, respectively, on sparse regression in general, including discussion of applications, and, e.g., \cite{Lucas_Carvalho_Wang_Bild_Nevins_West_2006,annurev:/content/journals/10.1146/annurev-economics-061109-080451,https://doi.org/10.1002/cem.1418} for reviews on the important role of sparsity specifically in the domains of genomics (Section~\ref{sec:hiv} will demonstrate our methods on a concrete data set in this domain), economics, and chemistry, respectively. \color{black}
While leveraging sparsity in regression is a heavily studied subject (we review existing approaches in Section~\ref{sec:relatedwork}), methods developed to leverage it often rely on sparsity for \emph{both} validity and increased power,
while we explicitly seek to rely on it \emph{only} for increased power.

\subsection{Summary of our contributions}
\label{sec:summary}

We develop a hypothesis test for $H_j: \beta_j=0$, which we call the $\ell$-test, that uses the absolute value of the fitted LASSO coefficient as its test statistic. Using novel analysis of the conditional distribution of the LASSO estimator given the sufficient statistic of the linear model, we derive the test statistic's exact null distribution, allowing us to efficiently compute $p$-values without resampling or Monte Carlo. We argue that the $\ell$-test will often achieve nearly the power of the one-sided $t$-test when $\bm\beta$ is sparse, despite not knowing the sign of $\beta_j$ and remaining exactly valid under identical assumptions as the two-sided $t$-test. We show the $\ell$-test can be efficiently inverted to produce exact confidence intervals, which due to the power improvement of the $\ell$-test \textcolor{black}{are substantially shorter} than the $t$-test-based confidence intervals when $\bm\beta$ is sparse{\color{black}; in our simulations, we find the typical improvement under sparsity to be about $10\%$}. For both the $\ell$-test and its corresponding confidence interval, we show that a cross-validation procedure can be used to select the penalty parameter $\lambda$ in the LASSO from the data without impacting our validity guarantees, making our proposed procedures tuning-parameter-free.
%
The nature of the $\ell$-test and our formula for its null distribution make it straightforward to derive and compute (novel) post-selection $\ell$-test $p$-values such that the post-selection $\ell$-test $p$-value for $H_j$ is exactly (and non-conservatively) valid \emph{conditional} on the LASSO estimate of $\beta_j$ being nonzero; this conditional test can also be inverted to construct a confidence interval that is valid conditional on selection. \textcolor{black}{We also show that the $\ell$-test can be extended to the rich class of locally asymptotically normal models, yielding asymptotically valid $p$-values. Furthermore, we leverage the framework of Approximate Message Passing (AMP) to perform a theoretical power analysis in medium dimensions that sheds analytical light on the power benefits of the $\ell$-test.
}
A wide range of simulations and an application to HIV drug resistance demonstrate our methods to be powerful, efficient, and robust.

\subsection{Related work}
\label{sec:relatedwork}

The standard choice for testing $H_{j}:\beta_j = 0$ (and, via inversion, constructing confidence intervals) in the linear model is the (two-sided) $t$-test, or the one-sided $t$-test when the sign is known. Given the age and ubiquity of the linear model in statistics, we do not attempt to cover all related literature (though \citet{leibickel2020}, particularly their Appendix B, provides an excellent and detailed review), but just note that many such works focus on developing methods with some form of guarantees under weaker assumptions than the standard (homoskedastic) Gaussian linear model assumed in this paper \citep{friedman1937,pitman1937b,1938,krushkalwallis1952,tukey1958,hajek1962,adichie1967,jaeckel1972,efron1979,freedman1981,gutenbrunner1993,leibickel2020}. 
As their goal is robustness, these methods generally do not outperform the $t$-test in the (homoskedastic) Gaussian linear model, whereas this is exactly the goal of the current paper: maintain the same guarantees as the $t$-test while improving its power when $\bm{\beta}$ is sparse.
To our knowledge, the only other work that leverages a linear model's sparsity for testing an individual coefficient is the de-biased LASSO \citep{debiasedLASSO,zhangandzhang2014,JMLR:v15:javanmard14a}, but, when $n\ge d$, it either reduces exactly to the $t$-test or is only asymptotically valid under strong sparsity assumptions on $\bm\beta$; in contrast, the validity of the $\ell$-test holds regardless of the sparsity of $\bm\beta$.
There are a number of excellent works that aim to leverage sparsity for more powerful inference in the Gaussian linear model, including (fixed-X) knockoffs \citep{knockoff} and subsequent follow-ups that improve its performance (e.g., \citet{AS-LJ:2020,luo2022, renandbarber(2023),lee2024boosting}), as well as methods based on mirror statistics \citep{gaussmirror,junfdrdatasplit}, but all of these methods can only be used for variable selection and do not provide single-variable inference.

Two approaches share a similar goal as ours in seeking to improve the $t$-test's power without making further assumptions.
The first approach is \citet{HABIGER2014153}, which proposes to split the observations into two disjoint (thus independent) parts and uses the first part to estimate the sign of $\beta_j$ and then leverages the estimated sign in a test using the second part of the data. As we will discuss in Section~\ref{sec:improved_performance}, the $\ell$-test also leverages an estimated sign of $\beta_j$, but since it does not involve data splitting, it does not suffer from the associated loss in sample size and power. 
\color{black}The second approach is called Frequentist, assisted by Bayes (FAB), which, given a prior on $\bm\beta$ and $\sigma^2$, produces Bayes-optimal power (or confidence interval width) subject to maintaining the same frequentist validity as the $t$-test \citep{yuandhofflinearregression}. But this paper only considers (dense) Gaussian priors for $\bm\beta$ and hence, unlike the $\ell$-test, does not leverage sparsity\footnote{The discussion section of \cite{yuandhofflinearregression} mentions the possibility of using FAB with spike-and-slab priors to incorporate sparsity but does not pursue it}.\color{black}

While the goal of the $\ell$-test is most similar to the works mentioned so far, its approach is most closely related to the idea of conditioning on a sufficient statistic under the null hypothesis. This approach is perhaps most prominently used in constructing uniformly most powerful unbiased tests \citep{lehmannscheffe1955}, including the $t$-test. But this idea is also fundamental to co-sufficient sampling \citep{bartlett1937properties,CSS}, which is used for testing in a wide variety of contexts; see \citet{RB-LJ:2020} for a recent review of such tests. However, the $\ell$-test is not sampling-based, and besides, to the best of our knowledge, the only work applying co-sufficient sampling to testing in the linear model is \citet{huang2020}, but there it is applied to knockoffs \citep{knockoff,EC-ea:2018} which is a method for variable selection and cannot perform single-coefficient inference.

Related to this paper's post-selection inference methods, there is a rich literature on obtaining $p$-values that are valid post-selection \citep{Cox1975ANO,berk2013,fithianthesis,bootstrap_splitting,bachoc2016,tian2018selective} and in particular in the linear model conditioned on LASSO selection \citep{lassoinference,tibshirani2016,liu2018powerful,panigrahi2021integrative}.
Closest to our work is \citet{liu2018powerful}, which 
can be thought of as adjusting the standard $t$-test (or, more accurately, the $z$-test, since they assume $\sigma^2$ known) for a coefficient to make it valid conditional on that coefficient being selected by the LASSO. The conditional $\ell$-test can be thought of as making an analogous adjustment to the $\ell$-test rather than the $t$-test, resulting in similar gains under sparsity over \citet{liu2018powerful}'s method as the unconditional $\ell$-test achieves over the standard $t$-test; see Sections~\ref{sec:postselection}, \ref{sec:expt_lasso_adjusted_ci}, and Appendix~\ref{sec:app_expt_lasso_adjusted_ci} for further comparison and discussion.

\subsection{Notation}
Throughout this paper, we will use boldfaced symbols to denote matrices and vectors. For a matrix $\bm A$, unless otherwise stated, $\bm A_j$ denotes its $j^{\mathrm{th}}$ column, $\bm A_{-j}$ denotes its sub-matrix with the $j^{\mathrm{th}}$ column dropped while $A_{ij}$ (note the symbol is no longer boldfaced) denotes the entry in the $i^{\mathrm{th}}$ row and the $j^{\mathrm{th}}$ column. Similarly for a vector $\bm a\in \mathbb R^d$, unless otherwise stated, $a_j$ denotes its $j^{\mathrm{th}}$ entry while $\bm a_{-j}$ denotes the sub-vector of $\bm a$ without the $j^{\mathrm{th}}$ entry. We will also use $\mathbb I(\cdot)$ as the \textit{indicator function} that takes the value 1 if the condition within the parantheses is true, and 0 otherwise.

\subsection{Software}
\label{sec:software}
\color{black} An \textbf{\textsf{R}} \citep{r} package implementing our methods is available at \url{https://github.com/SSouhardya/ell_test_R} and scripts for replicating all analyzes in this paper are available at \url{https://github.com/SSouhardya/l-test}.

\color{black} 
\section{The $\ell$-test}
\label{sec:elltest}
Our proposed test for $H_j:\beta_j=0$ is primarily based on the simple idea of conditioning on sufficient statistics. The minimal sufficient statistic for the linear model under $H_j$ is $\bm S^{(j)} := (\bm X_{-j}^T\bm y, \bm y^T\bm y)$. So, by sufficiency, the conditional distribution of $\bm y\mid \bm S^{(j)}$ does not depend on any unknown parameters under $H_j$, and the same holds true for the conditional distribution of any fixed function of $\bm y$, including the LASSO coefficient estimate. 
Let
\[
    \hat{\bm \beta}^{\lambda} = \underset{\bm \beta}{\arg\min}\left(\frac{1}{2n}\|\bm y - \bm X \bm \beta\|^2 +\lambda \sum_{k}|\beta_k|\right)
\]
denote the LASSO estimator of the entire coefficient vector with penalty parameter $\lambda$. Denote the cumulative distribution function (CDF) for the conditional distribution of $\hat{\beta}^\lambda_j\mid \bm S^{(j)}$ under $H_j$, which we call the $\ell$-distribution, evaluated at some $b\in \mathbb{R}$, by 
\[F^\lambda_{\ell}(b\mid \bm S^{(j)}):=\mathbb{P}_{H_j}(\hat{\beta}^\lambda_j\leq b\mid \bm S^{(j)}).\]
By sufficiency, the $\ell$-distribution does not depend on any unknown parameters under $H_j$, and hence we can in principle compute a valid $p$-value for $H_j$ that rejects for large $|\hat{\beta}^\lambda_j|$:
\[\bar{F}_{|\ell|}^\lambda(|\hat{\beta}^\lambda_j|\mid \bm S^{(j)}), \]
where $\bar{F}_{|\ell|}^\lambda(b\mid \bm S^{(j)}) := \mathbb{P}_{H_j}(|\hat{\beta}^\lambda_j|\ge b\mid \bm S^{(j)}) = 1-\lim_{b'\rightarrow b^{-}} F_{\ell}^\lambda(b'\mid \bm S^{(j)}) + F^\lambda_{\ell}(-b\mid \bm S^{(j)})$ denotes the tail probability of $|\hat{\beta}^\lambda_j|$ and is entirely determined by the $\ell$-distribution $F_\ell^\lambda$.
We call the test that uses the above $p$-value the $\ell$-test, though our recommended usage of it involves two modifications, one about breaking ties in the $p$-value when $\hat{\beta}_j^\lambda=0$ (which currently cause a point mass at 1 in the $p$-value distribution) and the other about the choice of $\lambda$.
We defer these two choices to Sections~\ref{sec:smoothing} and \ref{sec:lambda_choice_singletest}, respectively, and first provide a characterization of the $\ell$-distribution which allows us to efficiently compute the $\ell$-test $p$-value and helps explain why, when, and by how much the $\ell$-test increases power over the $t$-test.

\subsection{The $\ell$-distribution}
\label{sec:elldist}
As the first step towards characterizing the $\ell$-distribution $F^{\lambda}_{\ell}(\cdot\mid \bm S^{(j)})$, we restate \citet[Proposition E.1]{luo2022} that exactly characterizes the distribution of $\bm y\mid \bm S^{(j)}$ under $H_j$. Let $\bm P_{-j} = \bm X_{-j}(\bm X_{-j}^T \bm X_{-j})^{-1}\bm X_{-j}^T$ denote the projection matrix onto the column space of $\bm X_{-j}$.

\begin{lemma}[\citet{luo2022}]
\label{lem:conddist}
    For the Gaussian linear model \eqref{eqn:lm}, define $\hat{\bm y}_{j} = \bm P_{-j}\bm y$ and
    $\hat \sigma_j^2 = \|\bm y -\hat{\bm y}_{j}\|^2$,
    and let $\bm V\in\mathbb{R}^{n\times (n-d+1)}$ denote an orthonormal matrix orthogonal to the column space of $\bm X_{-j}$ with first column given by $\bm V_1 = \frac{(\bm I - \bm P_{-j})\bm X_j}{\|(\bm I - \bm P_{-j})\bm X_j\|}$. Then, there exists a unique vector $\bm u\in \mathbb R^{n-d+1}$, such that $\|\bm u\| = 1$ and the following relation holds:
            \begin{equation}
        \label{eqn:conddist_lm}
        \bm y = \hat{\bm y}_j + \hat\sigma_j \bm V \bm u.
    \end{equation}
    Furthermore, under $H_j$,
    \begin{equation}
        \label{eqn:udist}
        \bm u\mid \bm S^{(j)} \sim \mathrm{Unif}\left(\mathbb S^{n-d}\right),
    \end{equation}
    where $\mathbb S^{n-d}$ denotes the unit sphere of dimension $n-d$.
\end{lemma}
For completeness, a proof of the lemma is provided in Section \ref{sec:proof_conddist} of the Appendix.  
Next our main theoretical result, Theorem~\ref{thm:distexpression}, establishes a mapping between $\hat \beta_j^\lambda$ and just the first element of $\bm u$, $u_1$ (there is nothing special about index 1 here except that we defined $\bm V$ to have only its first column non-orthogonal to $\bm X_j$), which will give an immediate characterization of $F_{\ell}^\lambda$
via the known distribution of $u_1$ from Equation~\eqref{eqn:udist}.

\begin{theorem}[Characterization of the $\ell$-distribution]
    \label{thm:distexpression}
    Consider the unique decomposition~\eqref{eqn:conddist_lm} from Lemma~\ref{lem:conddist} and
    for any $\lambda>0$, $b\in\mathbb{R}$ and $\epsilon\in\{-1,1\}$, define the functions
    \begin{equation}
    \label{eqn:betax}
        \hat{\bm \beta}^{\lambda}_{-j}(b):= \underset{\bm{\beta}_{-j}\in \mathbb R^{d-1}}{\arg\min}\left(\frac{1}{2n}\|\bm y -b \bm X_j - \bm X_{-j}\bm {\beta}_{-j}\|^2 + \lambda \|\bm {\beta}_{-j}\|_{1}\right),
    \end{equation}
    \[
        \Lambda_j(b, \epsilon) = \frac{-\bm X_j^T\left(\hat {\bm y}_j - b\bm X_j - \bm X_{-j}\hat{\bm \beta}_{-j}(b)\right)+ n\lambda\epsilon}{\hat \sigma_j\|(\bm I-\bm P_{-j})\bm X_j\|}.
    \]
    Then the function $f_{\bm S^{(j)}}:\mathbb{R}\mapsto \mathbb R$, defined via its inverse as
            \begin{align}
            \label{eqn:lambda_characterization}
            f_{\bm S^{(j)}}^{-1}(b) = \begin{cases}
                \Lambda_j(b, \mathrm{sign}(b)), & \text{if }b\neq 0\\
                [\Lambda_j(0, -1), \Lambda_j(0, 1)], & \text{if }b = 0
            \end{cases},
        \end{align}
        satisfies $\hat \beta_j^{\lambda} = f_{\bm S^{(j)}}(u_1)$ and is continuous and non-decreasing in its domain, and strictly increasing on the set $\{u: f_{\bm S^{(j)}}(u)\neq 0\}$.
\end{theorem}
Theorem \ref{thm:distexpression} exactly characterizes the $\ell$-distribution: 
for $b\neq 0$, we have
\begin{align}
    \label{eqn:dist_b_nonzero}
    \{\hat \beta_j^{\lambda} \leq b\} = \left\{f_{\bm S^{(j)}}^{-1}\left(\hat \beta_j^{\lambda}\right) \leq f_{\bm S^{(j)}}^{-1}\left(b\right)\right\} = \left\{u_1 \leq \Lambda_j(b, \mathrm{sign}(b))\right\}.
\end{align}
Thus, in particular, if we let $F_u$ denote the CDF of $u_1$ under $H_j$ \textcolor{black}{(because the null distribution of $u_1$ does not depend on $\bm S^{(j)}$, this is also the CDF conditional on $\bm S^{(j)}$)},
then $F_{\ell}^\lambda(b) = \mathbb{P}_{H_j}(\hat \beta_j^{\lambda} \leq b\mid \bm S^{(j)}) = F_u(\Lambda_j(b, \mathrm{sign}(b)))$ because $\Lambda_j(b,\mathrm{sign}(b))$ is a function of $\bm S^{(j)}$ for any fixed $b$. Similarly, for $b=0$, it follows from Theorem \ref{thm:distexpression} that $F_{\ell}^\lambda(0) = F_u\left(\Lambda_j(0,1)\right)$.
Note that $F_u$ is easily evaluated via a one-to-one mapping to a $t$-distribution, namely, $\frac{\sqrt{n-d}u_1}{\sqrt{1-u_1^2}}\sim t_{n-d}$ (see Appendix \ref{sec:proof_uquantiles} for a proof), which, along with the relations above, can be used to explicitly calculate quantiles of the $\ell$-distribution.
The proof of Theorem \ref{thm:distexpression} in Appendix \ref{sec:app_proof_distexpression} hinges on two main ideas---first,  we use block-wise coordinate descent to characterize the event $\{\hat \beta_j^\lambda = b\}$ in terms of $u_1$, and second, we characterize $f_{\bm S^{(j)}}$ by obtaining an exact expression for $\frac{\partial}{\partial u_1}f_{\bm S^{(j)}}$, which turns out to be non-negative throughout, thereby showing $f_{\bm S^{(j)}}$ is non-decreasing. We will see next that Theorem~\ref{thm:distexpression} also provides critical insights into the power of the $\ell$-test.

\subsection{The power of the $\ell$-test}
\label{sec:improved_performance}
Our simulations in Section \ref{sec:ltest_expt} show that  not only does the $\ell$-test consistently beat the usual two-sided $t$-test when $\bm\beta$ is sparse, it achieves power close to the \emph{one-sided} $t$-test (in the correct direction), being nearly identical in some cases, without any knowledge about the true sign of $\beta_j$. 

To explain this behavior, we first characterize the relationship between the $t$-test statistic and $u_1$ in Lemma \ref{lem:pvalequivalence}. It turns out that $u_1$ is a scalar multiple of $\bm X_j^T(\bm I - \bm P_{-j})\bm y$ (we argue this in Equations \eqref{eqn:olsequivalent} through \eqref{eqn:ols_uj} of the Appendix and is a major component in the proof of the lemma), and hence, $u_1$ is a measure of the association between $\bm X_j$ and the component of $\bm y$ that cannot be explained by the rest of the columns.

 \begin{lemma}
 \label{lem:pvalequivalence}
       Let $T_j$ denote the $t$-test statistic for testing $H_j:\beta_j = 0$. Then, there exists a continuous, strictly increasing, anti-symmetric function $g_{\bm S^{(j)}}$ that is a functional of the sufficient statistic $\bm S^{(j)}$, such that $T_j = g_{\bm S^{(j)}}(u_1)$.
\end{lemma} 
We prove this result in Section \ref{sec:proof_pvalequivalence} of the appendix. Lemma \ref{lem:pvalequivalence} and Theorem \ref{thm:distexpression} together show that a (one-sided) conditional-on-$\bm S^{(j)}$ test based on any of the test-statistics---$u_1,T_j$ and $\hat \beta_j^\lambda$, yield exactly the same $p$-values as long as the observed LASSO estimate is non-zero, as in this case all the three test statistics are strictly increasing in each other. We also know that $\Tstat$ is independent of $\bm S^{(j)}$, which follows from standard theory on ancillary statistics, however we supply a separate proof for this in {\color{black}Appendix~}\ref{sec:app_ttest_indep_proof}. This implies that a one-sided conditional test based on $\Tstat\mid \bm S^{(j)}$ yields exactly the one-sided $t$-test $p$-value, which by the above argument is exactly equal to the one-sided $p$-value of the conditional test based on $\hat \beta^\lambda_j \mid \bm S^{(j)}$ when the observed LASSO estimate is non-zero.

Even though the above paragraph establishes that conditional one-sided testing based on $\hat \beta^{\lambda}_j$ can do only as well as the corresponding one-sided $t$-test, we will now argue that the former test can gain considerable power over the $t$-test in the \emph{two}-sided testing regime.
As a first step, note that we can use Theorem \ref{thm:distexpression} to characterize the set $\mathcal R := \{u_1 : \hat \beta_j^{\lambda}\neq 0\}$ as a disjoint union of two intervals: $\mathcal R = (-\infty,v_{-})\cup(v_{+},\infty)$, where, $v_{\pm} = \Lambda_j(0,\pm 1)$.
\begin{figure}[h]
    \centering
    \includegraphics[height = 7cm, width = 15 cm]{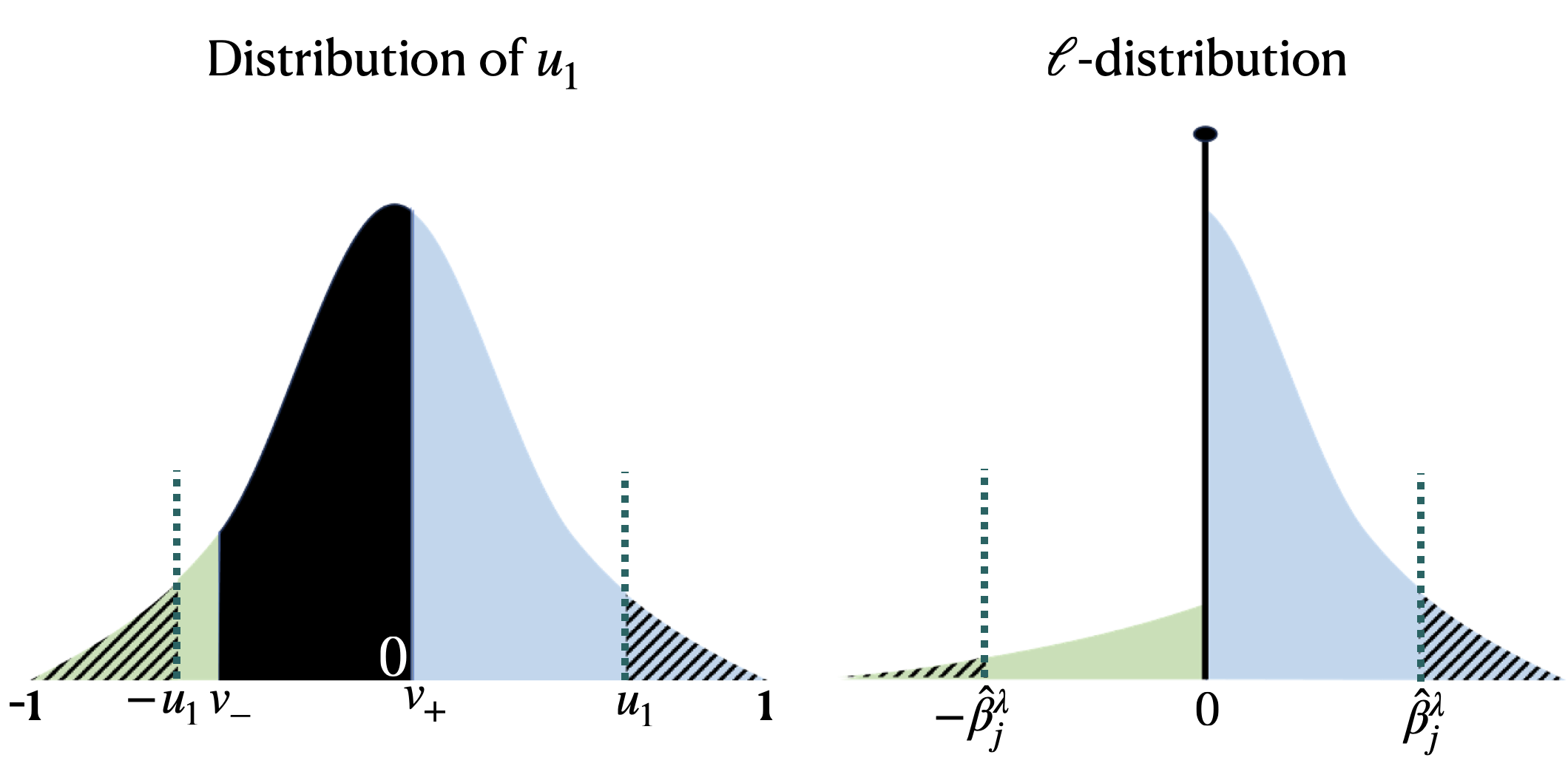}
    \caption{Connection between the distribution of $u_1$ (left) and the $\ell$-distribution (right). Regions of matching color between the two plots represent a one-to-one correspondence between the random variable values within them, as proved in Theorem~\ref{thm:distexpression}. \textcolor{black}{This figure is for demonstration and does not necessarily reproduce the exact shapes of the $u_1$- or $\ell$-distributions.}
    }
    \label{fig:t_and_hat_beta}
\end{figure}

We will argue that when $\beta_j\neq 0$, the test based on $|\hat \beta^{\lambda}_j|$ (i.e., the $\ell$-test) leverages the asymmetry of the interval $[v_-,v_+]$ about 0 to gain power. Observe that from Theorem \ref{thm:distexpression}, $\hat \beta_j^{\lambda}$ is negative when $u_1 \le v_-$, while it is positive when $u_1 \ge v_+$. Without loss of generality, we will assume that $\beta_j>0$ for the proceeding discussion, and assume the center of the interval $[v_{-},v_{+}]$ is negative (we will justify this latter assumption in a little bit). Consider Figure \ref{fig:t_and_hat_beta} for a visual representation of this, where under $H_j$, the left and the right figures show the conditional distributions of $u_1$ (i.e., $F_u$) 
and $\hat \beta_j^\lambda$ (i.e., the $\ell$-distribution $F_\ell^\lambda(\cdot\mid \bm S^{(j)})$), respectively. The correspondence between the two distributions is shown by matching colors---for example, as is evident from Theorem \ref{thm:distexpression}, the mass that the distribution of $u_1$ puts to the left of $v_{-}$ is exactly the mass the $\ell$-distribution puts on the negative half, and hence both these regions are colored green. As can be seen from the figure, the asymmetry in $[v_-,v_+]$ (and the symmetry of $F_u$) directly implies asymmetry in the $\ell$-distribution.

Since $\beta_j>0$, we expect that $\hat{\beta}_j^\lambda>0$ as well, as reflected in the right-hand plot, with corresponding positive $u_1$ also marked in the left-hand plot. Due to the symmetry of $u_1$'s distribution, the $p$-value of the two-sided test using $|u_1|$ is just twice the mass to the right of $u_1$, and by Lemma~\ref{lem:pvalequivalence}, this is also the $p$-value of the (two-sided) $t$-test. To understand the $\ell$-test $p$-value for comparison, note that by Theorem~\ref{thm:distexpression}, the mass to the right of $u_1$ in the left plot is exactly the mass to the right of $\hat{\beta}^\lambda_j$ in the right plot, but, critically, the mass to the \emph{left} of $-\hat{\beta}^\lambda_j$ in the right plot is \emph{far less} than the mass to the left of $-u_1$ in the left plot. Thus, the $\ell$-test's $p$-value, which is exactly the mass of the shaded regions in the right plot, is dominated by the right-most shaded region, whose mass is exactly the value of the \emph{one}-sided $t$-test.

Next, we argue why we expect the interval $[v_-,v_+]$ to lean opposite to the true sign of $\beta_j$, that is, towards the negative side in this case. Note that the interval $[v_-,v_+]$ has mid-point 
\begin{equation}
\label{eqn:midj}
       \hat m_j := \frac{v_-+v_+}{2} = \frac{\Lambda_j(0,-1) + \Lambda_j(0,+1)}{2} = \frac{-\bm X_j^T\left(\hat{\bm y}_j - \bm X_{-j}\hat{\bm \beta}^\lambda_{-j}(0)\right)}{\|(\bm I - \bm P_{-j})\bm X_j\|\hat \sigma_j}.
\end{equation}
We can think of $\hat{m}_j$ as an estimator of $m_j$, where $m_j$ has the exact same expression with $\hat{\bm \beta}^\lambda_{-j}(0)$ replaced by its estimand, $\bm \beta_{-j}$. Defining $\tau_j^2:= \bm X_j^T \bm P_{-j}\bm X_j\ge 0$, it can be seen that the numerator of $m_j$ satisfies
\begin{align*}
        -\bm X_j^T(\hat {\bm y}_{j}-\bm X_{-j}\bm \beta_{-j}) \sim  \mathcal{N}\left(-\beta_j \tau_j^2, \sigma^2 \tau_j^2\right).
\end{align*}
Thus under the alternative, $m_j$'s distribution is shifted towards the opposite sign of $\beta_j$ as long as $\bm X_j$ is not exactly orthogonal to $\bm X_{-j}$. And when $\bm \beta$ is sparse, we expect the lasso estimator $\hat{\bm \beta}_{-j}^\lambda(0)$ of $\bm\beta_{-j}$ to be a good one, and hence that $\hat{m}_j$'s distribution will also be shifted towards the opposite sign of $\beta_j$.
In particular, when $\beta_j>0$, this means we expect $[v_-,v_+]$ to be shifted in the negative direction, as we assumed it would be earlier in this subsection. 

In sum, when $\bm\beta$ is sparse, the LASSO leverages information in $\bm S^{(j)}$ (via $\hat{\bm \beta}_{-j}^\lambda(0)$) to guess the sign of $\beta_j$, and the point mass in the $\ell$-distribution uses that guess (via $[v_-,v_+]$) to reduce the ``wrong" tail of the $\ell$-test, resulting in a test with power approximating the one-sided $t$-test; see Appendix \ref{sec:app_improved_performance} for further discussion. Note that although the intuition for the power gain of the $\ell$-test over the $t$-test relied on sparsity, we emphasize that the \emph{validity} guarantees of the $\ell$-test remain identical to those of the $t$-test, and in particular do not require sparsity.

\subsection{Breaking ties when $\hat{\beta}_j^\lambda=0$}
\label{sec:smoothing}

As alluded to earlier in this section, the $\ell$-test $p$-value $\bar{F}_{|\ell|}^{\lambda}(|\hat \beta^{\lambda}_j|\mid \bm S^{(j)})$ is not $\mathrm{Unif}(0,1)$ under $H_j$. Instead, its conditional distribution given $\bm S^{(j)}$ is a mixture of $\mathrm{Unif}(0,\mathbb{P}_{H_j}(\hat{\beta}_j^\lambda\neq 0\mid \bm S^{(j)}))$ and a point mass at 1 of weight $\mathbb{P}_{H_j}(\hat{\beta}_j^\lambda = 0\mid \bm S^{(j)})$ because $|\hat{\beta}_j^\lambda| = 0$ is the ``least significant" value of the test statistic and occurs with positive probability.
It is preferable not to have such a point mass, since it makes the $\ell$-test somewhat conservative and because both users and many procedures which take $p$-values as inputs generally assume uniform null $p$-values. To remedy this, we need a way to break ties among data values that give $\hat{\beta}_j^\lambda=0$, since the test statistic $|\hat{\beta}_j^\lambda|$ does not distinguish between them. The strong connection between $\hat{\beta}_j^\lambda$ and $u_1$ established in Theorem~\ref{thm:distexpression} and visualized in Figure~\ref{fig:t_and_hat_beta} suggests that $u_1$, whose distribution is continuous on $[v_-,v_+]$ given the event $\{\hat{\beta}_j^\lambda=0\}$, provides a way forward. In particular, since $\hat{\beta}_j^\lambda$ approaches zero as $u_1$ approaches $v_-$ from the left or $v_+$ from the right, it is natural (and continuous in the data) to set $v_-$ and $v_+$ as tied for the ``most significant" values on the interval $[v_-,v_+]$, and then have the significance decrease as $u_1$ moves inward from those endpoints. This corresponds to breaking ties according to $|u_1-\hat{m}_j|$ when $\hat{\beta}_j^\lambda=0$, and can equivalently be thought of as using $|\hat{\beta}_j^\lambda| + \min\{|u_1-\hat{m}_j|-(v_+-\hat{m}_j),0\}$ as the test statistic. 
Recalling that $F_u$ denotes $u_1$'s CDF and defining $\bar{F}_{|u|}(u'\mid \bm S^{(j)}):=\mathbb{P}_{H_j}(|u_1-\hat{m}_j|\ge u'\mid \bm S^{(j)}) = 1-F_u(\hat{m}_j+u') + F_u(\hat{m}_j-u')$ as the tail probability of $|u_1-\hat{m}_j|$, we can express this $p$-value as
\begin{align}
    \label{eqn:ell_pval_previous}
    p^\lambda_j := \begin{cases}
        \bar{F}_{|\ell|}^{\lambda}(|\hat \beta^{\lambda}_j|\mid \bm S^{(j)}),&\textrm{ if }\hat \beta^{\lambda}_j \neq 0\\
        \bar{F}_{|u|}\left(|u_1 - \hat m_j|  \mid \bm S^{(j)}\right),&\textrm{ if }\hat\beta^\lambda_j = 0
    \end{cases},
\end{align}
which is exactly $\mathrm{Unif}(0,1)$ under $H_{j}$ and never larger than $\bar{F}_{|\ell|}^{\lambda}(|\hat \beta^{\lambda}_j|\mid \bm S^{(j)})$, the $p$-value proposed in Section \ref{sec:elldist}. 

\subsection{The choice of $\lambda$}
\label{sec:lambda_choice_singletest}
Thus far, we have treated $\lambda$ as a fixed tuning parameter, but in practice it is preferable to have an automated, data-dependent way to choose it. Standard practice for the LASSO is to choose $\lambda$ via cross-validation (on the full data $(\bm y, \bm X)$), but while this choice invalidates the theoretical guarantees of the $\ell$-test, just a slight modification of it is sufficient to retain those guarantees. Let $\tilde{\bm u}\sim\mathrm{Unif}(\mathbb{S}^{n-d})$ be drawn independently of $\bm u$, and plug it into Equation~\eqref{eqn:conddist_lm} and call the resulting left-hand side $\tilde{\bm y}$, so that $\tilde{\bm y}$ is conditionally independent of $\bm y$ given $\bm S^{(j)}$. Then it is easy to see that cross-validation on $(\tilde{\bm y}, \bm X_{-j})$ produces a $\lambda$, which we denote by $\hat{\lambda}_j$, that is exactly valid to use in the $\ell$-test, since conditioning on $\hat{\lambda}_j$ does not change the $\ell$-distribution\footnote{Cross-validation on $(\tilde{\bm y},\bm X)$ would also be valid and natural, but we prefer $(\tilde{\bm y},\bm X_{-j})$ for computational reasons; see Appendices~\ref{sec:app_min_vs_1se} and \ref{sec:app_single_test_randomness_lambda} for details on this and other ways to choose $\lambda$ that we considered, all of which we found to be empirically dominated by $\hat{\lambda}_j$}. 
And empirically, $\hat{\lambda}_j$ (our implementation uses 10-fold cross-validation) seems to be just as powerful as the (technically invalid) choice via cross-validation on $(\bm y, \bm X)$; see 
Appendix~\ref{sec:app_single_test_randomness_lambda}. Although $\hat{\lambda}_j$ is randomized through $\tilde{\bm u}$ and the random 10-fold partition of the data used by cross-validation, we find this exogenous randomness barely makes a difference: in our simulations it empirically accounts for at most about 0.8\% of the variability in the $\ell$-test $p$-value $p_j^{\hat{\lambda}_j}$; see Appendix~\ref{sec:app_single_test_randomness_lambda}. Furthermore, if desired, this randomness could be arbitrarily reduced by computing many conditionally independent $\hat{\lambda}_j$'s and using their mean or median for the $\ell$-test.

\subsection{Putting it all together: the $\ell$-test}
\label{sec:ell_test_complete}
We can now state our recommended implementation of the $\ell$-test: the $p$-value $p_j^{\hat{\lambda}_j}$ which combines the main $\ell$-test idea with the tie-breaking of Sections~\ref{sec:smoothing} and the $\lambda$ choice of Section~\ref{sec:lambda_choice_singletest}{\color{black}; see Algorithm~\ref{alg:ell_test} for formal statement}. 
\color{black}

\begin{algorithm}[h]
	\caption{\textcolor{black}{The $\ell$-test}}
	\label{alg:ell_test}
    \textcolor{black}{\textbf{Input:} Data: $(\bm y,\bm X_{n\times d})$, 
    index to test: $j$}\\
    \textcolor{black}{\textbf{Output: }The $\ell$-test $p$-value for $H_j:\beta_j=0$.}\\
    \textcolor{black}{ Compute $\tilde{\bm y} = \bm P_{-j}\bm y + \hat\sigma_j\bm V\tilde{\bm u}$, where $\tilde{\bm u} \sim \mathrm{Unif}\left(\mathbb S^{n-d}\right)$ independent of $(\bm y, \bm X)$.}\\
     \textcolor{black}{Choose $\hat \lambda_j$ via 10-fold cross-validation of the LASSO on $(\tilde{\bm y},\bm X_{-j})$.}\\
    \textcolor{black}{ \textbf{Return:} $p_j^{\hat \lambda_j}$ computed by using $\lambda=\hat\lambda_j$ in Equation~\eqref{eqn:ell_pval_previous}.}\\
\end{algorithm}\color{black}

Computing $p_j^{\hat{\lambda}_j}$ requires, aside from a cross-validated LASSO to compute $\hat{\lambda}_j$, one (non-cross-validated) LASSO to compute $\hat{\bm \beta}^{\hat{\lambda}_j}$ and, if $\hat{\beta}_j^{\hat{\lambda}_j}\neq 0$, then a second (non-cross-validated) LASSO to compute $\hat{\bm \beta}_{-j}^{\hat{\lambda}_j}(-\hat{\beta}_j^{\hat{\lambda}_j})$. When $\hat{\beta}_j^{\hat{\lambda}_j}\neq 0$, these two LASSO's allow us to compute the two tails for the $\ell$-test $p$-value via Theorem~\ref{thm:distexpression}, since $\hat{\bm \beta}_{-j}^{\hat{\lambda}_j}(\hat{\beta}_j^{\hat{\lambda}_j}) = \hat{\bm \beta}_{-j}^{\hat{\lambda}_j}$. 
And when $\hat{\beta}_j^{\hat{\lambda}_j} = 0$, by Equation~\eqref{eqn:ell_pval_previous}, the only LASSO quantity needed is $\hat{\bm \beta}_{-j}^{\hat{\lambda}_j}(0)$ to compute $\hat{m}_j$, but since $\hat{\bm \beta}_{-j}^{\hat{\lambda}_j}(0) = \hat{\bm \beta}_{-j}^{\hat{\lambda}_j}$ in this case, no additional LASSO run is needed beyond the first. Thus computation for the $\ell$-test requires just a very small constant number of LASSO runs.

It is an immediate consequence of our construction that $p_j^{\hat\lambda_j}$ is valid and non-conservative under no further assumptions than the (homoskedastic Gaussian) linear model, which we formally state here as a corollary.
\begin{corollary}[Validity of the $\ell$-test] For model~\eqref{eqn:lm},  for all $\alpha\in [0,1]$,
$\mathbb P_{H_j}(p_j^{\hat \lambda_j}\leq \alpha) = \alpha$.
\end{corollary}

\color{black}

\subsection{Asymptotic extension beyond the Gaussian linear model}
\label{sec:asymp}

Due to the close relationship between the Gaussian linear model and the multivariate Gaussian distribution, it is intuitively easy to extend the $\ell$-test to a large class of models that admit asymptotically Gaussian estimators. In particular, consider a sequence of models and for each model, a parameter vector $\bm \theta_n$ in a local asymptotic regime, i.e., $\bm \theta_n = \bm\theta/\sqrt{n}$ for some fixed vector $\bm\theta$. Then for any estimator $\hat{\bm\theta}_n$ such that $\sqrt{n}(\hat{\bm \theta}_n - \bm \theta_n)\dto \mathcal N(\bm 0, \bm \Sigma)$ for some fixed positive definite covariate matrix $\bm\Sigma$ with its own consistent estimator $\hat{\bm \Sigma}_n$, we can compute response vector $\bm y = \sqrt{n}\hat{\bm \Sigma}_n^{-1/2}\hat{\bm \theta}_n$ and design matrix $\bm X = \hat{\bm \Sigma}_n^{-1/2}$ and apply the $\ell$-test to $(\bm y,\bm X)$, which asymptotically satisfies a Gaussian linear model with coefficient vector $\bm\theta$ and error variance $\sigma^2 = 1$ (see \citet[Theorem 4]{li2021whiteout} for an asymptotic result like this for knockoffs). In fact, since $\sigma^2$ is known in this asymptotic regime, we can apply a version of the $\ell$-test that leverages the known $\sigma$ via some simple substitutions; for brevity of presentation, we defer these details to Appendix~\ref{sec:known_sigma_ell}. In the remainder of this section we will formalize this asymptotic application of the known-$\sigma$-$\ell$-test, but first pause to emphasize that this vastly expands the applicability of the $\ell$-test framework, e.g., to M-estimators (including the maximum likelihood estimator) in models that are differentiable in quadratic mean 
\citep
{vandervaart1998asymptotic}.
In particular, the ordinary least squares estimator paired with the robust sandwich variance estimator \citep{white1980} allows the $\ell$-test to be applied asymptotically to a \emph{heteroskedastic} linear model.



In addition to the aforementioned local asymptotic regime, we will need continuity and uniqueness of the cross-validation procedure in the $\ell$-test. These properties hold under very mild regularity conditions, but we defer their cumbersome statement to the appendix as Assumption~\ref{assm:cv}.

\begin{theorem}[The $\ell$-test for asymptotically Gaussian models]
    \label{thm:asymp}
    For some $\bm \theta$ and $\bm \Sigma$, let $\bm \theta_n = \bm \theta/\sqrt{n}$ and suppose that the sequence of random variables $\left\{\left(\hat{\bm \theta}_n, \hat {\bm \Sigma}_n\right)\right\}_n$ satisfies
    \[
        \sqrt{n}(\hat{\bm \theta}_n - \bm \theta_n)\dto \mathcal N(\bm 0, \bm \Sigma)\qquad \textrm{ and }\qquad \hat{\bm \Sigma}_n \Pto \bm \Sigma.
    \]
     Denote the known-$\sigma$-$\ell$-test $p$-value for testing $H_j$, applied to data ${\left(\bm y=\sqrt{n}\hat{\bm \Sigma}_n^{-1/2}\hat{\bm \theta}_n,\, \bm X = \hat{\bm \Sigma}_n^{-1/2}\right)}$, by $p_{j, n}$.
     Then under Assumption~\ref{assm:cv}, $p_{j, n}$ is asymptotically valid for testing $\theta_j=0$, i.e., if $\theta_j=0$, then
    \[
        \underset{n\rightarrow \infty}{\lim\sup}~\mathbb P\left(p_{j, n}\leq \alpha\right) = \alpha,\; \forall \alpha \in [0,1].
    \]
\end{theorem}
    The proof of this theorem, available in Appendix~\ref{sec:proof_asymp}, relies on showing that the known-$\sigma$-$\ell$-test $p$-value is a continuous function of the data $(\bm y, \bm X)$, followed by an application of the continuous mapping theorem. As an intermediate result, the known-$\sigma$-$\ell$-test $p$-value using a fixed $\lambda$ is also asymptotically valid for the above class of models (without the need for Assumption~\ref{assm:cv}), which we formally state in Corollary~\ref{cor:validity_known_sigma_ell_test} of the appendix.
    
    
    
    The understanding of the $\ell$-test's power gain developed in Section~\ref{sec:improved_performance} suggests that the asymptotic $\ell$-test described in Theorem~\ref{thm:asymp} will outperform the standard asymptotic $z$-test when
    (1) $\bm \theta$ is sparse and (2) $\bm\Sigma$ has sufficient off-diagonal entries (the latter ensures the columns of $\bm X = \hat{\bm \Sigma}_n^{-1/2} \approx \bm \Sigma^{-1/2}$ are not too close to orthogonal).

\color{black}

\color{black}
\section{Power analysis}
\label{sec:power_analysis}
As Section \ref{sec:improved_performance} argues, the power gain of the $\ell$-test under sparsity comes from the informative re-centering by $\hat{m}_j$ that guesses the sign of $\beta_j$. In this section we study this power gain theoretically by considering the test that simply re-centers $u_1$ by $\hat m_j$, i.e., it rejects for large values of $|u_1 - \hat m_j|$ conditional on $\bm S^{(j)}$. This test gains power over the $t$-test via the same re-centering argument as the $\ell$-test does, but the latter is less tractable to analyze theoretically.
To further ease our analysis, we treat $\sigma$ as known, allowing us to instead work with a known-$\sigma$-analogue of the statistic $|u_1 - \hat m_j|$, details of which are discussed in Section~\ref{sec:known_sigma_ell} of the appendix.  Finally, we treat $\lambda$ as fixed instead of chosen by cross-validation as in Section \ref{sec:lambda_choice_singletest}, and denote the $p$-value for this re-centered test when the sample size is $n$ as $p_{j,n}^{\mathrm{r},\lambda}$. We find in Appendix~\ref{sec:power_curves} that this test (with $\lambda$ fixed at the mean of its cross-validated value computed from independent simulations) empirically performs indistinguishably from the (known-$\sigma$-analogue of the) $\ell$-test, suggesting that despite the simplifications we make in this section for mathematical tractability, both qualitative and quantitative conclusions from our power analysis should apply to the $\ell$-test as well.

The primary technical toolbox we rely on is Approximate Message Passing (AMP) \citep{bayati2012}, which allows us to precisely characterize the behavior of the LASSO estimate inside the known-$\sigma$-analogue of $\hat m_j$ under a proportional asymptotic limit. In particular, we assume the following random design and nuisance parameter regime.


\begin{setting}\label{setting:amp} $d/n\to \kappa\in (0,1)$. The entries of $\bm X$ are drawn i.i.d. from $\mathcal{N}(0,1/n)$, the entries of $\bm \beta_{-j}$ are drawn i.i.d. from a distribution $B_0$ with finite second moment, and $\beta_j=h$ for some fixed scalar $h$ which we take, without loss of generality, to be non-negative.
\end{setting}

This type of setting is standard in AMP theory, and will allow us to characterize the asymptotic $p$-values of the tests we are considering in terms of a particular bivariate Gaussian distribution: let
\begin{align}
    \label{eqn:ab_dist}
        \begin{pmatrix}
            F\\
            G
        \end{pmatrix}\sim \mathcal N \left(\begin{pmatrix}
            \frac{h\sqrt{1-\kappa}}{\sigma}\\
            \frac{h\lambda}{ \alpha_\lambda \tau_\lambda\sigma\sqrt{1-\kappa}}
        \end{pmatrix}, 
        \begin{bmatrix}
            1 & 1\\
           1 & \frac{\lambda^2}{\alpha_\lambda^2\sigma^2(1-\kappa)}
        \end{bmatrix}\right),
        \end{align}
        where $\alpha_\lambda,\tau_\lambda$ satisfy the two AMP state-evolution equations \citep{bayati2012} (using notation from \citet{crt_power} as we use their results heavily in our proofs):
        \begin{align*}
    \lambda = \alpha_\lambda \tau_\lambda\left(1-\kappa\mathbb E\left[\eta'(B_0+\tau_\lambda W;\alpha_\lambda \tau_\lambda)\right]\right);\hspace{0.7cm}\tau_\lambda^2 = \sigma^2 + \kappa\mathbb E\left[\left(\eta(B_0+\tau_\lambda W;\alpha_\lambda \tau_\lambda) - B_0\right)^2\right],
        \end{align*}
        where $\eta(x;\omega):=\mathrm{sign}(x)(|x|-\omega)_+$ is the soft-thresholding function and $\eta'$ denotes its derivative with respect to its first argument. Note that although the AMP state-evolution equations are nonlinear and do not admit a closed form solution for $(\alpha_\lambda,\tau_\lambda)$, they do uniquely determine them as a function of only $\sigma^2$, $\lambda$, $\kappa$, and $B_0$. Furthermore, $\tau_\lambda$ can be interpreted as a measure of the asymptotic error of the LASSO because
        the asymptotic mean squared estimation error of the LASSO is $\kappa^{-1}(\tau^2_{\lambda} - \sigma^2)$ \citep[Corollary 1.6]{bayati2012}.
        
        First, $F$ represents the asymptotic distribution of the $z$-test statistic, so we can easily relate 
        the one- and two-sided $z$-test $p$-values (the known-$\sigma$ analogues of the $t$-tests), denoted $p_{j,n}^{z,1}$ and $p_{j,n}^{z,2}$, respectively, to $F$ in Setting~\ref{setting:amp} (see Appendix~\ref{sec:proof_thm_power} for derivation):
        \[ p_{j,n}^{z,1} \dto 1-\Phi\left(F \right) \quad\text{and}\quad p_{j,n}^{z,2} \dto 1-\left[\Phi\left(F \right) - \Phi\left(-F \right)\right]\mathrm{sign}(F), \]
        where the precise form of these expressions is chosen for ease of comparison with that for $p_{j,n}^{\mathrm{r},\lambda}$ in the following Theorem, whose proof appears in Appendix~\ref{sec:proof_thm_power}.
        
\begin{theorem}\label{thm:power}
    Under Setting~\ref{setting:amp}, 
    $\;\;p_{j,n}^{\mathrm{r},\lambda}\dto 1-\left[\Phi\left(F\right) - \Phi\left(F-2G\right)\right]\mathrm{sign}(G).$
\end{theorem}

Comparing asymptotic expressions for $p_{j,n}^{z,1}$, $p_{j,n}^{z,2}$, and $p_{j,n}^{\mathrm{r},\lambda}$, we see that if $G= F$, then the re-centered test has the same asymptotic power as the two-sided $z$-test, while as $G$'s distribution lies further and further to the right of 0 and of $F$'s distribution, the re-centered test's asymptotic power approaches that of the one-sided $z$-test. Although $G$'s distribution depends on many quantities, its mean is inversely proportional to $\tau_\lambda$,
suggesting that strong LASSO performance moves $G$'s distribution to the right, matching the intuition from Section~\ref{sec:improved_performance} that the power benefits of our re-centering increase with the performance of the LASSO. In Appendix~\ref{sec:app_power_analysis} we study further how $G$'s distribution depends on the data-generating distribution and how this translates to power for the re-centered and $\ell$-tests. Key takeaways include that the match between the asymptotic theoretical power from Theorem~\ref{thm:power} and empirical finite-sample power is quite strong, and that the asymptotic benefit of re-centering in Setting~\ref{setting:amp} is quite resilient to coefficient density, providing considerable gains across signal sizes even when $B_0$ is a Gaussian distribution (Figure~\ref{fig:power_normal}) and showing no apparent power loss to the standard 2-sided test even in the ``densest" setting when all entries of $\bm \beta$ take the same non-zero value (Figure~\ref{fig:power_ber}). Note we can think of Gaussian coefficients as somewhat/approximately sparse, in the sense that some coefficients are closer to zero than others, even if none are exactly zero; this seems to be sufficient for the $\ell$-test to gain power via re-centering. Section~\ref{sec:experiments} contains further empirical demonstrations and discussion of the power benefits of re-centering.

\begin{remark}
    Although the mathematically simpler direct re-centering approach we consider in this section leads to a viable alternative to the $\ell$-test, the focus of this paper remains on the $\ell$-test as presented in Section~\ref{sec:elltest} because it (a) is more interpretable for most (especially non-statistician) users as it is based on a natural estimator of $\beta_j$ under sparsity, (b) has similar run-time as the re-centering test (both being dominated by the cross-validated LASSO), and (c) has an immediate extension to post-selection inference (Section~\ref{sec:postselection}). 
\end{remark}
\color{black}

\section{$\ell$-test confidence intervals}
\label{sec:ciforbeta}

If we can use the $\ell$-test to test $H_{j}(\gamma):\beta_j = \gamma$ for any $\gamma\in\mathbb{R}$, then this family of tests can be inverted to obtain a valid confidence region for $\beta_j$.
But extending the $\ell$-test to $H_j(\gamma)$ is straightforward, since $\bm y$ satisfying $H_j(\gamma)$ is equivalent to $\bm y-\gamma \bm X_j$ satisfying $H_j(0)=H_j$, so we can simply apply the regular $\ell$-test (exactly as detailed in Section~\ref{sec:elltest}) to the data $(\bm y - \gamma \bm X_j, \bm X)$.
Defining $p_j^{\hat{\lambda}_j(\gamma)}(\gamma)$ as the $\ell$-test $p$-value for $H_j(\gamma)$, the $100(1-\alpha)\%$ $\ell$-test confidence region is given by $\{\gamma\in\mathbb{R}: p_j^{\hat{\lambda}_j(\gamma)}(\gamma) > \alpha\}$, and for interpretability purposes, we take its convex hull (i.e., the smallest interval containing it) as our $\ell$-test confidence interval:
\[ \hat{C}_j := \mathrm{conv}(\{\gamma\in\mathbb{R}: p_j^{\hat{\lambda}_j(\gamma)}(\gamma) > \alpha\}).\]
The validity of $\hat{C}_j$ follows directly from that of the $\ell$-test $p$-values $p_j^{\hat{\lambda}_j(\gamma)}(\gamma)$ and the fact that taking the convex hull can only make a set bigger and hence only increase its coverage. We recommend using the same $\tilde{\bm u}$ and cross-validation partition for all $\gamma$ when constructing $\hat{C}_j$, so that the slight randomness in the $\ell$-test $p$-values is consistent across $\gamma$. 

Computationally, one may be concerned that computing $\hat{C}_j$ requires many LASSO runs for a fine grid of $\gamma$ values. 
It is known \citep{lars,piecewiselinear} that the LASSO solution
is piecewise linear in $\lambda$ and that these paths can be generated by efficient algorithms, but in Appendix~\ref{sec:app_beta_path}, we show that the LASSO solution ($\hat {\bm \beta}^\lambda_{-j}(\gamma)$) is also piecewise linear in $\gamma$ (for fixed $\lambda$) and we provide an algorithm to efficiently generate these paths as well.
Combining these two path-generating algorithms (in $\lambda$ and $\gamma$) provides an efficient way to share computation to efficiently compute all the LASSO solutions needed for $\hat{C}_j$.

\section{$\ell$-test inference conditional on LASSO selection}
\label{sec:postselection}

The $\ell$-test $p$-value's distributional form \eqref{eqn:ell_pval_previous} makes it extremely straightforward (both conceptually and computationally) to adjust it to be conditionally valid given $\hat{\beta}_j^\lambda\neq 0$: simply divide the $\ell$-test $p$-value $p_j^\lambda$ by $r_j^\lambda:=\mathbb{P}_{H_j}(\hat{\beta}_j^\lambda\neq 0\mid \bm S^{(j)}) = \mathbb{P}_{H_j}(u_1\notin [v_-,v_+] \mid \bm S^{(j)}) = 1 - F_u(v_+)+F_u(v_-)$ to get $\underline{p}_j^\lambda:= p_j^\lambda/r_j^\lambda$. The fact that, under $H_j$, $\underline{p}_j^\lambda$ has a $\mathrm{Unif}(0,1)$ distribution conditional on $\bm S^{(j)}$ and $\hat{\beta}_j^\lambda\neq 0$ follows because $r_j^\lambda$ is exactly the supremum value $p_j^\lambda$ can take as long as $\hat{\beta}_j^\lambda \neq 0$, and the density of $p_j^\lambda$ between 0 and $r_j^\lambda$ is uniformly distributed (see Section~\ref{sec:smoothing}); it follows further that $\underline{p}_j^\lambda$'s null distribution conditional only on $\hat{\beta}_j^{\lambda}\neq 0$ is also $\mathrm{Unif}(0,1)$. 
\begin{corollary}[Validity of the conditional $\ell$-test] For model~\eqref{eqn:lm},  for all $\alpha\in [0,1]$,
$\mathbb P_{H_j}(\underline{p}_j^{\lambda}\leq \alpha\mid \hat{\beta}_j^{\lambda}\neq 0) = \alpha$.
\end{corollary}
Computationally, only one extra LASSO (for $\hat{\bm \beta}_{-j}^\lambda(0)$, which goes into $r_j^\lambda$) needs to be run to compute $\underline{p}_j^\lambda$ for an index $j$ with $\hat{\beta}_j^\lambda\neq 0$.
And since everything above is conditional on $\bm S^{(j)}$, the same result holds true when using $\hat{\lambda}_j$ from Section~\ref{sec:lambda_choice_singletest} (i.e., $\underline{p}_j^{\hat{\lambda}_j}$ is conditionally valid given $\hat{\beta}_j^{\hat{\lambda}_j}\neq 0$), since $\hat{\lambda}_j$ is conditionally independent of the data given $\bm S^{(j)}$. 

Now for obtaining post-LASSO-selection valid confidence interval for $\beta_j$, we need to invert a conditionally valid test for $H_j(\gamma), \gamma \in \mathbb R$. For testing $H_j(\gamma)$, following the suggestion in Section \ref{sec:ciforbeta}, the test statistic should be based on the LASSO estimate on $(\bm y - \gamma \bm X_j, \bm X)$, that is $|\hat \beta_j^\lambda(\gamma)|$, whereas the model selection event is still based on the original, un-centerd LASSO estimate, $\hat \beta_j^\lambda$. Furthermore, one can choose to use a different penalty parameter for the test statistic than the one used for the selection event. This prompts us to understand tests for $H_j(\gamma)$ based on $|\hat \beta_j^{\lambda_\ell}(\gamma)|$, valid conditionally on $\{\hat \beta_j^{\lambda_s}\neq 0\}$, where $\lambda_\ell$ and $\lambda_s$ need not be the same. Because now the test-statistic and the selection event are based on different LASSO estimates, a conditional $p$-value would not have such a simple form as for $\underline{p}^{\lambda}_j$, but we can still obtain a valid $p$-value using CDF transforms if we can characterize the distribution of $|\hat \beta_j^{\lambda_{\ell}}(\gamma)|\mid \{\hat \beta_j^{\lambda_s}\neq 0\}$, under $H_j(\gamma)$.

First note that, as discussed in Section \ref{sec:ciforbeta}, Lemma \ref{lem:conddist} can be applied to $(\bm y - \gamma \bm X_j, \bm X)$ under $H_j(\gamma)$ to show that $\bm S^{(j)}(\gamma) = (\bm X_j^T\bm y, \hat \sigma_j(\gamma))$ is sufficient under $H_j(\gamma)$, where, $\hat \sigma_j(\gamma) = \|(\bm I - \bm P_{-j})(\bm y - \gamma \bm X_j)\|$, and that $\bm y - \gamma \bm X_j$ can be written as $\bm P_{-j}\left(\bm y - \gamma \bm X_j\right) + \hat \sigma_j(\gamma)\bm V \bm u^{\gamma}$, with $\bm u^{\gamma}\mid \bm S^{(j)}(\gamma)\stackrel{H_j(\gamma)}{\sim} \mathrm{Unif}(\mathbb S^{n-d})$. In light of this result, one can now apply Theorem \ref{thm:distexpression} to $(\bm y - \gamma \bm X_j, \bm X)$ to conclude that $\hat \beta_j^{\lambda_{\ell}}(\gamma)\leq b$ if and only if $u_1^{\gamma}\leq \Lambda_j^*(b, \mathrm{sign}(b);\gamma)$, where $ \Lambda_j^*(b, \epsilon;\gamma)$ is exactly equal to $\Lambda(b,\epsilon)$ but with $\bm y$ replaced with $\bm y - \gamma \bm X_j$ and $\mathrm{sign}(0):=1$. In fact Theorem \ref{thm:general_distexpression} (that characterizes the distribution of $\hat \beta_j^{\lambda}\mid \bm S^{(j)}(\gamma)$ under $H_j(\gamma)$) in Appendix \ref{sec:app_ldist} shows that this same $u_1^{\gamma}$ can be used to characterize the event $\{\hat \beta_j^{\lambda_s} = 0\}$, stating it is equivalent to $\{u^{\gamma}_{1}\in [\Lambda_j(0,\pm 1;\gamma)]\}$, where,
\[
    \Lambda_j(0,\pm 1;\gamma) = \frac{-\bm X_j^T(\hat {\bm y}_j + \gamma (\bm I - \bm P_{-j})\bm X_j - \bm X_{-j}\hat {\bm \beta}^{\lambda_s}_{-j}(0)) \pm n\lambda_s}{\hat \sigma_j(\gamma)\|(\bm I - \bm P_{-j})\bm X_j\|}.
\]
Because $\Lambda_j(0,\pm 1;\gamma)$ and $\Lambda_j^*(b, \mathrm{sign}(b);\gamma)$ are all functions of $\bm S^{(j)}(\gamma)$, we have that
\begin{align*}
    \underline{F}^{\lambda_{\ell},\lambda_{s}}_{\ell}(b\mid \bm S^{(j)}(\gamma);\gamma)&:=\mathbb P_{H_{j}(\gamma)}(\hat \beta_j^{\lambda_{\ell}}(\gamma)\leq b\mid \hat \beta_j^{\lambda_s}\neq 0, \bm S^{(j)}(\gamma))\\
    &= \mathbb P_{H_j(\gamma)}(u_1^{\gamma}\leq \Lambda_j^*(b, \mathrm{sign}(b);\gamma)\mid u^{\gamma}_{1}\notin [\Lambda_j(0,\pm 1;\gamma)], \bm S^{(j)}(\gamma)),
\end{align*}
where the last expression can exactly be evaluated using the known quantiles of $F_u$ defined in Section \ref{sec:elldist}. This gives us the adjusted $p$-value for testing $H_j(\gamma)$:
\[
    \underline{p}_j^{\lambda_{\ell},\lambda_{s}}(\gamma) := 1-  \underline{F}^{\lambda_{\ell},\lambda_{s}}_{\ell}(|\hat \beta_j^{\lambda_{\ell}}(\gamma)|\mid \bm S^{(j)}(\gamma);\gamma) + \lim_{b\rightarrow |\hat \beta^{\lambda_{\ell}}_j(\gamma)|^{-}} \underline{F}^{\lambda_{\ell},\lambda_{s}}_{\ell}(b\mid \bm S^{(j)}(\gamma);\gamma),
\]
where we can also use the strategy in Section \ref{sec:smoothing} to break ties when $\hat \beta^{\lambda_{\ell}}_j = 0$. As one would expect, our original conditional $\ell$-test $p$-value $\underline{p}_j^\lambda$ is a special case of the above, with
$\underline{p}_j^{\lambda} = \underline{p}_j^{\lambda,\lambda}(0)$. Finally, for $\lambda_s = \lambda$, the above test can be inverted to obtain a confidence interval for $\beta_j$ valid conditionally on $\{\hat \beta_j^\lambda \neq 0\}$. Two particularly interesting choices for such a $100(1-\alpha)\%$ confidence interval are
\[
     \hat{ \underline{C}}_{j}^{\lambda}
    = \mathrm{conv}\left(\left\{\gamma: \underline{p}_j^{\lambda,\lambda}(\gamma)> \alpha\right\}\right)\textrm{ and } \hat{ \underline{C}}_{j}^{*\lambda}
    =\mathrm{conv}\left( \left\{\gamma: \underline{p}_j^{\hat \lambda_\ell(\gamma),\lambda}(\gamma)> \alpha\right\}\right),
\]
which use $\lambda_\ell=\lambda_s=\lambda$ and a cross-validated choice for $\lambda_{\ell}$ (see Section~\ref{sec:lambda_choice_singletest}), respectively. Note that our cross-validation strategy does not allow for a data-adaptive choice for the selection $\lambda$ (i.e., $\lambda_s$), as for the individual $\ell$-test $p$-values $\underline{p}_j^{\lambda_\ell,\lambda_s}(\gamma)$ to be valid, $\lambda_s$ needs to be a function of $\bm S^{(j)}(\gamma)$ (and maybe some external, conditionally independent, sources of randomness), and this needs to hold for all $\gamma \in \mathbb R$. 

Like \citet{liu2018powerful} but unlike, e.g., \citet{lassoinference}, our conditional inferences do not condition on anything about the LASSO's selection except that it selects the $j^{\mathrm{th}}$ coefficient. As mentioned in Section~\ref{sec:relatedwork}, the key difference between our conditional inference and \citet{liu2018powerful}'s is essentially the same as the difference between the (unconditional) $\ell$-test and the $t$-test, and indeed in Section~\ref{sec:expt_lasso_adjusted_ci} we find that our conditional confidence intervals improve over those of \citet{liu2018powerful} similarly to how the unconditional $\ell$-test confidence intervals from Section~\ref{sec:ciforbeta} improve over standard $t$-test confidence intervals. \color{black}Hence, we can understand when and how the conditional $\ell$-test outperforms \cite{liu2018powerful}'s method in the same way that we understand the analogous outperformance of the $\ell$-test over the $t$-test; see Sections~\ref{sec:improved_performance} and \ref{sec:power_analysis} for intuitive and formal details on this outperformance, respectively.\color{black}

\section{Experiments}
\label{sec:experiments}

In this section, we perform experiments to evaluate the performance of the $\ell$-test and its corresponding confidence intervals and post-selection procedures. For all simulations except in Section~\ref{sec:robust} (where we study the robustness of the $\ell$-test to deviations from model~\eqref{eqn:lm}), we use a linear model~\eqref{eqn:lm} with $k$ out of the $d$ elements of $\bm \beta$ chosen uniformly without replacement and set to $A$ or $-A$ with equal probability, and all the other remaining entries set to 0. We perform inference on one randomly chosen signal coefficient, $\beta_j$. The rows of $\bm X$ are drawn i.i.d. from $\mathcal{N}_d(\bm 0, \bm \Sigma)$ and then the columns are normalized. We will specify the values of $n,d,k,A,\bm \Sigma$, and $\sigma$ for each of the simulation settings we consider.

\subsection{Power of the $\ell$-test}
\label{sec:ltest_expt}
We compare the power of the following three tests: The $\ell$-test, the two-sided $t$-test, and the one-sided $t$-test in the direction of the true sign of $\beta_j$, under simulation settings studying the effect of varying the amplitude of the signal variables, the number of signal variables, and the inter-variable correlation. Note there is no need to compare Type I error rates, since all three methods have guaranteed exactly nominal Type I error (as long as model~\eqref{eqn:lm} is well-specified, which it is in this subsection). The results are reported in Figure \ref{fig:unconditional_test} (see the caption for further details of the simulations). A simulation considering the case where $d$ is closer to $n$ is provided in Appendix \ref{sec:app_ltest_further_expt}; the agreement between the $\ell$-test and the one-sided $t$-test becomes even stronger in this case.

\begin{figure}[h]
    \centering
    \includegraphics{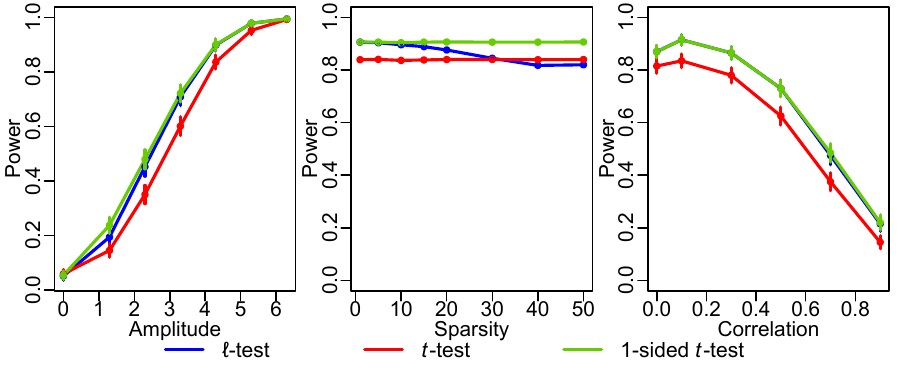}
    \caption{Power comparison of size-5$\%$ tests for $H_{j}$. For all the settings, we fix $n=100,d=50, \sigma =1$. For left, we fix $k=5,\bm \Sigma = \bm I$ and vary the amplitude $A$. For center, we fix $A=4.3, \bm \Sigma = \bm I$ and vary the sparsity level $k$. For right, we fix $A=4.3, k = 5, \Sigma_{ij} = \rho^{|i-j|}$ and vary the inter-variable correlation $\rho$. The error bars represent plus or minus two standard errors. }
    \label{fig:unconditional_test}
\end{figure}

The $\ell$-test significantly outperforms the $t$-test in sparse settings for any signal amplitude, achieving essentially one-sided $t$-test power for moderate-to-high signal amplitudes. The $\ell$-test's power remains close to that of the one-sided $t$-test when as many as 30$\%$ of the coefficients are signals, and this phenomenon seems to be similar across covariate correlation levels. Furthermore, the $\ell$-test only starts to under-perform the $t$-test after more than 60\% of the coefficients are non-zero (as one might expect, given the $\ell$-test is designed to leverage sparsity), and even when the signal is fully dense (100\% nonzero entries, all with equal magnitude), the $\ell$-test's power loss is still only a fraction of its power gain in sparse settings.

\subsection{Robustness of the $\ell$-test}
\label{sec:robust}
\begin{figure}[h]
    \centering \includegraphics{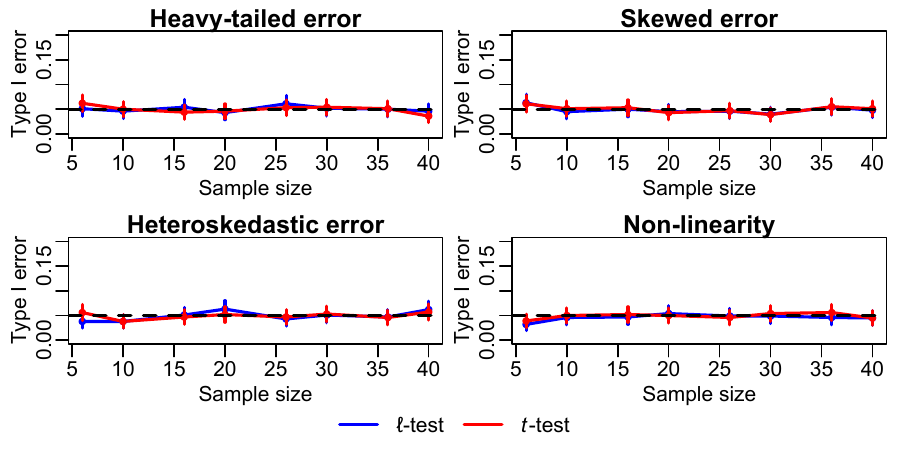}
    \caption{We fix $\rho = 0$ and vary the sample size ($n$) on the x-axis and set $d=n/2$. For the top-left, we draw $\epsilon_i\iid t_2$ (which does not have a finite second moment), for the top-right, $\epsilon_i \iid \mathrm{Exp}(1)$, standardized with its theoretical mean and standard deviation, and for the bottom-left, $\epsilon_i\iid N(0,1)$ \color{black}if mean of the $i^{\mathrm{th}}$ row of $\bm X$ is less than 0, while $\epsilon_i \iid N(0,8)$ otherwise\color{black}. For bottom-right, we generate $y_i\sim \mathcal N\left((\bm X_i^4)^T\bm \beta, \bm I\right)$, where, $(\bm X_i^4)_j = X_{ij}^4$. All these settings use the nominal size of $5\%$ and the error bars represent plus or minus two units of standard error.}
    \label{fig:robustness_main_paper}
\end{figure}

One thing that makes the $t$-test remarkable and so useful in practice is its robustness, even in relatively small samples, to violations of model~\eqref{eqn:lm}. \color{black}To evaluate the $\ell$-test's robustness, we fix $d=n/2,k=1,A=3.3$, $\bm \Sigma$ to have a Toeplitz structure with $\Sigma_{i,j} = \rho^{|i-j|}$ \color{black} and test on a null index $\beta_j=0$. Figure~\ref{fig:robustness_main_paper} shows Type I error results for the $\ell$-test and $t$-test for four types of model violation, when we choose $\bm \Sigma = \bm I$ (that is, $\rho = 0$): heavy-tailed errors, skewed errors, heteroskedastic errors, and model non-linearity (the figure caption gives exact specifications of each of the model violations) for a range of small sample sizes. In Appendix~\ref{sec:app_robustness} we present further simulations of these same four types of model violations, but Figure~\ref{fig:robustness_main_paper} shows the most extreme example of each of the four. \color{black}Appendix~\ref{sec:app_robustness} also presents the counterparts of these settings that includes correlation between the variables by setting $\rho = 0.5$.
Despite substantial deviations from \eqref{eqn:lm}, the $\ell$-test, like the $t$-test, remains quite robust---their Type I error violations remain close and both are at most just a few percentage points above nominal.\color{black}

\subsection{Confidence Intervals for $\beta_j$}
\label{sec:expt_ci_unconditional}

\begin{figure}[h]
    \centering
    \includegraphics{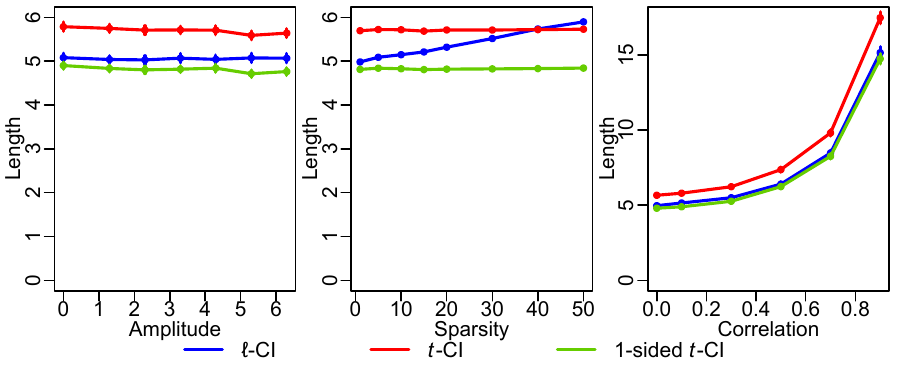}
    \caption{Coverage and lengths of the 95$\%$ confidence intervals. Figures on the left, center and right are under the exact same experimental settings as are the respective figures in Figure \ref{fig:unconditional_test}. The error bars represent plus or minus two standard errors.}
    \label{fig:unadjusted_CI}
\end{figure}

We now perform simulations to compare the $\ell$-test confidence interval with the usual (two-sided) $t$-test confidence interval, as well as the interval obtained by inverting the one-sided $t$-test in the direction of the true sign of the alternative for every alternative $\beta_j=\gamma$. In Appendix \ref{sec:app_1-sided-interval}, we discuss an explicit characterization of this last interval, however a point we highlight here is that this is an oracle procedure (even more so than the one-sided $t$-test we compare to in Section~\ref{sec:ltest_expt}) and can only be constructed if we know the \emph{exact} true value of the coefficient, not just its sign. To obtain the $\ell$-test confidence interval, we choose a grid of candidate values for $\beta_j$ and report the coverage and length of the smallest interval that strictly encloses the rejected values of $\beta_j$ from both the ends. Note that this interval will always have length and coverage at least as large as the true $\ell$-test confidence interval (which we can never obtain exactly from a finite grid of $\beta_j$ values). We use brute force calculations using the highly optimized functions available in the package \texttt{glmnet} \citep{glmnet1} in \textbf{\textsf{R}} instead of our theoretically efficient algorithm in Appendix \ref{sec:app_beta_path} for $\ell$-test inversion (and defer the task of designing an optimized implementation for it to future research).

For our simulations we consider the exact settings as in Figure \ref{fig:unconditional_test} and summarize the results in Figure \ref{fig:unadjusted_CI}. We only report the lengths of all the intervals for a more compact presentation whereas the full results with empirical converges are reported in Appendix \ref{sec:app_ci_with_1se} (the coverage is always extremely close to the nominal 95\%).
As expected, we see similar trends as in Figure \ref{fig:unconditional_test}, with the $\ell$-test confidence intervals being close to the oracle one-sided $t$-test intervals, consistently across amplitudes, in sparse settings. In the left and right plots, the $\ell$-test confidence intervals are consistently about $12\%$ shorter than their $t$-test based counterparts. Perhaps surprisingly, the center plot shows that the $\ell$-test's benefit in confidence interval width over the $t$-test remains up until about 80\% nonzero entries in the coefficient vector, which is a larger outperformance range of sparsity than in Figure~\ref{fig:unconditional_test}. Similar to the power results, we see only a small detriment to using the $\ell$-test confidence interval in the densest setting, relative to its benefit in the sparsest setting. As with the power simulations, we also consider a setting with $d$ closer to $n$ in Appendix \ref{sec:app_ci_with_1se}, and again find this further narrows the gap between the $\ell$-test and one-sided $t$-test procedures.

\subsection{Post-selection $\ell$-test inference}
\label{sec:expt_lasso_adjusted_ci}
\begin{figure}[h]
    \centering
    \includegraphics{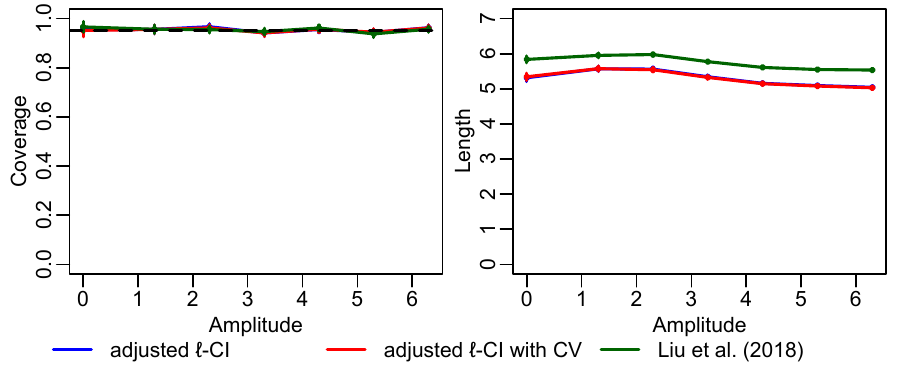}
    \caption{Length and coverage of various post-selection $95\%$ confidence intervals. We are under exactly the same setting as the left panel of Figure \ref{fig:unadjusted_CI}. Here `adjusted $\ell$-CI' and `adjusted $\ell$-CI with CV' refer to $\underline{\hat C}_j^\lambda$ and $\underline{\hat C}_j^{*\lambda}$, respectively, and we have used $\lambda = 0.01$ (which approximately matched the average cross-validated value found on independent data sets, for the entire range of amplitudes). The error bars represent plus or minus two standard errors. }
    \label{fig:setting_1_adjusted}
\end{figure}
In this section, we perform simulations to empirically evaluate the performance of the adjusted $\ell$-test confidence intervals in Section \ref{sec:postselection} for post-selection inference in the linear model under LASSO selection. In particular, we will compare $\underline{\hat C}_j^\lambda$ and $\underline{\hat C}_j^{*\lambda}$, where for both of these methods we invert the respective adjusted $\ell$-tests on a grid of values using the same strategy as in Section \ref{sec:expt_ci_unconditional}, along with the conditional confidence interval procedure of \citet{liu2018powerful}. Figure \ref{fig:setting_1_adjusted} shows the effect of varying coefficient amplitude on the length and coverage of the intervals (conditional on the LASSO selecting the coefficient) for the same setting as the left panels of Figure~\ref{fig:unconditional_test}. We see there is practically no difference between the performance of $\underline{\hat C}_j^\lambda$ and $\underline{\hat C}_j^{*\lambda}$ (see figure caption for how $\lambda$ was chosen) and both are consistently shorter than the method in \citet{liu2018powerful}, reflecting again the benefits of the $\ell$-test under sparsity but now conditional on LASSO selection. In Appendix \ref{sec:app_expt_lasso_adjusted_ci}, we show results for an additional setting where $d$ is closer to $n$. 

\subsection{Analysis of the HIV drug resistance data}
\label{sec:hiv}

The HIV drug resistance data \citep{hivdata} consists of 16 different regressions, each containing data on a set of genetic mutations (the covariates) and a score measuring resistance to an HIV drug (the response). We follow the same pre-processing step suggested in \citet{knockoff}, resulting in 16 regressions with $n$ ranging between 328 and 842 and $d$ ranging between 147 and 313. Running the $t$-test on all covariates across all data sets results in an average \textcolor{black}{of 16.7\% of the covariates being discovered}, while for the $\ell$-test, this figure is 18.6\% (this represents an 11\% improvement), {\color{black}suggesting} that the $\ell$-test's theoretical benefits under sparsity provide genuine power gains in real data sets. Similarly, the $t$-test confidence intervals have an average width of 3.85 while the $\ell$-test confidence intervals' is 3.54, representing about an 8\% improvement. We also note that previous works studying the same data have aggregated discovered mutations to the gene level, and if we do this, the comparison remains similar: the $t$-test's average power is 31.8\% discovered genes while the $\ell$-test's is 36.4\% (a 14\% improvement).

\section{Discussion}
\label{sec:discussions}

The $\ell$-test leverages sparsity by using the LASSO coefficient estimate $|\hat \beta^\lambda_j|$ as its test statistic and can achieve power close to that of the one-sided $t$-test without any knowledge about the true sign of the coefficient. 
The $\ell$-test can be inverted to obtain confidence intervals that are about $10\%$ shorter than the $t$-test intervals {\color{black} in the sparse settings of our simulations}, 
and the $\ell$-test and confidence intervals can also be adjusted for LASSO selection. {\color{black}The $\ell$-test can also be applied asymptotically to a rich class of models far beyond the homoskedastic Gaussian linear model.}

\color{black}
The key idea underlying the $\ell$-test is that a belief about model structure can be leveraged via information in the null model's sufficient statistic to improve (when that structure holds) inferential power. We believe this idea can be used for types of structure other than sparsity as well, e.g., 
\color{black} if $\bm\beta$ is dense but smooth (i.e., it has small total variation), a test based on the fused LASSO \citep{tibshirani2005sparsity} may be more appropriate. {\color{black}However, while the key idea of constructing a valid test based on a structure-leveraging test statistic as described at the beginning of Section~\ref{sec:elltest} can always be carried out via co-sufficient sampling \citep{CSS}, the details in Sections~\ref{sec:elldist}--\ref{sec:improved_performance} of what make the $\ell$-test computationally efficient and powerful rely heavily on the particulars of the LASSO, and hence we defer to future work the possibility of similar results for other types of structure.}

\section*{Acknowledgements}
The authors would like to thank Danielle Paulson and Jonathan Taylor for valuable discussions and comments. LJ and SS were partially supported by DMS-2045981.

\bibliography{refs}

\appendix

\section{Characterization of the $\ell$-distribution under $H_j(\gamma):\beta_j = \gamma$}
\label{sec:app_ldist}

In Section \ref{sec:elldist}, we characterized the conditional distribution, $\hat \beta_j^\lambda\mid \bm S^{(j)}$ under $H_j$ based on the quantiles of $u_1$. In this section, we extend the result to provide a similar characterization of the $\ell$-distribution under $H_j(\gamma):\beta_j = \gamma$. 
\begin{theorem}
\label{thm:general_distexpression}
    Consider the linear model defined in Theorem \ref{thm:distexpression}, fix $\gamma \in \mathbb R$ and define $\hat \sigma_j(\gamma) = \|(\bm I - \bm P_{-j})(\bm y - \gamma \bm X_j)\|$. Then,
    \begin{enumerate}[label = (\alph*)]
        \item $\bm S^{(j)}(\gamma) = (\bm X_{-j}^T\bm y, \|(\bm I - \bm P_{-j})(\bm y - \gamma \bm X_j)\|)$ is sufficient under $H_j(\gamma)$. Furthermore, fixing an orthogonal matrix $\bm V$ for the column-space of $\bm X_{-j}$ as in Lemma \ref{lem:conddist}, there exists a unique vector $\bm u^\gamma \in \mathbb S^{n-d}$ such that $\bm y = \hat{\bm y}_j + \gamma (\bm I - \bm P_{-j})\bm X_j + \hat \sigma_j(\gamma)\bm V \bm u^\gamma$ and
        \[
        \bm u^\gamma \mid \bm S^{(j)}(\gamma) \stackrel{H_j(\gamma)}{\sim}\mathrm{Unif}\left(\mathbb S^{n-d}\right).
        \]
        \item Furthermore, analogously define $\Lambda_j(\cdot,\cdot;\gamma):\mathbb R^2\mapsto \mathbb R$ by
    \[
        \Lambda_j(b, \epsilon;\gamma) = \frac{-\bm X_j^T(\hat {\bm y}_j + \gamma (\bm I - \bm P_{-j})\bm X_j -b\bm X_j - \bm X_{-j}\hat {\bm \beta}^\lambda_{-j}(b)) + n\lambda \epsilon }{\hat \sigma_j(\gamma)\|(\bm I - \bm P_{-j})\bm X_j\|}.
    \]
    Then the function $f_{\bm S^{(j)}(\gamma)}^\gamma:\mathbb R\mapsto \mathbb R$, whose inverse is defined as
    \[
        \left(f_{\bm S^{(j)}(\gamma)}^\gamma\right)^{-1}(b) =\begin{cases}
                \Lambda_j(b, \mathrm{sign}(b);\gamma), & b\neq 0\\
                [\Lambda_j(0, -1;\gamma), \Lambda_j(0, 1;\gamma)], & b = 0
            \end{cases},
    \]
    is continuous and increasing in $\mathbb R$ and strictly increasing in $\{u: f^\gamma_{\bm S^{(j)}(\gamma)}(u)\neq 0\}$ and satisfies  $\hat \beta_j^{\lambda} = f^\gamma_{\bm S^{(j)}}(u^\gamma_j)$.
    \end{enumerate}
\end{theorem}
Note that, as discussed in Section \ref{sec:postselection}, the proof of item 1 of the above theorem directly follows from Lemma \ref{lem:conddist} applied to $(\bm y - \gamma \bm X_j, \bm X)$. Analogous to the proof of Theorem \ref{thm:distexpression} using Lemma \ref{lem:conddist}, one can derive item 2 of Theorem \ref{thm:general_distexpression} from its item 1, by defining $\bm z : = \bm y - \gamma\bm X_{j}$ and $\bm \delta$ by $\bm \delta_{-j} = \bm \beta_{-j}$ and $\delta_j = \beta_j - \gamma$ and copying exactly the same proof as in Appendix \ref{sec:app_proof_distexpression} but by replacing $(\bm y, \bm \beta, \hat \sigma_j)$ with $(\bm z, \bm \delta, \hat \sigma_j(\gamma))$ and then substituting back for $\bm \delta$ and $\bm z$ at the end. We thus skip an explicit proof of Theorem \ref{thm:general_distexpression}.

Here we would also like to draw attention to the fact that the function $f_{\bm S^{(j)}(\gamma)}^\gamma$, for any $\gamma \in \mathbb R$, is defined for \textit{any} real number $u$ and not necessarily restricted to $[0,1]$ and also that $\Lambda_j(b,\mathrm{sign}(b))$ can take values outside the $[-1,1]$ range. In fact if the values exceed this range, we can often draw conclusive insights about the behavior of the LASSO estimate $\hat \beta_j^\lambda$. For example, $\Lambda(0,\pm 1)<-1$ implies that $u_1^\gamma>\Lambda(0,1)$ and hence that for any $\bm y$ generated using the condition in Theorem \ref{thm:general_distexpression} for a $\bm u^\gamma \in \mathbb S^{n-d}$, the resultant LASSO estimator of the $j^{\mathrm{th}}$ coefficient, $\hat \beta_j^\lambda$, will always be positive. This can happen in situations as we next described in Section \ref{sec:app_improved_performance} and can in-fact, result in the $\ell$-test producing exactly the one-sided $t$-test $p$-value.

\section{Achieving the power of a one-sided $t$-test}
\label{sec:app_improved_performance}
The conclusions of the previous section show that if $\beta_j>0$ (which, without any loss of generality, we will assume throughout this section), the $\ell$-test would produce the exact $p$-value of a one-sided $t$-test if $\Lambda(0,\pm 1) = v_{\pm}< -1$ (as in this case $\hat \beta_j^\lambda>0$ under the null conditional distribution of $\bm y\mid \bm S^{(j)}$, and hence, the $\ell$-distribution puts all its mass on the positive half and we saw in Section \ref{sec:elltest} that the contribution from this part to the $\ell$-test $p$-value is exactly the one-sided $t$-test $p$-value). In this section, we will try to take a closer look at when this can be the case. Note that we introduced $\hat m_j$ in Section \ref{sec:improved_performance}, defined by,
\[
    \hat m_j = \frac{-\bm X_j^T (\hat {\bm y}_j - \bm X_{-j}\hat{\bm \beta}^\lambda_{-j}(0))}{\hat \sigma_j \|(\bm I - \bm P_{-j})\bm X_j\|},
\]
which is the mid-point of the interval $[v_-, v_+]$, and argued that its numerator is a proxy to a quantity given by,
\[
-\bm X_j^T(\hat {\bm y}_j - \bm X_{-j}\bm \beta_{-j})\sim \mathcal{N}(-\beta_j \bm X_j^T \bm P_{-j}\bm X_j, \sigma^2 \bm X_j^T \bm P_{-j}\bm X_j).
\]
Thus roughly speaking, one can observe $v_-,v_+\leq -1$ if $q_j = \frac{\beta_j \bm X_j^T \bm P_{-j}\bm X_j}{\|(\bm I - \bm P_{-j})\bm X_j\|}$ is highly positive. In particular, the magnitude of $q_j$ intuitively quantifies the reliability of the sign-guess. Notably, larger values of $|q_j|$ indicates higher differences between the mass the $\ell$-distribution assigns in the two halves of the real line (and hence is more asymmetric), thereby implying that the $\ell$-test $p$-value is much different from its two-sided-$t$-test counterpart. Now that we have an understanding of the role that $q_j$ plays, we next describe two situations in which one can observe $v_-,v_+\leq -1$:
\begin{itemize}
    \item \textbf{Strong signal size: }For any fixed design matrix $\bm X$, if the signal size $\beta_j$ is strong enough to make the quantity $q_j$ highly positive, one can expect that $\hat m_j$ would be sufficiently negative and hence, we can expect to see the $p$-value of a one-sided $t$-test. This result is intuitive as with increase in the signal size the variable gets more and more distinguishable. Thus the power of the one-sided $t$-test, the two-sided test and the $\ell$-test, all increase with the power of the latter getting closer to the power of the one-sided $t$-test.

    \item \textbf{High feature correlation: }Note that except for $\beta_j$, the other factor in $q_j$, 
    \[
        \frac{\bm X_j^T\bm P_{-j}\bm X_j}{\|(\bm I - \bm P_{-j})\bm X_j\|} = \frac{\bm X_j^T\bm P_{-j}\bm X_j}{\sqrt{\|\bm X_j\|^2 - \bm X_j^T\bm P_{-j}\bm X_j }},
    \]
    depends on the design. This shows that as $\bm X_j^T\bm P_{-j}\bm X_j$ increases and gets closer to $\bm X_j^T\bm X_j$, this factor starts blowing up and makes $q_j$ more positive highlighting another case when the power of the $\ell$-test can get closer to the power of the one-sided $t$-test. This, on the first glance, might seem non-intuitive because the increase in $\bm X_j^T\bm P_{-j}\bm X_j$ actually implies that the $j^{\mathrm{th}}$ variable gets more correlated with the rest of the variables and hence it should become harder to distinguish its effect. Note that unlike the previous case, in this case the power does not increase and as expected, the larger the quantity $\bm X_j^T\bm P_{-j}\bm X_j$ gets, the more the performance of all the three tests deteriorate. However with increase in $\bm X_j^T\bm P_{-j}\bm X_j$, the quantity $q_j$ increases in magnitude suggesting an increase in the belief about the validity of the sign guess. Put another way, with increase in $\bm X_j^T\bm P_{-j}\bm X_j$, it becomes possible to obtain a more reliable sign-guess as a function of the sufficient statistic, $\bm S^{(j)}$. Thus, though with this increasing correlation the tests loose power, the performances of the $\ell$-test and the one-sided $t$-test gets closer because of the improved sign-guessing ability of the former.
    
    Note that for design matrices, $\bm X$, of dimension $n\times d$ with i.i.d. drawn Gaussian columns (as is the case with most of the simulations in this paper), it indeed holds that with $d$ getting closer to $n$, the component $\bm X_j^T\bm P_{-j}\bm X_j$ increases in magnitude. Thus, in this case we would expect that the power curves of $\ell$-test and the one-sided $t$-test come closer as $d$ gets closer to $n$ (that is, as we move closer to unidentifiability). 
\end{itemize}

Finally, note that it follows as a direct consequence of the discussions in this section and Section \ref{sec:improved_performance} that the $\ell$-test gains no power over the $t$-test if $\bm X_j$ is orthogonal with the rest of the columns (and in particular, for orthogonal designs). In this case, $\hat m_j = 0$, so that we have no estimate of the sign of $\beta_j$ and hence, the $\ell$-distribution is symmetric about 0. With smoothing out of the $p$-value at $1$, we would expect the $\ell$-test and the $t$-test, as well as the respective confidence intervals, to perform similarly.

\color{black}
\section{$\ell$-test when $\sigma$ is known}
\label{sec:known_sigma_ell}

The $\ell$-test can easily be extended to the case with a known error variance $\sigma$, in particular, with the following substitutions over the usual, unknown-$\sigma$-case results: For a generic $(\bm y, \bm X_{n\times d})$---the sufficient statistic under $H_j$ now just shrinks to its first element and, re-using the old notation, is now given by $\bm S^{(j)} = \bm X_{-j}^T\bm y$, while the known-$\sigma$ analogue of Equation~\eqref{eqn:udist} is
\begin{align}
    \label{eqn:zdist}
    \bm y = \bm P_{-j}\bm y + \sigma \bm V\bm z,
\end{align}
for some unique $(n-d+1)$-dimensional vector $\bm z$ with distribution $\mathcal N(\bm 0, \bm I_{n-d+1})$ under $H_j$. A known-$\sigma$ analogue of Theorem~\ref{thm:distexpression} still holds true with $\hat \sigma_j$ replaced by the known $\sigma$ and $u_1$ replaced by $z_1$, and in particular the known-$\sigma$ analogue of Equation~\eqref{eqn:dist_b_nonzero} equates the event $\{\hat \beta_j^{\lambda} \leq b\}$ with $\left\{z_1 \leq \Lambda_j(b, \mathrm{sign}(b))\right\}$, where $\Lambda_j(b, \mathrm{sign}(b))$ is easily seen to be a function only of ${\bm S}^{(j)}$ as required for a straightforward evaluation of the probability of this event via the Gaussian CDF, since the known-$\sigma$ expression of $\Lambda_j(b, \mathrm{sign}(b))$ (defined by the aforementioned known-$\sigma$ analogue of Theorem~\ref{thm:distexpression}) now has $\sigma$, instead of $\hat \sigma_j$ in the denominator. 
The known-$\sigma$ analogues of Sections \ref{sec:smoothing} and \ref{sec:lambda_choice_singletest} similarly replace $\bm u$ and $\tilde{\bm u}$ with $\bm z$ and $\tilde{\bm z}$, respectively, along with replacing their corresponding distributions (in particular, Unif($\mathbb{S}^{n-d}$) is replaced with $\mathcal{N}(\bm 0,\bm I_{n-d+1})$).

In Section~\ref{sec:power_analysis}, we discuss a re-centered test, which, in the unknown-$\sigma$-case, would be based on the test statistic $|u_1-\hat m_j|$, but in the known-$\sigma$ case (which is the case considered in Section~\ref{sec:power_analysis} and Theorem~\ref{thm:power} therein) the re-centered test replaces $u_1$ with $z_1$ and replaces $\hat\sigma_j$ in the expression for $\hat m_j$ with the known $\sigma$. We will denote this known-$\sigma$ analogue of $\hat m_j$ by $\hat m_j'$.
\color{black}

\section{Proofs}

\subsection{Proof of Lemma \ref{lem:conddist}}
\label{sec:proof_conddist}
\begin{proof}
Because $\bm V$ is a full column-rank matrix, there exists a unique $\bm u$, such that,
\[
    \bm V \bm u = \frac{\bm y - \hat{\bm y}_j}{\|\bm y - \hat{\bm y}_j\|} = \frac{\bm y - \hat{\bm y}_j}{\hat \sigma_j}.
\]
Clearly,
\[
    \|\bm V \bm u\| = \frac{\|\bm y - \hat{\bm y}_j\|}{\|\bm y - \hat{\bm y}_j\|} = 1.
\]
Because orthogonal transformations preserve norm, we must have that this unique $\bm u$ satisfy,
\[
    \|\bm u \| = \|\bm V \bm u\| = 1,
\]
which proves \eqref{eqn:conddist_lm} in the statement of Lemma \ref{lem:conddist}. The proof \eqref{eqn:udist} follows directly from \citet[Proposition E.1]{luo2022}.
\end{proof}

\subsection{Proof of Lemma \ref{lem:pvalequivalence}}
\label{sec:proof_pvalequivalence}
\begin{proof}
       Note that we can write the OLS coefficient as (see \citet[Section 4]{fithianthesis}),
        \begin{align}
            \label{eqn:olsequivalent}
            \hat\beta_{j,\mathrm{OLS}} = \frac{\bm X_j^T(\bm I - \bm P_{-j})\bm y}{\|(\bm I - \bm P_{-j})\bm X_j\|^2}.
        \end{align}
        We first start by showing that the OLS estimate, $\hat \beta_{j, \mathrm{OLS}}$, is a constant multiple of the statistic, $u_1$, where the constant is a deterministic function of the sufficient statistic, $\bm S^{(j)}$. Based on Lemma \ref{lem:conddist}, we have the decomposition,
        \begin{align*}
            \bm y = \hat{\bm y_j} + \underbrace{\hat \sigma_j \bm V \bm u}_{\bm \hat e_j}.
        \end{align*}
        This implies,
        \begin{align}
            \label{eqn:u_decomp}
        \begin{aligned}
            \bm X_j^T(\bm I - \bm P_{-j})\bm y &= \bm X_j^T(\bm I - \bm P_{-j}) \bm P_{-j}\bm y + \hat \sigma_j  \bm X_j^T(\bm I - \bm P_{-j}) \bm V \bm u\\
            &= \hat \sigma_j \bm X_j^T \bm V \bm u. \textrm{ [since the columns of }\bm V\textrm{ are orthogonal to the columns of }\bm X_{-j}]
        \end{aligned}
        \end{align}
     Note that from the choice of $\bm V$, we have that $\bm V_1 = \frac{(\bm I - \bm P_{-j})\bm X_j}{\|(\bm I - \bm P_{-j})\bm X_j\|}$, and hence one can find other orthogonal vectors $\bm v_i, i\in \{2,\dots, n-d+1\}$, such that, $\bm V = [\bm v_1, \dots, \bm v_{n-d+1}]$. Then it follows that $\bm X_j^T\bm v_i = 0, i >1$ and hence for $\bm u = (u_1,\dots, u_{n-d+1})^T$,
    \begin{align}
    \label{eqn:u_1}
        \bm X_j^T \bm V \bm u = \bm X_j^T \sum_{i=1}^{n-d+1}\bm v_i u_i = \bm X_j^T\bm V_1u_1 = \frac{\bm X_j^T(\bm I-\bm P_{-j})\bm X_j}{\|(\bm I - \bm P_{-j})\bm X_j\|}u_1 = \|(\bm I - \bm P_{-j})\bm X_j\|u_1,
    \end{align}
    thereby implying that,
    \begin{align}
         \label{eqn:ols_uj}
         \hat\beta_{j,\mathrm{OLS}} = \frac{\bm X_j^T(\bm I - \bm P_{-j})\hat{\bm e}_j}{\|(\bm I - \bm P_{-j})\bm X_j\|^2}. = \frac{\hat \sigma_j u_1}{\|(\bm I - \bm P_{-j})\bm X_j\|}
    \end{align}
    Note that the factor pre-multiplying $u_1$ in the above equation is a function of the sufficient statistic, $\bm S^{(j)}$, and the design matrix, $\bm X$. Thus for samples from $\bm y\mid \bm S^{(j)}$, under $H_{j}$, showing that a statistic is increasing in $u_1$ or $\hat \beta_{j, \mathrm{OLS}}$ are equivalent. 

    Next, we show that the $t$-test statistic is an increasing function of $\hat \beta_{j, \mathrm{OLS}}$. For that, we first decompose the error term $\hat{\bm e}_j$ as $\hat {\bm e}_j = \hat {\bm e}_{\parallel} + \hat {\bm e}_{\perp}$, where,
        \begin{align}
        \label{eqn:perpandparallel}
        \begin{aligned}
            & \hat{\bm e}_{\parallel} = \frac{(\bm I - \bm P_{-j})\bm X_j \bm X_j^T (\bm I - \bm P_{-j})\hat{\bm e}_{j}}{\|(\bm I - \bm P_{-j})\bm X_j\|^2} = (\bm I - \bm P_{-j})\bm X_j \hat\beta_{j,\mathrm{OLS}}\textrm{, and }\\
            &\hat{\bm e}_{\perp} = \left(\bm I - \frac{(\bm I - \bm P_{-j})\bm X_j \bm X_j^T (\bm I - \bm P_{-j})}{\|(\bm I - \bm P_{-j})\bm X_j\|^2}\right)\hat{\bm e}_{j}.
        \end{aligned}
        \end{align}
        Note that $\hat{\bm e}_{\parallel}$ is the component of $\hat{\bm e}_j$ along the component of $\bm X_j$ orthogonal to $\bm X_{-j}$, while the component $\hat{\bm e}_{\perp}$ is perpendicular to the columnspace of the entire $\bm X$. Furthermore, the components,  $\hat{\bm e}_{\perp}$ and $\hat{\bm e}_{\parallel}$ are themselves orthogonal to each other. Using the relation $\hat{\bm e}_j = \hat{\sigma}_{j}\bm V \bm u$ and the fact that the matrix pre-multiplying  $\hat{\bm e}_j$ to obtain $\hat{\bm e}_{\perp}$ is idempotent, one can write,
        \begin{align*}
            \|\hat{\bm e}_{\perp}\|^2 &= \hat\sigma_j^2 \bm u^T \bm V^T\left(\bm I - \frac{(\bm I - \bm P_{-j})\bm X_j \bm X_j^T (\bm I - \bm P_{-j})}{\|(\bm I - \bm P_{-j})\bm X_j\|^2}\right) \bm V \bm u\\
            &= \hat\sigma_j^2 \bm u^T \bm V^T\bm V \bm u - \frac{\left(\bm X_j^T \hat{\bm e}_j\right)^2}{\|(\bm I - \bm P_{-j})\bm X_j\|^2}\\
            &= \hat \sigma_j^2- \frac{\left(\bm X_j^T \hat{\bm e}_j\right)^2}{\|(\bm I - \bm P_{-j})\bm X_j\|^2} \textrm{ }\left[\because\bm V\textrm{ is an orthogonal matrix and }\bm u^T\bm u = 1\right]\\
            &= \hat \sigma_j^2- \left(\hat\beta_{j,\mathrm{OLS}}\right)^2\|(\bm I - \bm P_{-j})\bm X_j\|^2
        \end{align*}
         Because $\hat{\bm e}_{\perp}$ is in the orthogonal complement of the columnspace of $\bm X$, we can write,
        \[
            (\bm I - \bm P)\hat{\bm e}_j = \hat{\bm e}_{\perp} + (\bm I - \bm P)\hat{\bm e}_{\parallel}.
        \]
        Also,
        \begin{align*}
            \|(\bm I - \bm P)\hat{\bm e}_j\|^2 &= \|\hat{\bm e}_{\perp}\|^2 + \|(\bm I - \bm P)\hat{\bm e}_{\parallel}\|^2\\
            &=  \hat \sigma_j^2- \left(\hat\beta_{j,\mathrm{OLS}}\right)^2\|(\bm I - \bm P_{-j})\bm X_j\|^2 + \|(\bm I-\bm P)(\bm I - \bm P_{-j})\bm X_j\|^2(\hat\beta_{j,\mathrm{OLS}})^2\\
            &=  \hat \sigma_j^2 - (\hat\beta_{j,\mathrm{OLS}})^2\left(\|(\bm I - \bm P_{-j})\bm X_j\|^2 - \|(\bm I-\bm P)(\bm I - \bm P_{-j})\bm X_j\|^2\right)\\
            &= \hat \sigma_j^2 - (\hat\beta_{j,\mathrm{OLS}})^2\underbrace{\|\bm P(\bm I - \bm P_{-j})\bm X_j\|^2}_{=:\, \kappa}.
        \end{align*}
        With these expressions, we can write,
        \[
            \Tstat = \frac{\hat \beta_{j,OLS}}{\hat \sigma \sqrt{(\bm X^T \bm X)^{-1}_{j,j}}} = C \frac{\hat\beta_{j,\mathrm{OLS}}}{\sqrt{\hat\sigma_j^2 - \left(\hat\beta_{j,\mathrm{OLS}}\right)^2\kappa}} = C \frac{\frac{\hat \sigma_j u_1}{\|(\bm I - \bm P_{-j})\bm X_j\|}}{\sqrt{\hat\sigma_j^2 - u_1^2 \left(\frac{\hat \sigma_j}{\|(\bm I - \bm P_{-j})\bm X_j\|}\right)^2\kappa}},
        \]
        where, $C = \sqrt{\frac{n-d}{(\bm X^T\bm X)^{-1}_{j,j}}}$ and $\kappa$ are both positive. Defining, $C' = C\cdot \frac{\hat \sigma_j}{\|(\bm I - \bm P_{-j})\bm X_j\|}$, $\kappa' = \kappa \cdot \left(\frac{\hat \sigma_j}{\|(\bm I - \bm P_{-j})\bm X_j\|}\right)^2$ and $g_{\bm S^{(j)}}(u) = \frac{C' u}{\sqrt{\hat \sigma_j^2 - u^2\kappa'}}$, the above equation shows $T_j = g_{\bm S^{(j)}}(u_1)$, thereby establishing that, $g_{\bm S^{(j)}}$ is a functional of $\bm S^{(j)}$ and itself is continuous, strictly increasing and anti-symmetric.
\end{proof}

\subsection{Proof of Theorem \ref{thm:distexpression}}
\label{sec:app_proof_distexpression}
\color{black}
We first start with an auxiliary lemma that states conditions to establish the continuity (in its remaining arguments) of the minimizer of a marginal of a function. In order to do so, we first recall the concept of upper hemicontinuity of a set-valued function.
\begin{defn}[Upper hemicontinutiy]
    Consider topological spaces $X$ and $Y$ and consider a set-valued function $\phi:X\mapsto 2^Y$. We call $\phi$ upper hemicontinuous at a point $x\in X$, if for any open set $U\subseteq Y$ such that $\phi(x)\subseteq U$, there exists a neighborhood $V$ of $x$ such that $w\in V\implies \phi(w)\subseteq U$.
\end{defn}
It is easy to see that if the set-valued function $\phi$ is actually a function (that is, only singleton sets belong to its range), then hemicontinuity of $\phi$ at $x$ actually implies the continuity of $\phi$ at $x$ as a function, as the definition implies that pre-image under $\phi$ of every open set is open. Now we state the following lemma that will be used multiple times in this, and some of the following proofs.

\begin{lemma}
    \label{lem:minima_cont}
    Consider a real-valued function $\phi:\mathcal X\times \mathcal Y\mapsto \mathbb R_{\geq 0}$ (with $\mathcal X, \mathcal Y$ subsets of Euclidean spaces) and $\bm y^*\in \mathcal Y$ satisfying the following properties:
    \begin{enumerate}
        \item $(\bm x,\bm y)\mapsto \phi(\bm x, \bm y)$ is continuous.
        \item For all $d>0$, there exists $M:=M(d)>0$ such that for $\epsilon = \underset{\bm y\in \bar B(\bm y^*, d)}{\sup}\phi(\bm 0, \bm y)$ (this supremum exists as $\bm y\mapsto \phi(\bm 0, \bm y)$ is a continuous transformation and $\bar B(\bm y^*, d)=\{\bm y: \|\bm y - \bm y^*\|\leq d\}$ is a closed ball), we have that $\|\bm x\|>M\implies \underset{\bm y \in \bar B(\bm y^*, d)}{\inf}\phi(\bm x, \bm y)>\epsilon$.
    \end{enumerate}
    Then the map $\bm y\mapsto \underset{\bm x}{\min}~\phi(\bm x, \bm y)$ is continuous at $\bm y^*$ and the map $\bm y\mapsto \underset{\bm x}{\arg\min}~\phi(\bm x, \bm y)$ is upper hemicontinuous (where $\arg\min$ denotes the set of possible minimizers) at $\bm y^*$. Hence, if the minimizer is unique for values of $\bm y$ in a neighborhood of $\bm y^*$, then the map is continuous at $\bm y^*$.
\end{lemma}

\begin{proof}
    Fix any $\bm y^*\in \mathcal Y$ and let $G\subset \mathcal Y$ be a compact set containing $\bm y^*$. Due to the Hiene--Borel theorem \citep[Theorem 2.41]{rudin1976principles}, $G$ is closed and bounded and hence, there is some closed ball $\bar B(\bm y^*, d)$ of radius $d>0$ centerd at $\bm y^*$ such that $G\subset \bar B(\bm y^*, d)$. Note that due to Property 2 above we have that with $M = M(d)$, $\|\bm x\|>M\implies \phi(\bm x, \bm y)> \underset{\bm y\in \bar B(\bm y^*, d)}{\sup}\phi(\bm 0, \bm y)$ for all $\bm y \in \bar B(\bm y^*, d)$. This implies that for any $\bm y \in \bar B(\bm y^*, d)$, $\phi(\bm 0, \bm y)<\phi(\bm x, \bm y)$ for any $\|\bm x\|>M$. This implies that for any $\bm y \in \bar B(\bm y^*, d)$,
    \[
        \underset{\bm x}{\arg\min}~\phi(\bm x, \bm y) = \underset{\bm x: \|\bm x\|\leq M}{\arg\min}~\phi(\bm x, \bm y)
    \]
    Now define the set-valued function:
    \[
        C(\bm y) = \{\bm x: \|\bm x\|\leq M\},
    \]
    and note that $C(\bm y)$ is compact and $\bm y\mapsto C(\bm y)$ is a constant function. We can then apply Berge's Maximum Theorem \citep{berge} to conclude that the map $\bm y\mapsto \underset{\bm x\in C(\bm y)}{\arg\min}~\phi(\bm x, \bm y)$ is upper hemicontinuous at $\bm y^*$, and since that map is equal to the map $\underset{\bm x}{\arg\min}~\phi(\bm x, \bm y)$ in a neighborhood of $\bm y^*$, the latter map is also upper hemicontinuous at $\bm y^*$. Similarly, Berge's Maximum Theorem implies the map $\bm y\mapsto \underset{\bm x\in C(\bm y)}{\min}~\phi(\bm x, \bm y)$ is continuous at $\bm y^*$, and since that map is equal to the map $\underset{\bm x}{\min}~\phi(\bm x, \bm y)$ in a neighborhood of $\bm y^*$, the latter map is also continuous at $\bm y^*$.
\end{proof}
\color{black}

Next, we first prove that $f_{\bm S^{(j)}}$ is continuous and strictly increasing in $\{u:f_{\bm S^{(j)}}(u)\neq 0\}$ (in \ref{sec:proof_distexpression_1}), followed by a proof of the characterization of its inverse (in \ref{sec:proof_distexpression_2}).
\subsubsection{Proof of continuity and increasing properties of $f_{\mathbf{S}_j}$}
\label{sec:proof_distexpression_1}
In this section, we will prove that $f_{\bm S^{(j)}}$ is continuous and strictly increasing in $\{u:f_{\bm S^{(j)}}(u)\neq 0\}$.
\begin{proof}
    We again start with the decomposition,
        \begin{align*}
            \bm y = \hat{\bm y}_j + \underbrace{\hat \sigma_j \bm V \bm u}_{\hat{\bm  e}_j},
        \end{align*}
        and as we did in the proof in Section \ref{sec:proof_pvalequivalence}, we again decompose,
        \[
            \hat{\bm e}_j = \hat{\bm e}_{||} + \hat{\bm e}_{\perp}.
        \]
        Next using Equation \eqref{eqn:olsequivalent} and \eqref{eqn:ols_uj},
        \[
            u_1 = \frac{\bm X_j^T(\bm I - \bm P_{-j})\bm y}{\hat \sigma_j \|(\bm I - \bm P_{-j})\bm X_j\|} = \frac{\bm X_j^T(\bm I - \bm P_{-j})\hat{\bm e}_j}{\hat \sigma_j \|(\bm I - \bm P_{-j})\bm X_j\|}.
        \]
        Also from Equation \eqref{eqn:perpandparallel},
        \[
            \hat{\bm e}_{\parallel} = \frac{(\bm I - \bm P_{-j})\bm X_j \bm X_j^T (\bm I - \bm P_{-j})\hat{\bm e}_{j}}{\|(\bm I - \bm P_{-j})\bm X_j\|^2}= \frac{\hat{\sigma}_j (\bm I - \bm P_{-j})\bm X_j u_1}{\|(\bm I - \bm P_{-j})\bm X_j\|}.
        \]
    Based on these relations, we have,
        \begin{align}
        \begin{aligned}
            &\|\bm y - \bm X\bm \beta\|^2 + 2n\lambda|\beta_j| + 2n\lambda \sum_{i\neq j}|\beta_i| \\
            &= \|\hat{\bm y}_j + \hat{\bm e}_j - \bm X_{-j}\bm \beta_{-j} - \bm X_j\beta_j\|^2 + 2n\lambda|\beta_j| + 2n\lambda \sum_{i\neq j}|\beta_i| \\
            &= \|\hat{\bm y}_j - \bm X_{-j}\bm \beta_{-j} - \bm P_{-j} \bm X_j\beta_j + \hat{\bm e}_j - (\bm I - \bm P_{-j})\bm X_j \beta_j\|^2  + 2n\lambda|\beta_j| + 2n\lambda \sum_{i\neq j}|\beta_i| \\
            &= \|\hat{\bm y}_j - \bm X_{-j}\bm \beta_{-j} - \bm P_{-j} \bm X_j\beta_j + \hat{\bm e}_{\perp} + \hat{\bm e}_{\parallel} - (\bm I - \bm P_{-j})\bm X_j \beta_j\|^2  \\
            &+ 2n\lambda|\beta_j| + 2n\lambda \sum_{i\neq j}|\beta_i| \textrm{ [}\hat{\bm e}_{\perp}, \hat{\bm e}_{\parallel}\textrm{ are defined above]}\\
            &= \|\hat{\bm y}_j - \bm X_{-j}\bm \beta_{-j} - \bm P_{-j}\bm X_j \beta_j\|^2 + \|\hat{\bm e}_{\perp}\|^2 + \|\hat{\bm e}_{\parallel} - (\bm I - \bm P_{-j})\bm X_j \beta_j\|^2  + 2n\lambda|\beta_j| + 2n\lambda \sum_{i\neq j}|\beta_i|\\
            &= \|\hat{\bm y}_j - \bm X_{-j}\bm \beta_{-j} - \bm P_{-j}\bm X_j \beta_j\|^2 + \|\hat{\bm e}_{\perp}\|^2 + \left(\frac{u_1 \hat \sigma_j}{\|(\bm I - \bm P_{-j})\bm X_j\|} - \beta_j\right)^2\|(\bm I - \bm P_{-j})\bm X_j\|^2  \\
            &+ 2n\lambda|\beta_j| + 2n\lambda \sum_{i\neq j}|\beta_i| \\
            &= \tilde{f}( \bm \beta; \bm y, \bm X)  + \left(\frac{u_1 \hat \sigma_j}{\|(\bm I - \bm P_{-j})\bm X_j\|} - \beta_j\right)^2\|(\bm I - \bm P_{-j})\bm X_j\|^2,
        \end{aligned}
        \end{align}
        $\tilde f(\bm \beta; \bm y, \bm X)$ denotes the expression it is replacing, and whenever the context is clear, we will use $\tilde f$ to denote $\tilde f(\bm \beta;\bm y, \bm X)$. Note that $\tilde f$ also depends on $\lambda>0$, but for compactness, we have suppressed this in the notation as in the following lines we will not be interested in the behavior of $\tilde f$ as a function of $\lambda$. In fact, we will only analyze $\tilde f$ as a function of the argument $\bm \beta$. Thus we have,
        \begin{align*}
            &\underset{\bm \beta}{\arg\min}\left(\|\bm y - \bm X\bm \beta\|^2 + 2n\lambda|\beta_j| + 2n\lambda \sum_{i\neq j}|\beta_i|\right)\\
            &= \underset{\bm \beta}{\arg\min}\left(\tilde{f} + \left(\frac{\hat \sigma_j u_1}{\|(\bm I - \bm P_{-j})\bm X_j\|} - \beta_j\right)^2\|(\bm I - \bm P_{-j})\bm X_j\|^2\right).
        \end{align*}
        Define 
        \begin{align}
        \label{eqn:hatbeta_a}
            \hat{\bm{\beta}}(a) := \underset{\bm \beta}{\arg\min}\left(\underbrace{\tilde{f} + \left(a - \beta_j\right)^2\|(\bm I - \bm P_{-j})\bm X_j\|^2}_{=:U( a,\bm \beta; \bm y, \bm X)}\right).
        \end{align}
        Thus, $\hat{\beta}_j\left(\frac{\hat\beta_{j,\mathrm{OLS}}}{\|(\bm I - \bm P_{-j})\bm X_j\|^2}\right)$ is just the LASSO estimate of $\beta_j$, $\hat \beta^{\lambda}_j$ and hence, it suffices to show that $\hat \beta^{\lambda}_j(a)$ is a non-decreasing function of $a$. As is the case with the function, $\tilde f$, we will also be primarily be interested in the behavior of $U$ as a function of the arguments, $(a,\bm \beta)$. We start with listing some properties of the function, $U$. 
        
        First, note that $U$ is a continuous, convex function in its arguments and is strictly convex in $a$ for any fixed value of $\bm \beta$. \color{black}Next, we will show that $\hat{\bm \beta}(a)$ is continuous in $a$ by invoking Lemma \ref{lem:minima_cont}. In order to do so, observe that $(a,\bm \beta)\mapsto U(a, \bm \beta)$ is a continuous function satisfying item 1 of the requirements. Next, observe that $U(a, \bm \beta)\geq \tilde f(\bm \beta)\geq 2n\lambda \|\bm \beta\|_1$ for any $a$ and $\bm \beta$, and note that $2n\lambda \|\bm \beta\|_1$ does not depend on $a$. Thus, noting that $\lambda>0$, for any $d$ and $a^*$, we have that $\inf_{|a-a^*|\leq d}U(a,\bm \beta)\geq 2n\lambda \|\bm \beta\|_1\to \infty$, as $\|\bm \beta\|\to \infty$. This satisfies item 2, and we can now invoke Lemma \ref{lem:minima_cont} to conclude the continuity of $a\mapsto \hat{\bm \beta}(a)$.\color{black}
        
        Next, note that the only non-differentiable component in the expression of $U$ is the $\ell_1$-penalty of $\bm \beta$, which implies that the $U$ has a partial derivative in $\beta_j$ at all non-zero values of $\beta_j$. Consider a value of $a$ such that $\hat\beta_j(a)\neq 0$ and let $\mathcal A(a) = \{ i\in  [1:k]\setminus \{j\} : \hat \beta_i(a)\neq 0\}$ denote the active set among the remaining variables. 
        
        Note that because $\hat {\bm \beta}(a)$ minimizes $U$, and $\forall i\notin \mathcal A(a)\cup \{j\}$, $\hat {\bm \beta}_j(a) = 0$, it holds that $\hat {\bm \beta}_{\mathcal A(a) \cup \{j\}}(a)$ is a minimizer of $U(a,\bm \gamma; \bm y, \bm X_{\mathcal A(a)\cup \{j\}})$, where now $\bm \gamma$ is a vector of length $|\mathcal A(a)\cup \{j\}|$. For a proof, see Lemma 1 of \citet[Section 3.2]{dcrt}. Because all the entries of $\hat {\bm \beta}_{\mathcal A(a)\cup \{j\}}$ are non-zero, $U(a,\bm \gamma; \bm y, \bm X_{\mathcal A(a)\cup \{j\}})$ is differentiable in $\bm\gamma$ at $\hat {\bm \beta}_{\mathcal A(a)\cup \{j\}}$, and because the latter is a minimizer, an appeal to the first-order stationary conditions yield that for any $i\in \mathcal A(a)$,
        \begin{align*}
            &\frac{\partial }{\partial \gamma_i} U(a,\bm \gamma; \bm y, \bm X_{\mathcal A(a)\cup \{j\}})\big\vert_{\bm \gamma = \hat{\bm \beta}_{\mathcal A(a) \cup \{j\} }}(a)=0\\
            \implies & -\bm X_i^T(\hat{\bm y}_j - \bm X_{\mathcal A(a)}\hat{\bm \beta}_{\mathcal A(a)}(a) - \bm P_{-j}\bm X_j \hat \beta^{\lambda}_j(a)) + 2n\lambda \mathrm{sign}(\hat \beta_i(a)) = 0
        \end{align*}
        Note that from the continuity of $\hat{\bm \beta}(a)$ in the neighborhood of $a$ where $\mathcal A(a)$ does not change, the sign of the active variables also remain constant so that $\mathrm{sign}(\hat \beta_i(a))$ is a constant in that neighborhood, $\forall i \in \mathcal A(a)\cup \{j\}$. Thus, differentiating the above equation both the sides with respect to $a$ yields,
        \begin{align*}
            &\bm X_i^T\bm X_{\mathcal A(a)}\frac{\partial }{\partial a}\hat{\bm \beta}_{\mathcal A(a)} + \bm X_{i}^T\bm P_{-j}\bm X_j \frac{\partial }{\partial a}\hat\beta_j(a)=0, \forall i\in \mathcal A(a)\\
            \implies & \bm X_{\mathcal A(a)}^T\bm X_{\mathcal A(a)}\frac{\partial }{\partial a}\hat{\bm \beta}_{\mathcal A(a)} + \bm X_{\mathcal A(a)}^T\bm P_{-j}\bm X_j \frac{\partial }{\partial a}\hat\beta_j(a)=\bm 0\\
            \implies &\frac{\partial }{\partial a}\hat{\bm \beta}_{\mathcal A(a)} = -\left(\bm X_{\mathcal A(a)}^T\bm X_{\mathcal A(a)}\right)^{-1}\bm X_{\mathcal A(a)}^T \bm P_{-j}\bm X_j \frac{\partial }{\partial a}\hat{\beta}_j(a)
        \end{align*}
        Similarly using the first-order stationary condition on the index $j$ yields,
        \begin{align*}
            &-\bm X_j^T\bm P_{-j}(\hat{\bm y}_j - \bm X_{\mathcal A(a)}\hat{\bm \beta}_{\mathcal A(a)}(a) - \bm P_{-j}\bm X_j \hat \beta^{\lambda}_j(a)) - (a-\hat \beta^{\lambda}_j(a))\|(\bm I - \bm P_{-j})\bm X_j\|^2 =0\\
            \implies & \bm X_j^T \bm X_{\mathcal A(a)}\frac{\partial }{\partial a}\hat{\bm \beta}_{\mathcal A(a)}(a)+\bm X_j^T \bm P_{-j}\bm X_j \frac{\partial }{\partial a}\hat \beta^{\lambda}_j(a)- \left(1- \frac{\partial }{\partial a}\hat \beta^{\lambda}_j(a)\right)\|(\bm I - \bm P_{-j})(\bm X_j)\|^2 = 0\\
            \implies & -\bm X_j^T \bm X_{\mathcal A(a)}\left(\bm X_{\mathcal A(a)}^T\bm X_{\mathcal A(a)}\right)^{-1}\bm X_{\mathcal A(a)}^T \bm P_{-j}\bm X_j \frac{\partial }{\partial a}\hat{\beta}_j(a)+\bm X_j^T\bm X_j \frac{\partial }{\partial a}\hat \beta^{\lambda}_j(a) = \|(\bm I - \bm P_{-j})(\bm X_j)\|^2\\
             \implies & -\bm X_j^T \bm P_{\mathcal A(a)} \bm P_{-j}\bm X_j \frac{\partial }{\partial a}\hat{\beta}_j(a)+\|\bm X_j\|^2 \frac{\partial }{\partial a}\hat \beta^{\lambda}_j(a) = \|(\bm I - \bm P_{-j})(\bm X_j)\|^2\\
             \implies & -\bm X_j^T  \bm P_{-j}\bm P_{\mathcal A(a)} \bm P_{-j}\bm X_j \frac{\partial }{\partial a}\hat{\beta}_j(a)+\|\bm X_j\|^2 \frac{\partial }{\partial a}\hat \beta^{\lambda}_j(a) = \|(\bm I - \bm P_{-j})(\bm X_j)\|^2\\
             \implies & -\bm X_j^T  \bm P_{-j}\bm P_{\mathcal A(a)}\bm P_{\mathcal A(a)} \bm P_{-j}\bm X_j \frac{\partial }{\partial a}\hat{\beta}_j(a)+\|\bm X_j\|^2 \frac{\partial }{\partial a}\hat \beta^{\lambda}_j(a) = \|(\bm I - \bm P_{-j})(\bm X_j)\|^2\\
            \implies & -\|\bm P_{\mathcal A(a)} \bm P_{-j}\bm X_j\|^2 \frac{\partial }{\partial a}\hat{\beta}_j(a)+\|\bm X_j\|^2 \frac{\partial }{\partial a}\hat \beta^{\lambda}_j(a) = \|(\bm I - \bm P_{-j})(\bm X_j)\|^2\\
             \implies & \frac{\partial }{\partial a}\hat \beta^{\lambda}_j(a) = \frac{\|(\bm I - \bm P_{-j})(\bm X_j)\|^2}{\|\bm X_j\|^2 - \|\bm P_{\mathcal A(a)} \bm P_{-j}\bm X_j\|^2},
        \end{align*}
        
        which is positive (as $\|\bm P_{\mathcal A(a)} \bm P_{-j}\bm X_j\|^2\leq \| \bm P_{-j}\bm X_j\|^2\leq \|\bm X_j\|^2$), whenever $\hat \beta^{\lambda}_j(a)\neq 0$. Hence, for an $a\in \mathbb R$, either $\hat \beta^{\lambda}_j(a) = 0$ or $\frac{\partial \hat\beta_j(a)}{\partial a}>0$, showing that $\hat \beta^{\lambda}_j$ is locally increasing around $a$ in the latter case. Now define,
        \[
            f_{\bm S^{(j)}}(u) = \hat \beta_j^\lambda\left(u\cdot \frac{\hat \sigma_j}{\|(\bm I - \bm P_{-j})\bm X_j\|}\right),
        \]
        and note that, $\hat \beta_j^{\lambda} = f_{\bm S^{(j)}}(u_1)$. Furthermore, note that $f_{\bm S^{(j)}}$ is a functional of the sufficient statistic, $\bm S^{(j)}$, and using the arguments above, is a continuous, piece-wise linear function that is increasing in the region, $\{u: f_{\bm S^{(j)}}(u) \neq 0\}$. This establishes all the claims about $f_{\bm S^{(j)}}$ in the statement of Theorem \ref{thm:distexpression} except for \eqref{eqn:lambda_characterization}, which we turn to next.

\end{proof}

\subsubsection{Proof of Equation \eqref{eqn:lambda_characterization} in Theorem \ref{thm:distexpression}}
\label{sec:proof_distexpression_2}
To prove Equation \eqref{eqn:lambda_characterization}, we first start with an intermediate result.        
\begin{theorem}
\label{thm:eventequivalence}
   For data $(\bm y, \bm X)$, with $\bm X$ full column-rank and $\lambda>0$, define $\hat{\bm \beta}_*^\lambda(b)$ to be $b$ at the $j^{\mathrm{th}}$ coordinate and $\hat{\bm \beta}_{-j}^\lambda(b)$ on the rest, where, $\hat{\bm \beta}_{-j}^\lambda(b)$ is defined as in Equation \eqref{eqn:betax}. Also let
    \begin{equation}
    \label{eqn:lassoobjective}
        f_{\lambda}(\bm y;\bm \beta) := \frac{1}{2n}\|\bm y - \bm X \bm \beta\|^2 + \lambda \|\bm \beta\|_1,
    \end{equation}
    be the LASSO objective function. 
    Then for any $b\in \mathbb R$, the following are equivalent
    \begin{enumerate}[label = (\alph*)]
        \item $0\in \partial_{\beta_j} f_{\lambda}({\bm y}; {\bm \beta})\big\lvert_{\bm \beta = \hat{\bm\beta}_*^\lambda(b)}$
        \item $\hat{\bm \beta}^\lambda = \hat{\bm \beta}_*^\lambda(b)$
        \item $\hat \beta_j^\lambda = b$
    \end{enumerate}

\end{theorem}

\begin{proof}[Proof of Theorem \ref{thm:eventequivalence}]
        We first show the equivalence of (a) and (b).
        Assume that, $0 \in  \partial_{\beta_j} f_{\lambda}({\bm y}; {\bm \beta})\big\lvert_{\bm \beta = \hat{\bm\beta}_*^\lambda(b)} $. The convexity of $f_{\lambda}({\bm y}; (\bm \beta_{-j} = \hat {\bm \beta}_{-j}^\lambda(b), \beta_j))$ in $\beta_j$ shows that the $j^{\mathrm{th}}$ entry of $\hat {\bm \beta}_*^\lambda(b)$ is the minimizer of $f_{\lambda}({\bm y}; (\bm \beta_{-j} = \bm \hat {\bm \beta}_{-j}^\lambda(b), \beta_j))$ in $\beta_j$. That is,
        \begin{align*}
                 \underset{\beta_{j}}{\arg\min} f_{\lambda}({\bm y};(\bm \beta_{-j} = \hat{\bm \beta}_{-j}^\lambda(b), \beta_j))  = b.
        \end{align*}
        But by the definition of $\hat {\bm \beta}^{\lambda}_{-j}(b)$, we know that
        \[
             \underset{\bm \beta_{-j}}{\arg\min} f_{\lambda}({\bm y};(\bm \beta_{-j}, \beta_j = b)) = \hat{\bm \beta}_{-j}^\lambda(b).
        \]
        Thus, if one runs a blockwise coordinate descent with blocks $\{j\}$ and $[1:p]\setminus \{j\}$ 
        starting at $\hat{\bm \beta}_*^\lambda(b)$, we see that the iterates will be constant at $\hat{\bm \beta}_*^\lambda(b)$, thereby implying that this is a limit point of the iterates. One can now invoke \citet[Proposition 5.1]{tseng2001} (the conditions for applying this proposition follow directly as $f_{\lambda}$ can be separated into the squared error loss and the non-differentiable $\ell_1$-penalty) to conclude that,
        \[
            \hat{\bm \beta} = \underset{\bm \beta}{\arg\min}f_{\lambda}({\bm y};\bm \beta)= \hat{\bm \beta}_*^\lambda(b),
        \]
        which establishes the implication of (a) to (b). The reverse implication follows directly from the fact that $\hat{\bm \beta}$ is the optimizer of the LASSO objective, so that each coordinate, and in particular the $j^{\mathrm{th}}$ coordinate of the sub-gradient, contains 0, that is, $0\in \partial_{\beta_j} f_{\lambda}({\bm y}; {\bm \beta})\big\lvert_{\bm \beta = \hat{\bm\beta}}$. The fact that $ \hat{\bm \beta}= \hat{\bm \beta}_*^\lambda(b)$ completes the argument. It is also straightforward to see that (b) implies (c). To show that (c) implies (b), note that one can use the blockwise coordinate descent argument used above to conclude that $\hat {\bm \beta}^{\lambda} = \hat {\bm \beta}^{\lambda}(\hat \beta_j^\lambda)$, and then use the hypothesis of (c) (that is, $\hat \beta_j^\lambda = b$) to conclude (b).
\end{proof}
Hence, Theorem \ref{thm:eventequivalence} now implies that $\hat \beta_j^{\lambda} = b$ if and only if $0\in \partial_{\beta_j}f_{\lambda}(\bm y; \bm \beta)\lvert _{\bm \beta = \hat{\bm \beta}_*^\lambda(b)}$. We will now prove item 2 of Theorem \ref{thm:distexpression} by evaluating this sub-gradient.

\begin{proof}[Proof of item 2 of Theorem \ref{thm:distexpression}]
    Note that we have the following decomposition,
    \begin{align*}
        &f_{\lambda}({\bm  y}; \bm \beta)\\
        &= \frac{1}{2n}\|{\bm  y} - \bm X \bm \beta\|^2 + \lambda \|\bm \beta\|_1\\
        &= \frac{1}{2n}\|\hat{\bm y}_j +\hat \sigma_j \bm V \bm u - \bm P_{-j}\bm X_j \beta_j - (\bm I - \bm P_{-j})\bm X_j \beta_j - \bm X_{-j}\bm \beta_{-j}\|^2 + \lambda \|\bm \beta\|_1\\
        &= \frac{1}{2n}\|\hat{\bm y}_j -\bm P_{-j}\bm X_j\beta_j- \bm X_{-j}\bm \beta_{-j} \|^2 + \frac{1}{2n}\|\hat \sigma_j \bm V \bm u - (\bm I - \bm P_{-j})\bm X_j \beta_j \| + \lambda |\beta_j| + \lambda \sum_{i\neq j}|\beta_i|
    \end{align*}
    Now define,
    \[
        s(\beta_j) = \begin{cases}
            [-1,1], &\beta_j = 0\\
            \mathrm{sign}(\beta_j), & \beta_j \neq 0
        \end{cases}.
    \]
    Then we have,
    \begin{align*}
        & \partial_{\beta_j}f_{\lambda}({\bm  y}; \bm \beta)\\
        &= -\frac{1}{n}\bm X_j^T\bm P_{-j}(\hat{\bm y}_j - \bm P_{-j}\bm X_j\beta_j- \bm X_{-j}\bm \beta_{-j}) -\frac{1}{n}\bm X_j^T (\bm I - \bm P_{-j}) \left(\hat \sigma_j \bm V \bm u - (\bm I - \bm P_{-j})\bm X_j \beta_j \right) + \lambda s(\beta_j)\\
        &= -\frac{1}{n}\bm X_j^T \hat{\bm y}_j -\frac{\hat \sigma_j}{n}\bm X_j^T \bm V\bm u + \frac{\bm X_j^T \bm P_{-j}\bm X_j\beta_j  + \bm X_j^T(\bm I - \bm P_{-j})\bm X_j\beta_j + \bm X_j^T \bm X_{-j}\bm \beta_{-j}}{n} + \lambda s(\beta_j)\\
        &= -\frac{1}{n}\bm X_j^T \hat{\bm y}_j -\frac{\hat \sigma_j}{n}\bm X_j^T \bm V\bm u + \frac{\bm X_j^T \bm X_j\beta_j + \bm X_j^T \bm X_{-j}\bm \beta_{-j}}{n} + \lambda s(\beta_j)\\
        &= -\frac{1}{n}\bm X_j^T \hat{\bm y}_j -\frac{\hat \sigma_j}{n}\bm X_j^T \bm V\bm u + \frac{\bm X_j^T \bm X\bm \beta}{n} + \lambda s(\beta_j)\\
        &= -\frac{1}{n}\bm X_j^T (\hat{\bm y}_j - \bm X\bm \beta) - \frac{\hat \sigma_j}{n}\bm X_j^T\bm V \bm u + \lambda s(\beta_j).
    \end{align*}
    From the calculations in Section \ref{sec:proof_pvalequivalence},
    \[
        \bm X_j^T \bm V \bm u = \bm X_j^T \sum_{i=1}^{n-d+1}\bm v_i u_i = \bm X_j^T\bm v_ju_1 = \frac{\bm X_j^T(\bm I-\bm P_{-j})\bm X_j}{\|(\bm I - \bm P_{-j})\bm X_j\|}u_1 = \|(\bm I - \bm P_{-j})\bm X_j\|u_1,
    \]
    and thus we have,
    \begin{align*}
        &\partial_{\beta_j}f_{\lambda}({\bm  y};\bm \beta)\\
        &= -\frac{1}{n}\bm X_j^T (\hat{\bm y}_j - \bm X\bm \beta) - \frac{\hat \sigma_j\|(\bm I - \bm P_{-j})\bm X_j\|}{n}u_1 + \lambda s(\beta_j).
    \end{align*}
    Setting $\bm \beta = \hat{\bm \beta}^*(b)$ for a $b$, the above equation along with Theorem \ref{thm:eventequivalence} establishes that for $b\neq 0$, $\hat \beta_j^{\lambda} = b$ if and only if,
    \[
        u_1 = \frac{-\bm X_j^T(\hat{\bm y}_j - b \bm X_j - \bm X_{-j}\hat {\bm \beta}_{-j}(b))+n\lambda \mathrm{sign}(b)}{\hat \sigma_j \|(\bm I - \bm P_{-j})\bm X_j\|} = \Lambda_j(b, \mathrm{sign}(b)),
    \]
    while for $b\neq 0$, $\hat \beta_j^\lambda = 0$ if and only if,
    \[
        u_1 \in \left[ \frac{-\bm X_j^T(\hat{\bm y}_j - b \bm X_j - \bm X_{-j}\hat {\bm \beta}_{-j}(b))\pm n\lambda }{\hat \sigma_j \|(\bm I - \bm P_{-j})\bm X_j\|}\right] = [\Lambda_j(0,-1), \Lambda_j(0,1)].
    \]
    Noting that $\hat \beta_j^\lambda = f_{\bm S^{(j)}}(u_1)$ completes the proof of item 2.
\end{proof}

\subsection{Proof of the fact that $T_j$ is independent of the sufficient statistic, $\mathbf{ S}_j$ under $H_{j}:\beta_j = 0$}
\label{sec:app_ttest_indep_proof}

\begin{proof}
In Section \ref{sec:proof_pvalequivalence}, we showed that,
\begin{align*}
    \Tstat &\propto  \frac{\hat\beta_{j,\mathrm{OLS}}}{\sqrt{\hat \sigma_j^2 - {\hat\beta_{j,\mathrm{OLS}}}^2}}\\
    &= \frac{\frac{\hat\beta_{j,\mathrm{OLS}}}{\hat \sigma_j}}{\sqrt{1 - \left(\frac{\hat\beta_{j,\mathrm{OLS}}}{\hat \sigma_j}\right)^2}},
\end{align*}
where the proportionality constant consists of terms that entirely depend on the design matrix. Thus, the only stochastic component in the expression for $\Tstat$ is $\frac{\hat\beta_{j,\mathrm{OLS}}}{\hat \sigma_j} = \frac{\bm X_j^T(\bm I - \bm P_{-j})\bm y}{\|(\bm I - \bm P_{-j})\bm X_j\|^2\hat \sigma_j} =: \frac{\bm X_j^T}{\|(\bm I - \bm P_{-j})\bm X_j\|^2}\bm L_j$, where $\bm L_j$ equals the term its replacing. Thus, it suffices to show that under $H_{j}$, the unconditional distribution of $\bm L_j$ is the same as its conditional distribution, $\bm L_j\mid \bm S^{(j)}$.\\

Note that from Equation \eqref{eqn:conddist_lm}, we have that under $H_{j}$
\[
    \bm L_j\mid \bm S^{(j)} \sim \bm V \bm u,
\]
where, $\bm u$ is uniformly distributed over $\mathbb S^{n-d}$.
Now, let us evaluate the unconditional distribution. Under $H_{j}$, we can write,
\[
    \bm y \sim \bm X_{-j}\bm \beta_{-j} + \bm \epsilon,
\]
for some, $\bm \epsilon \sim N_n(\bm 0, \sigma^2 \bm I_n)$. Then, because $\bm V$ denotes a matrix with columns forming an orthonormal basis for the complement of the columnspace of $\bm X_{-j}$, we have, $\bm I - \bm P_{-j} = \bm V\bm V^T$. Thus, we have under $H_{j}$,
\begin{align*}
    \bm L_j &= \frac{\bm V\bm V^T \bm y}{\|\bm V \bm V^T \bm y\|}\\
    &= \frac{\bm V\bm V^T \bm \epsilon}{\|\bm V \bm V^T \bm \epsilon\|}\\
    &= \bm V \frac{\bm V^T \bm \epsilon}{\|\bm V^T \bm \epsilon\|}. \textrm{ [since orthogonal transformations do not change the norm]}.
\end{align*}
Now, since $\bm V$ is orthogonal, we have, $\bm v = \bm V^T \bm \epsilon \sim \mathcal{N}(\bm 0, \sigma^2 \bm I_{rank(\bm V)}) =  \mathcal{N}(\bm 0, \sigma^2 \bm I_{n-d+1})$. Thus we have under $H_{j}$,
\[
   \bm L_j =  \bm V \frac{\bm v}{\|\bm v\|}  = \bm V \bm u^*,
\]
where, $\bm u^* = \bm v/\|\bm v\|$ is uniformly distributed over $\mathbb S^{n-d}$. This completes the proof.
\end{proof}

\subsection{Proof of the map of $U_j$ to the $t$-distribution}
\label{sec:proof_uquantiles}
In this section, we will prove that $ \bm u\sim \mathrm{Unif}(\mathbb S^m)$ implies that
\[
    \frac{\sqrt{m}\cdot u_j}{\sqrt{1-u_j^2}}\sim t_{m}.
\]
The proof follows from the following representation of $\bm u$: Let $\bm X\sim N_{m+1}(\bm 0, \bm I_{m+1})$, then we have that,
\[
    u_j \stackrel{d}{=}\frac{X_j}{\sqrt{X_j^2+\sum_{i\neq j} X_i^2}}\implies \sqrt{m}\cdot\frac{u}{\sqrt{1-u^2}}\stackrel{d}{=} \frac{\sqrt{m}X_j}{\sqrt{\sum_{i\neq j} X_i^2}}.
\]
The proof follows from the fact that $X_j\sim N(0,1)$ and $\sum_{i\neq j}X_i^2\sim \chi^2_{m}$, independent of $X_j$.

\color{black}
\subsection{Supplementary details for and proof of Theorem \ref{thm:asymp}}
\label{sec:proof_asymp}
Before presenting a proof of Theorem~\ref{thm:asymp}, we introduce some notation for and formalize cross-validation, and present a technical assumption about it.

\subsubsection{Formalization of cross-validation for Theorem~\ref{thm:asymp}}
\label{sec:missing_details}
For a general $(\bm y, \bm X_{n'\times d'})$, we formalize cross-validation as follows: For $m\in \mathbb N$, let $\Xi = \{S_1,\dots, S_m\}$ denote a partition of $[1:n']$ into $m$ disjoint non-empty sub-groups, and let $\tilde{\bm z}\in \mathbb R^{n'-d'+1}$. Define the multivariate function $(\bm y, \bm X)\mapsto\bm l(\bm y, \bm X)$ as the function that maps the data to a set of (positive) candidate $\lambda$ values. Then the value chosen by cross-validation, denoted $\hat\lambda(\bm y,\bm X)$, is the element of $\bm l(\bm y, \bm X)$ that minimizes (with ties broken arbitrarily) the average error across $S_1,\dots,S_m$, where the error for $S_j$ is computed by training the LASSO (with solution chosen arbitrarily among minimizers if non-unique) on the data points $[1:n']\setminus S_j$ and computing the mean squared error of that fitted LASSO on the observations in $S_j$. Applying this to the cross-validation procedure in the $\ell$-test, the chosen $\hat\lambda$ is given by $\hat\lambda\left(\bm P_{-j}\bm y + \sigma\bm V\tilde{\bm z}, \bm X\right)$ for $\Xi$ and $\tilde{\bm z}$ chosen independently of $(\bm y,\bm X)$. 

Since the following assumption is only used in Theorem~\ref{thm:asymp}, we will write $\Xi$ and $\tilde{\bm z}$ as $\Xi^{(n)}$ and $\tilde{z}^{(n)}$ (the latter is scalar-valued because $n'=d'=d$ in the setting of Theorem~\ref{thm:asymp}).


\begin{assm}
    \label{assm:cv}
    \
    \begin{enumerate}
        \item The random variables $\Xi^{(n)}$ and $\tilde{z}^{(n)}$ are sampled independently of one another and of $(\hat{\bm \theta}_n, \hat{\bm \Sigma}_n)$ from a distribution that does not depend on $n$.
        \item Letting $r$ denote the maximum integer such that with probability 1, all elements of $\Xi^{(n)}$ have size at least $r$, then any matrix comprised of $r$ or more rows of $\bm \Sigma^{-1/2}$ has columns in general position.
        \item For $\bm X = \bm \Sigma^{-1/2}$ and any (possibly, randomly drawn) $\tilde z$, any partition $\Xi$ of $[1:d]$, any $\bm \beta\in \mathbb R^{d}$, and $\bm y\sim \mathcal N(\bm \Sigma^{-1/2}\bm \beta, \sigma^2\bm I_d)$ (drawn independently of $\tilde {z}$, if $\tilde z$ is drawn randomly), with probability 1, no two distinct $\lambda_1,\lambda_2\in \bm l(\bm P_{-j}\bm y + \sigma\bm V \tilde { z},\bm X)$ have the same cross-validated error. 
        \item (a) The map $\bm l$ outputs a vector of positive values whose dimension does not depend on its inputs
        and (b) the map $(\bm y, \bm X)\mapsto \bm l(\bm y, \bm X)$ is continuous.
    \end{enumerate}
\end{assm}

Note that items 1 and 4 are entirely under the control of the analyst and hold for typical choices.
For example, the function \texttt{glmnet()} in the \textbf{R} package \texttt{glmnet} \citep{glmnet1} by default sets the grid of candidate $\lambda$ values to be 100 values equally spaced on the logarithmic scale between $0.0001\bm X^T\bm y/n$ and $\bm X^T\bm y/n$. Items 2 and 3 are also quite weak assumptions and are simply needed to ensure asymptotic uniqueness of the LASSO solutions within the cross-validation procedure and of the final minimizer $\hat\lambda$; these are needed to argue continuity of the entire procedure as we do in the proof of Theorem~\ref{thm:asymp}.


\subsubsection{Proof of Theorem~\ref{thm:asymp}}
\begin{proof}[proof of Theorem~\ref{thm:asymp}]
    Let $p_{j,n}$ denote the known-$\sigma$-$\ell$-test $p$-value for testing $H_j:\theta_j = 0$ using response $\sqrt{n}\hat {\bm \Sigma}_n^{-1/2}\hat{\bm \theta}_n$($=:\bm y^{(n)}$) and design matrix $\hat{\bm \Sigma}_n^{-1/2}$ ($=:\bm X^{(n)}$) and $\hat\lambda$ chosen using cross-validation as formalized in Section~\ref{sec:missing_details}.

    First note that we can write our $p$-values as $p_{j,n} = \zeta^{(\Xi^{(n)})}(\bm y^{(n)},\bm X^{(n)},\tilde z^{(n)})$, for an appropriate map $\zeta^{(\Xi^{(n)})}$, where $\Xi^{(n)}$ and $\tilde z^{(n)}$ are described in item 1 of Assumption~\ref{assm:cv}, where the partition and the $\tilde z$ are sampled separately for each $n$. Now, fix any (non-random) $\Xi$ partitioning the rows of $\bm X^*:=\bm \Sigma^{-1/2}$ into $m$ disjoint sets, with each set having size at least $r$. Suppose for this fixed $\Xi$, we consider the sequence $\{\zeta^{(\Xi)}(\bm y^{(n)},\bm X^{(n)},\tilde z)\}_n$, for some $\tilde z$ drawn from the common distribution of $\tilde z^{(n)}$ as stated in item 1 of Assumption~\ref{assm:cv} (let us denote this distribution by $\mu_z$), independent of $\left\{\left(\bm y^{(n)},\bm X^{(n)}\right)\right\}_n$.

    We start with the key lemma showing the continuity of the map $(\bm y, \bm X)\mapsto  \zeta^{(\Xi)}\left(\bm y, \bm X, \tilde{ z}\right)$, whose proof is deferred to Section~\ref{sec:additional_lemma_proofs}.

    \begin{lemma}
    \label{lem:asymp_pval_cont}
    Let $\bm y^*\sim \mathcal N(\bm X^*\bm \theta, \bm I)$ and $\tilde z$ be drawn independently of $\bm y^*$. Then, for the partition $\bm \Xi$ fixed above, and under Assumption~\ref{assm:cv}, the map $(\bm y, \bm X)\mapsto  \zeta^{(\Xi)}\left(\bm y, \bm X, \tilde{ z}\right)$ is continuous at $(\bm y^*, \bm X^*)$, almost surely. Here, the `almost sure' statement is made under the law of $(\bm y^*,\tilde z)$.
    \end{lemma}
    
    Thus, Lemma~\ref{lem:asymp_pval_cont} implies, 
    \begin{align}
        \label{eqn:prob_1_transformation}
        \begin{aligned}
        \mathbb P_{\tilde z,\bm y^*}\left((\bm y, \bm X)\mapsto \zeta^{(\Xi)}\left(\bm y, \bm X, \tilde{z}\right)\textrm{ is continuous at }(\bm y^*, \bm X^*)\right)=1.\textrm{ [law of double expectation]}
        \end{aligned}
    \end{align}
    Now define the event,
    \[
        E = \left\{(\bm y, \bm X, \tilde{z}):\zeta^{(\Xi)}_{\tilde z}\textrm{ is continuous at }(\bm y, \bm X)\right\},
    \]
    where $\zeta^{(\Xi)}_{\tilde z}(\bm y, \bm X) := \zeta^{(\Xi)}(\bm y, \bm X, \tilde z)$ and note that because of Equation~\eqref{eqn:prob_1_transformation},\\ $ \bm E^* = \left\{(\bm y, \bm X,\tilde{z})\in E: \bm X=\bm X^*\right\}\subseteq E$ has probability 1 under the law of $(\bm y^*, \bm X^*,\tilde z)$ (which is given by $\mu^* = \mathcal N\left(\bm X^*\bm \theta,\bm I\right)\otimes\delta_{\left(\bm X^*\right)}\otimes \mu_z$), and hence, so does the set $E$. 

    For the next step, we will need the following auxiliary lemma whose proof is again deferred to Section~\ref{sec:additional_lemma_proofs}.
    \begin{lemma}
    \label{lem:asymp_6}
        Let $W_n\dto W$ and consider a random variable $U$, such that $U\ind\{W_n\}_n$. For a function $(w,u)\mapsto f(w,u)$ define,
        \[
            E = \left\{(w,u):f_u\textrm{ is continuous at }w\right\},
        \]
        where $f_u(w):=f(w,u)$. Let $\mu_W,\mu_U$ denote the laws of $W$ and $U$ respectively. If $\mu_W\otimes\mu_U(E) = 1$, then $f(W_n,U)\dto f(W,U)$. Here, $\mu\otimes\nu$ denotes the product measure formed by independent cross of $\mu$ and $\nu$. 
    \end{lemma}

    Hence, invoking Lemma~\ref{lem:asymp_6}, $\zeta^{(\Xi)}\left(\bm y^{(n)},\bm X^{(n)},\tilde z\right)\dto \zeta^{(\Xi)}\left(\bm y^*, \bm X^*,\tilde z\right)$ for any fixed $\Xi$. Note that $\Xi$ can take only finitely many values and let us denote that finite space of values by $\mathcal F_{\Xi}$. Then, we have that
    \begin{align*}
        &\zeta^{(\varphi)}\left(\bm y^{(n)},\bm X^{(n)},\tilde z\right)\dto \zeta^{(\varphi)}\left(\bm y^*, \bm X^*,\tilde z\right), \forall \varphi\in \mathcal F_{\Xi}\\
        \implies &\mathbb P\left(\zeta^{(\varphi)}\left(\bm y^{(n)},\bm X^{(n)},\tilde z\right)\leq b\right)\to \mathbb P\left(\zeta^{(\varphi)}\left(\bm y^{*},\bm X^{*},\tilde z\right)\leq b\right), \forall b\in D_{\varphi}\textrm{, with }\mu^*(D_{\varphi})=1, \forall \varphi\in \mathcal F_{\Xi}.
    \end{align*}
    Define $D=\cap_{\varphi \in \mathcal F_{\Xi}}D_{\varphi}$ and because $\mathcal F_{\Xi}$ is finite, we conclude that $\mu^*(D) = 1$. Now, let $\tilde \Xi$ be a randomly chosen partition, independent of $\tilde z$ and $\{(\bm y^{(n)}, \bm X^{(n)})\}_n$, then by the definition of $D$, for all $b\in D$,
    \begin{align*}
        \mathbb P\left(\zeta^{(\tilde\Xi)}\left(\bm y^{(n)},\bm X^{(n)},\tilde z\right)\leq b\right)=&\sum_{\varphi\in \mathcal F_{\Xi}} \mathbb P\left(\zeta^{(\varphi)}\left(\bm y^{(n)},\bm X^{(n)},\tilde z\right)\leq b\mid \tilde \Xi = \varphi\right)\mathbb P\left(\tilde\Xi = \varphi\right)\\
        \underset{n\to \infty}{\longrightarrow}&\sum_{\varphi\in \mathcal F_{\Xi}} \mathbb P\left(\zeta^{(\varphi)}\left(\bm y^{*},\bm X^{*},\tilde z\right)\leq b\mid \tilde \Xi = \varphi\right)\mathbb P\left(\tilde \Xi = \varphi\right)\textrm{ [Finiteness of }\mathcal F_{\Xi}]\\
        =&\mathbb P\left(\zeta^{(\tilde\Xi)}\left(\bm y^{*},\bm X^{*},\tilde z\right)\leq b\right).
    \end{align*}
    Finally note that,
    \begin{align*}
        & \mathbb P\left(\zeta^{(\tilde\Xi)}\left(\bm y^{(n)},\bm X^{(n)},\tilde z\right)\leq b\right)\\
        =&  \mathbb P(\zeta^{(\Xi^{(n)})}(\bm y^{(n)},\bm X^{(n)}, \tilde z^{(n)})\leq b)\textrm{ [because }(\Xi^{(n)}, \tilde z^{(n)}, \bm y^{(n)}, \bm X^{(n)})\stackrel{d}{=}(\tilde\Xi, \tilde z, \bm y^{(n)}, \bm X^{(n)})]\\
        =& \mathbb P(p_{j,n}\leq b).
    \end{align*}
    Thus, we have $\forall b\in D$, $ \mathbb P\left( p_{j,n}\leq b\right)\underset{n\to \infty}{\longrightarrow} \mathbb P\left(\zeta^{(\tilde\Xi)}\left(\bm y^{*},\bm X^{*},\tilde z\right)\leq b\right)$. Because $\mu^*(D)=1$, we have,
    \[
        p_{j,n}\dto \zeta^{(\tilde\Xi)}\left(\bm y^{*},\bm X^{*},\tilde z\right).
    \]
    Because this limit is a known-$\sigma$-$\ell$-test $p$-value of a Gaussian linear model, we conclude that $\{p_{j,n}\}$ is an asymptotically valid sequence of $p$-values.
\end{proof}

\subsubsection{Lemma proofs}
\label{sec:additional_lemma_proofs}
\begin{proof}[Proof of Lemma~\ref{lem:asymp_pval_cont}]
    This proof is primarily composed of the following three lemmas, whose proofs are deferred to after the end of this proof.

    \begin{lemma}
        \label{lem:asymp_2}
        If $g(\bm x, \bm y)$ is a continuous function in $(\bm x, \bm y)\in \mathcal X\times \mathcal Y$ and so is $\bm x \mapsto z(\bm x)\in \mathcal Y$, then, $\bm x \mapsto g(\bm x, z(\bm x))$ is also continuous.
    \end{lemma}

    \begin{lemma}
    \label{lem:asymp_4}
        For a linear model with response $\bm y$ and design matrix $\bm X$, let (over-loading notations) $p_j^\lambda(\bm y, \bm X)$ be the known-$\sigma$-$\ell$-test $p$-value with regularizer $\lambda>0$. Let us define $\bm \xi := (\bm y, \bm X, \lambda)$. Then for a $\bm \xi^*= (\bm y^*, \bm X^*, \lambda^*)$ such that $\bm X^*$ is full column-rank, there exists a neighborhood $V$ about $\bm \xi^*$ such that the map $\bm \xi \mapsto p_j^\lambda(\bm y, \bm X)$ is a real-valued function in $V$ and continuous at $\bm \xi^*$.
    \end{lemma}    

    \begin{lemma}
        \label{lem:asymp_5}
        Suppose a matrix $\bm X^*_{n'\times d'}$ satisfies item 2 of Assumption~\ref{assm:cv} and that $\bm y^*\in \mathbb R^{n'}$, a partition $\Xi$ of $[1:n']$ as described in Section~\ref{sec:missing_details}, and $\tilde{\bm z} \in \mathbb R^{n'-d'+1}$ are such that no two $\lambda \in \bm l(\bm P^*_{-j}\bm y^*+\sigma\bm V^*\tilde{\bm z}, \bm X^*)$ give rise to the same cross-validated error (here, $\bm P_{-j}^*,\bm V^*$ are computed based on $\bm X^*$). Furthermore, suppose item 4(a) of  Assumption~\ref{assm:cv} is also satisfied. Then $(\bm y, \bm X)\mapsto \hat{\lambda}$($=\hat \lambda(\bm P_{-j}\bm y + \sigma \bm V\tilde{\bm z})$ that uses the partition $\Xi$ for cross-validation) is continuous at $(\bm y^*, \bm X^*)$.
    \end{lemma}

    Now note $\left(\bm y^{(n)},\bm X^{(n)}\right)\dto \left(\bm y^*,\bm X^*\right)$, with $\bm y^*\sim \mathcal N\left(\bm X^*\bm \theta, \bm I\right)$. Hence, Assumption~\ref{assm:cv}, along with Lemma~\ref{lem:asymp_5} implies that $(\bm y, \bm X)\mapsto \hat \lambda$ is continuous at $(\bm y^*, \bm X^*)$ almost surely, while Lemma~\ref{lem:asymp_4} implies that $(\bm y, \bm X,\lambda)\mapsto p_j^\lambda(\bm y, \bm X)$ is continuous at $(\bm y^*, \bm X^*,\lambda^*)$, for any $\lambda^*$, almost surely. Lemma~\ref{lem:asymp_2} now implies $(\bm y, \bm X)\mapsto  \zeta^{(\Xi)}\left(\bm y, \bm X, \tilde{ z}\right)$ is continuous at $(\bm y^*, \bm X^*)$, almost surely.
\end{proof}

\begin{proof}[Proof of Lemma~\ref{lem:asymp_2}]
    Fix $\epsilon>0$ and $\bm x^*\in \mathcal X$, $\bm y^* = \bm z(\bm x^*)$. Because $(\bm x, \bm y)\mapsto g(\bm x, \bm y)$ is a continuous transformation, we have that there is some $\delta$ such that $\|(\bm x, \bm y) - (\bm x^*, \bm y^*)\|<\delta\implies |g(\bm x, \bm y) - g(\bm x^*, \bm y^*)|<\epsilon$. Now, fix a $\delta_1>0$. Then there is a $0<\delta_2<\delta_1$ such that $\|\bm x- \bm x^*\|<\delta_2\implies |\bm z(\bm x) - \bm z(\bm x^*)|<\delta_1$ because of the continuity of $\bm z$. This also implies that $\|(\bm x, \bm z(\bm x)) - (\bm x^*, \bm z(\bm x^*))\|< \sqrt{\delta_1^2 + \delta_2^2}<\delta_1\sqrt{2}$ for $\bm x$ such that $\|\bm x - \bm x^*\|<\delta_2$. Choose $\delta_1 = \delta/\sqrt{2}$ and let $\delta_2'$ denote the corresponding $\delta_2$. Then, we have that 
    \[
        \|\bm x - \bm x^*\|<\delta_2 \implies \|(\bm x, \bm z(\bm x)) - (\bm x^*, \bm z(\bm x^*))\|<\delta \implies |g(\bm x, \bm z(\bm x)) - g(\bm x^*, \bm z(\bm x^*))|<\epsilon.
    \]
    This completes the proof.
\end{proof}

\begin{proof}[Proof of Lemma~\ref{lem:asymp_4}]
    Observe that from the expression of $p_j^\lambda(\bm y, \bm X)$,
    \begin{align}
    \label{eqn:ell_pval_lambda}
        p_j^\lambda(\bm y, \bm X) = \begin{cases}
            \Phi(-\Lambda_j(-|\hat \beta_j^{\lambda}|,-1)) + 1 - \Phi(\Lambda_j(|\hat \beta_j^{\lambda}|,1)), &\textrm{if }\hat\beta_j^{\lambda} \neq 0 \Leftrightarrow z_1\notin [\Lambda_j(0,\pm 1)]\\
            \Phi(\hat m_j - |z_1 - \hat m_j|) + 1- \Phi(\hat m_j + |z_1 - \hat m_j|), &\textrm{if }\hat\beta_j^{\lambda} = 0\Leftrightarrow z_1\in [\Lambda_j(0,\pm 1)]
        \end{cases},
    \end{align}
    where $\Lambda_j$ and $\hat{m}_j$ refer to their known-$\sigma$ variants as discussed in Section~\ref{sec:known_sigma_ell} with $\sigma$ in the denominator in place of $\hat \sigma_j$. Three LASSO-estimate-based quantities enter into Equation~\eqref{eqn:ell_pval_lambda}: $\hat{\beta}_j^\lambda$, $ \hat{\bm \beta}^\lambda_{-j}(|\hat \beta_j^\lambda|)$, and $\hat{\bm \beta}^\lambda_{-j}(-|\hat \beta_j^\lambda|)$.
    
    If $\bm X$ is full column-rank (more generally, if its columns are in general position), the LASSO estimator is unique \citep[Lemma 3]{lassouniqueness}. Now, first observe the fact that $\bm X\mapsto \bm X^T\bm X\mapsto \det\left(\bm X^T\bm X\right)$ is a continuous transformation. A matrix $\bm X^*$ being full column-rank implies that $\det\left({\bm X^*}^T\bm X^*\right)>0$, and hence, by the continuity of the transformation, this implies that for an $\bm X$ in an open ball $U$ around $\bm X^*$, we have $\det\left(\bm X^T\bm X\right)>0$, and hence, all $\bm X$'s in that ball are full column-rank. This means that we can find a ball $V$ about $\bm \xi^*$ such that $\bm \xi\in V$ implies that $\bm X\in  U$ and hence, they are full column-rank, and hence so is any $\bm X_{-j}$ for $\bm X\in  U$. Hence, as along as $\bm \xi\in V$, all the LASSO-based quantities $\hat{\beta}_j^\lambda$ and $\hat{\bm \beta}^\lambda_{-j}(b),\forall b$, and in particular $\hat{\bm \beta}^\lambda_{-j}(\pm |\hat \beta_j^\lambda|)$, are uniquely defined. Hence, we can find an open neighborhood $V$ about $\bm \xi^*$ such that for any $\bm \xi\in V$, the above LASSO-estimate-based quantities are unique, and hence, $\bm \xi\mapsto p_j^\lambda(\bm y, \bm X)$ is a real-valued (and not set-valued) function in $V$. As the remainder of the proof concerns validity at $\bm\xi^*$, we can and will always restrict our functions to $V$ (without mentioning so explicitly) so that we can treat those functions as vector-valued (as opposed to set-valued).

    Let us define $N_j := \{\bm \xi: \hat \beta_j^\lambda\neq \{0\}\}$. We will establish the continuity of $\bm \xi\mapsto p_j^\lambda(\bm y, \bm X)$ at $\bm \xi^*$ separately for the three cases: $\bm \xi^*\in N_j$, $\bm \xi^*\in \mathrm{int}~N^c_j$ and $\bm \xi^*\in \del N_j$, where $\mathrm{int}~A$ denotes the interior of a set $A$ (defined as the union of all possible open subsets of $A$), and $\del A:=\bar A\setminus \mathrm{int}~A$ denotes the boundary of the set $A$, where $\bar A$ denotes the closure of $A$. But, first, let us state a lemma that establishes the continuity of the LASSO-based quantities on which $p_j^\lambda(\bm y, \bm X)$ depends, whose proof is deferred to after the end of this current proof.

    \begin{lemma}
    \label{lem:asymp_3}
        Consider any $\bm \xi^*$ with full column-rank $\bm X^*$. Then there is a neighborhood $V$ about $\bm \xi$, such that $\bm \xi \mapsto \hat \beta_j^\lambda$, $\bm \xi \mapsto \hat {\bm \beta}^\lambda_{-j}(|\hat \beta_j^\lambda|)$, and $\bm \xi \mapsto \hat {\bm \beta}^\lambda_{-j}(-|\hat \beta_j^\lambda|)$ are all vector-valued (and not set-valued) transformations in $V$ and continuous at $\bm \xi^*$.
    \end{lemma}

    \underline{Case I: $\bm \xi^*\in N_j$}. Recall that as mentioned earlier, we are only interested in the behavior of the transformation in $V$, in which $\bm \xi\mapsto p_j^\lambda(\bm y, \bm X)$ is a real-valued function. In this case, we have that $\bm \xi^*\in N_j\cap V$. We will first show that there is a neighborhood $W\subseteq N_j\cap V$ such that $\bm\xi^*\in W$ and for all $\bm \xi\in W$, the corresponding $\hat \beta_j^\lambda \neq 0$. Existence of such a $W$ ensures that the form of $p_j^\lambda(\bm y, \bm X)$ (that is, which of the two cases in Equation~\eqref{eqn:ell_pval_lambda} determines its form) does not change in a neighborhood around $\bm \xi^*$.
    To prove this, assume to the contrary that the statement is false. Then, we can find a sequence $\{\bm \xi_k\}$ such that $\bm \xi_k\notin N_j\cap V$ and $\bm \xi_k\to \bm \xi^*$. Because $V$ is open and contains $\bm \xi^*$, a tail of the sequence $\{\bm \xi_k\}$ is inside $V$, which implies that for $k$ larger than or equal to some $K$, $\bm \xi_k\in N_j^c\cap V$. If $\hat\beta_{j,k}^\lambda$ denotes the (uniquely defined) LASSO estimates obtained using $\bm \xi_k$ for $k\geq K$, then this implies $\hat \beta^\lambda_{j,k}=0$. $\bm \xi\mapsto \hat{\beta}_j^\lambda$ is continuous at $\bm \xi^*$ by Lemma~\ref{lem:asymp_3} and the $\hat \beta_{j,k}^\lambda$'s are uniquely defined for $k\geq K$, so we have that the LASSO estimate evaluated at $\bm \xi^*$ is the limit of $\hat\beta^\lambda_{j,k}$, as $k\to \infty$, which is equal to 0. This implies $\bm \xi^*\in N_j^c$, which is a contradiction.

    Because now we have established that there is a neighborhood $W\subseteq N_j\cap V$ containing $\bm \xi^*$, where the LASSO estimates are uniquely defined, and where the form of $p_j^\lambda(\bm y, \bm X)$ does not change, for the remainder of the proof of this case, without any loss of generality, we will restrict the transformation $\bm \xi\mapsto p_j^\lambda(\bm y, \bm X)$ to within $W$.
    
    Now, note that,
    \[
        \Lambda_j(b,\epsilon) = \frac{-\bm X_j^T(\bm P_{-j}\bm y - b\bm X_j - \bm X_{-j}\hat{\bm \beta}_{-j}^{\lambda}(b)) + n\lambda\epsilon}{\sigma \|(\bm I - \bm P_{-j})\bm X_j\|},\quad \epsilon \in \{-1,1\}.
    \]
     Note that an application of Lemma~\ref{lem:asymp_2} (and using the continuity of $\bm \xi\mapsto \hat \beta_j^\lambda$ at $\bm \xi^*$ via Lemma~\ref{lem:asymp_3}) shows that all the terms appearing in the expression of $\Lambda_j(|\hat\beta_j^\lambda|,1)$ are continuous in $\bm \xi$ at $\bm \xi^*$, and hence so is the transformation $\bm \xi\mapsto \Lambda_j(|\hat\beta_j^\lambda|,1)$. Thus, $\bm \xi \mapsto \Phi(\Lambda_j(|\hat \beta_j^\lambda|,1))$ is continuous at $\bm \xi^*$. Similar arguments establish that $\bm \xi\mapsto \Phi(-\Lambda_j(-|\hat \beta_j^{\lambda}|,-1)) + 1 - \Phi(\Lambda_j(|\hat \beta_j^{\lambda}|,1))$ is also continuous at $\bm \xi^*$. 

    \underline{Case II: $\bm \xi^*\in \mathrm{int}~N_j^c$}. Because by definition $\mathrm{int}~N_j^c$ is an open set, and so is $V$, there is an open ball $W'\subseteq V\cap \mathrm{int}~N_j^c$ such that $\bm \xi^*\in W'$. Now because $W'\subseteq \mathrm{int}~N_j^c\subseteq N_j^c$, by definition of $N_j^c$ for all $\bm \xi\in W'$, $\hat \beta_j^\lambda = 0$ and hence, the form of $p_j^\lambda(\bm y, \bm X)$ does not change in $W'$. As before, we will restrict our domain to $W'$ to prove the continuity at $\bm \xi^*$ in this case. Note that similar to Case I above, we can establish the continuity of $\bm \xi \mapsto \Lambda_j(0,\pm1)$ at $\bm \xi^*$ and hence that of $\bm \xi \mapsto \hat m_j = \frac{\Lambda_j(0,1)+\Lambda_j(0,-1)}{2}$ and $\bm \xi^*$. Also, we have that 
    \[
        z_1 = \frac{(\bm I - \bm P_{-j})\bm y}{\sigma\|(\bm I - \bm P_{-j})\bm X_j\|},
    \]
    so that $\bm \xi \mapsto z_1$ is also continuous at $\bm \xi^*$. This implies that $\bm \xi \mapsto \Phi(\hat m_j - |z_1 - \hat m_j|) + 1- \Phi(\hat m_j + |z_1 - \hat m_j|)$ is also continuous at $\bm \xi^*$. 

    \underline{Case III: $\bm \xi^* \in \del N_j=\del N_j^c$.}

    Because, $\bm \xi^*\in \del N_j^c\cap V=\del N_j\cap V$ , we have that $\bm \xi^*\in \del N_j$ and hence, is a limit point of $N_j$. Thus, we have some sequence $\{\bm \xi_k\}\subseteq N_j$ such that $\bm \xi_k\to \bm \xi^*$. Because $V$ is open and $\bm \xi^*\in V$, we know that a tail of the sequence $\{\bm \xi_k\}$ is eventually inside $V$. Without any loss of generality, we can assume that $\{\bm \xi_k\}\subseteq N_j\cap V$.
    
    By a similar argument, $\bm \xi^*$ is also a limit point of $N_j^c$. And thus there is some sequence $\{\bm \xi_k'\}\subseteq N_j^c\cap V$, such that $\bm \xi_k'\mapsto \bm \xi^*$. Note that by definition of $N_j^c$, the LASSO estimates evaluated at each of $\bm \xi_k'$ are all 0. Now let $\hat\beta^\lambda_{j*}$ denote the LASSO estimate evaluated at $\bm \xi^*$, we can invoke the continuity of the map $\bm \xi\mapsto \hat\beta_j^\lambda$ from Lemma~\ref{lem:asymp_3} to conclude that $\hat\beta^\lambda_{j*} = 0$, implying that $\bm \xi^*\in N_j^c\cap V$. Now for all $\bm \xi$ such that $\hat\beta_j$ is uniquely defined, define,
    \[
        \tilde p_j^\lambda(\bm y, \bm X) := \Phi(\hat m_j - |z_1 - \hat m_j|) + 1- \Phi(\hat m_j + |z_1-\hat m_j|),
    \]
    and note that when $\bm \xi\in N_j^c$, $p_j^\lambda(\bm y, \bm X) = \tilde p_j^\lambda(\bm y, \bm X)$. Lemma~\ref{lem:asymp_3} implies that $\bm \xi\mapsto \tilde p_j^\lambda(\bm y, \bm X)$ is continuous in $\{\bm \xi: \hat\beta_j\textrm{ is uniquely defined}\}$. Let us also use the notations $p_j(\bm \xi)=p_j^\lambda(\bm y, \bm X)$ and $\tilde p_j(\bm \xi)=\tilde p_j^\lambda(\bm y, \bm X)$ to streamline the notation.
    Then the above shows that,
    \begin{align}
    \label{eqn:outside_nj}
        \lim_{\bm \xi\to \bm \xi^*, \bm \xi\in N_j^c}~p_j(\bm \xi) =  \lim_{\bm \xi\to \bm \xi^*, \bm \xi\in N_j^c}~\tilde p_j(\bm \xi) = \tilde p_j(\bm \xi^*) =p_j(\bm \xi^*),
    \end{align}
    where recall, again, that the last equality follows because $\bm \xi^*\in N_j^c$. However, note that unlike case I and II, there is no neighborhood of $\bm \xi^*$ contained entirely in $N_j^c\cap V$. 

    Next, note again from the continuity of the LASSO maps established in lemma~\ref{lem:asymp_3} along with Lemma~\ref{lem:asymp_2}, we have that 

    \begin{align}
    \label{eqn:inside_nj}
    \begin{aligned}
        \underset{\bm \xi \rightarrow \bm \xi^*, \bm \xi \in N_j}{\lim}p_j(\bm \xi) &= \underset{\bm \xi \rightarrow \bm \xi^*, \bm \xi \in N_j}{\lim}\left[\Phi\left(-\Lambda_j(-|\hat \beta_j^\lambda|,-1)\right) + 1- \Phi(\Lambda_j(|\hat \beta_j^\lambda|,1))\right]\\
        &=  \Phi\left(\Lambda_j^*(0,-1)\right) + 1- \Phi(\Lambda_j^*(0,1)),
    \end{aligned}
    \end{align}
    where $\Lambda_j^*(0,\pm 1)$ are the values of $\Lambda_j(0,\pm 1)$ evaluated at $\bm \xi^*$. Thus, if we can just show that $p_j(\bm \xi^*) =  \Phi\left(\Lambda_j^*(0,-1)\right) + 1- \Phi(\Lambda_j^*(0,1))$ (implying that the limits of \eqref{eqn:outside_nj} and \eqref{eqn:inside_nj} match), we would establish that $\bm \xi\mapsto p_j^\lambda(\bm y, \bm X) = p_j(\bm \xi)$ is continuous at $\bm \xi^*$. Let $z_1^*$ denote $z_1$ computed using $\bm \xi^*$. Because $p_j^\lambda = p_j(\bm \xi) = \Phi(\hat m_j - |z_1 - \hat m_j|) + 1- \Phi(\hat m_j + |z_1 - \hat m_j|)$ for $\bm \xi \in N_j^c$, we will establish the claim above by showing that $z_1^*\in \{\Lambda_j^*(0,-1), \Lambda_j^*(0,+1)\}$, because note that if $z_1^*$ takes either of these values then (using the fact that at $\bm \xi^*$, $\hat m_j
    = \frac{\Lambda_j^*(0,-1)+\Lambda_j^*(0,1)}{2}$) we have at $\bm \xi^*$, $\Phi(\hat m_j - |z_1 - \hat m_j|) + 1- \Phi(\hat m_j + |z_1 - \hat m_j|)$ takes the value $ \Phi\left(\Lambda_j^*(0,-1)\right) + 1- \Phi(\Lambda_j^*(0,1))$.

    Now assume to the contrary. Because $\hat \beta_{j*}^\lambda=0 \Leftrightarrow z_1^*\in [\Lambda^*(0,-1), \Lambda^*(0,1)]$, we must have that there is some $a,b\in \mathbb R$ such that the following relation holds:
    \[
        \Lambda_j^*(0,-1)<a<z_1^*<b<\Lambda_j^*(0,1).
    \]
    Now because the map $\bm \xi \mapsto z_1$ is continuous, we can find a $\delta>0$ such that for any $\bm \xi \in B(\bm \xi^*;\delta)$ we have that $z_1 := z_1(\bm \xi)\in (a,b)$. Similarly, there is also another $\delta'>0$ such that for all $\bm \xi \in B(\bm \xi^*, \delta')\cap V$, we have that $\Lambda_j(0,-1)<a$ and $\Lambda_j(0,1)>b$ (this again follows from the continuity of the map $\bm \xi \mapsto \Lambda_j(0,\pm 1)$ in $V$). Thus, for $\delta'' = \min\{\delta, \delta'\}$ and $\bm \xi\in B(\bm \xi^*, \delta'')\cap V$, we have the relation
    \[
        \Lambda_j(0,-1)<z_1<\Lambda_j(0,1),
    \]
    and hence, it must be the case that: $\bm \xi\in B(\bm \xi^*, \delta'')\cap V\implies \hat \beta_j^\lambda = 0\implies \bm \xi \in N_j^c$. Thus, $B(\bm \xi^*, \delta'')\cap V\subseteq N_j^c$. Now because $V$ is open, we can again find a $\delta'''$ such that $B(\bm \xi^*, \delta ''')\subseteq B(\bm \xi^*, \delta'')\cap V\subseteq N_j^c\implies B(\bm \xi^*, \delta ''')\subseteq \mathrm{int}~N_j^c\implies \bm \xi^*\in \mathrm{int}~N_j^c$. But this is a contradiction since, $\bm \xi^*\in \partial N_j=\partial N_j^c$, and the boundary of a set is disjoint from its interior. This completes the proof.
\end{proof}
    
\begin{proof}[Proof of Lemma~\ref{lem:asymp_3}]
    The fact that $\bm \xi \mapsto \hat \beta_j^\lambda$ and $\bm \xi \mapsto \hat{\bm \beta}_{-j}^\lambda(|\hat\beta_j^\lambda|)$ are vector-valued functions in a neighborhood $V$ around $\bm \xi^*$ follows from a similar argument as presented in the second paragraph of the proof of Lemma~\ref{lem:asymp_4}. In the remainder of the proof, we will restrict $\bm\xi$ to the neighborhood $V$ without mentioning explicitly every time, since in order to establish the continuity of a function at a point, we are only interested in its behavior in a neighborhood of the point. We next state and prove a lemma that establishes the continuity of LASSO estimates.

\begin{lemma}
    \label{lem:asymp_1}
    Let $(\bm y, \bm X)\sim P$, a probability measure on $\mathcal X\times \mathcal Y\subseteq \mathbb R^{n\times p}\times \mathbb R^n$, and let $\lambda>0$. Then, $(\bm y, \bm X, \lambda)\mapsto \hat{\bm \beta}^\lambda$ is an upper-hemicontinuous transformation (in case the minimizer of the LASSO objective is not unique, we abuse notations and use this same $\hat{\bm \beta}^\lambda$ to denote the set of minimizers), where we endow the space $\bm y$ and $\lambda$ lie in with the usual Euclidean metric, while the space $\bm X$ lies in is endowed with the Frobenius metric.
    \end{lemma}
 
    \begin{proof}[Proof of Lemma~\ref{lem:asymp_1}]
    Define $\bm \xi = (\bm y,\bm X, \lambda)$ and  $f(\bm \beta; \bm \xi) = \frac{1}{2n}\left(\bm \beta^T\bm X^T\bm X\bm \beta - 2\bm \beta^T\bm X^T\bm y + \|\bm y\|^2\right) + \lambda \|\bm \beta\|_1$, and note that $f$ is continuous in ($\bm \beta$, $\bm \xi$) (satisfying property 1 of Lemma \ref{lem:minima_cont} with $\phi$ replaced by $f$, $\bm x$ replaced by $\bm \beta$ and $\bm y$ replaced by $\bm \xi$). Also note that we have 
    \[
        \hat{\bm \beta}^\lambda = \underset{\bm \beta}{\arg\min}f(\bm \beta;\bm \xi).
    \]
    Throughout this proof, we will use the notation $\bar B(\bm h, r)$ to denote a ball of radius $r$ in the metric space the object $\bm h$ lies in, without clarifying this explicitly each time. Now fix a point $\bm \xi^* = (\bm y^*, \bm X^*, \lambda^*)$. This implies that for any $d>0$ and all $\bm \xi = (\bm y, \bm X, \lambda)\in \bar B(\bm \xi^*, d)$, noting that $\lambda>0$, we have,
   \[
        \underset{\bm \xi \in \bar B(\bm \xi^*, d)}{\inf}~f(\bm \beta, \bm \xi)\geq \lambda \|\bm \beta\|_1\rightarrow \infty,\textrm{ as }\|\bm \beta\|\rightarrow \infty.
   \]
   This implies property 2 of Lemma \ref{lem:minima_cont}. We can now invoke Lemma \ref{lem:minima_cont} to conclude that $\bm \xi\mapsto \hat{\bm \beta}^\lambda$ is upper hemicontinuous.
    \end{proof}

    Now coming back to the proof of Lemma~\ref{lem:asymp_3}, from Lemma \ref{lem:asymp_1} and the fact that the maps in question are vector-valued, it follows that $\bm \xi \mapsto \hat{\beta}_j^\lambda$ and hence $\bm \xi \mapsto| \hat{\beta}_j^\lambda|$ are continuous transformations at $\bm \xi^*$. Now consider the transformation
    \begin{align*}
        (\bm z, \bm y, \bm X, \lambda)\mapsto (\bm y - \bm z, \bm X_{-j}, \lambda)\mapsto\hat{\bm \gamma}_{-j}(\bm z)
    \end{align*}
    where,
    \[
        \hat{\bm \gamma}_{-j}(\bm z):=\underset{\bm \gamma}{\arg\min}\left(\frac{1}{2n}\left(\bm \gamma^T\bm X_{-j}^T\bm X_{-j}\bm \gamma - 2(\bm y - \bm z)^T\bm X_{-j}\bm \gamma + \|\bm y - \bm z\|^2\right) + \lambda \|\bm \gamma\|_1 \right).
    \]
    Note that the proof of Lemma \ref{lem:asymp_1} can exactly be copied with $\bm y$ replaced with $\bm y - \bm z$, $\bm \beta$ replaced with $\bm \gamma$ and $\bm X$ replaced with $\bm X_{-j}$ (which is also full-column rank because so is $\bm X$) to conclude that $(\bm y- \bm z, \bm X_{-j}, \lambda)\mapsto \hat{\bm \gamma}_{-j}(\bm z)$ is a continuous transformation at any point $(\bm y^* - \bm z^*,\bm X^*, \lambda^*)$ with full-column rank $\bm X^*$. Finally, noting that $(\bm z, \bm y)\mapsto \bm y-\bm z$ is continuous and so is $\bm X\mapsto \bm X_{-j}$, we conclude that $(\bm z, \bm y, \bm X, \lambda)\mapsto (\bm y - \bm z, \bm X_{-j}, \lambda)$ is continuous. Lemma~\ref{lem:asymp_2} now gives  $(\bm z, \bm y, \bm X, \lambda)\mapsto \hat{\bm \gamma}_{-j}(\bm z)$ is a continuous transformation at any point $(\bm z^*,\bm y^*,\bm X^*, \lambda^*)$ with full-column rank $\bm X^*$. Lemma~\ref{lem:asymp_2} can again be invoked to get that $\bm \xi \mapsto \hat{\bm \gamma}_{-j}(\bm y - \bm X_j| \hat \beta_j^\lambda|) = \hat{\bm \beta}_{-j}^\lambda(|\hat \beta_j^\lambda|)$ is also continuous at $\bm \xi^*$, and so is the transformation $\hat{\bm \beta}_{-j}^\lambda(-|\hat \beta_j^\lambda|)$.
\end{proof}

\begin{proof}[Proof of Lemma~\ref{lem:asymp_5}]
    We proceed by first establishing a sequence of results.
    
    \underline{\textit{Claim 1:}} There exists a ball around $\bm X^*$ such that all the matrices in the ball satisfy item 2 of Assumption~\ref{assm:cv}.
    
    \underline{\textit{Proof of Claim 1}}. Without loss of generality, we will prove the following: Suppose $\tilde {\bm X}^*$ is a matrix formed by some $s\geq r$ rows of $\bm X^*$. We will use the notation $\tilde {\bm A}$ to denote the matrix formed by the same $s$ rows of some matrix $\bm A$. Then there exists a ball about $\bm X^*$, such for any matrix $\bm X$ in the ball, the columns of $\tilde{\bm X}$ are in general position. 
    Assume to the contrary. Then there exists a sequence of matrices $\{\bm X^{(i)}\}$ such that $\bm X^{(i)}\to \bm X$ and such that the columns of $\tilde{\bm X}^{(i)}$ are not in general position, for any $i$. This means that for every $i$, there is some set of $k'$ columns of $\tilde{\bm X}^{(i)}$ that lie in a $(k'-2)$-dimensional subspace, for $2\leq k'\leq s+1$. Because the number of such subsets is finite and there are infinitely many elements in $\{\tilde{\bm X}^{(i)}\}_i$, there is some set of columns of size $k$ (which without loss of generality we assume to be the first $k$ columns), and a subsequence $\{\tilde{\bm X}^{(i_l)}\}\subseteq \{\tilde{\bm X}^{(i)}\}$, such that $\{\tilde{\bm X}^{(i_l)}_1,\dots, \tilde{\bm X}^{(i_l)}_k\}$ lie in a $(k-2)$-dimensional affine subspace, for all $l$. Now, define a matrix, $\bm W^{(i_l)} := \left[\left(\tilde{\bm X}^{(i_l)}_1-\tilde{\bm X}^{(i_l)}_k\right), \cdots, \left(\tilde{\bm X}^{(i_l)}_{k-1}-\tilde{\bm X}^{(i_l)}_k\right)\right]$ and $\bm W^*$ similarly with $\tilde{\bm X}^{(i_l)}$ replaced by $\tilde{\bm X}^*$. Because the first $k$ columns of $\tilde{\bm X}^*$ are in general position, we have that the columns of ${\bm W}^*$ are linearly independent. Since this is not the case for $\tilde{\bm X}^{(i_l)}$, the columns of  $\bm W^{(i_l)}$ are not linearly independent. Observe that $\tilde{\bm X}\mapsto \bm W$ is a continuous transformation and hence, $\bm W^{(i_l)}\to \bm W^*$, as $l\to \infty$. 
    Again note the sequence of matrix transformations: $\bm A \mapsto \bm A^T\bm A\mapsto \mathrm{det}(\bm A^T\bm A)$ is continuous, and hence, we must have $\mathrm{det}\left({\bm W^{(i_l)}}^T\bm W^{(i_l)} \right)\to \mathrm{det}\left({\bm W^*}^T{\bm W^*}\right)$, as $l\rightarrow \infty$. But note that $\mathrm{det}\left({\bm W^{(i_l)}}^T\bm W^{(i_l)} \right) = 0,\forall l$ as the columns of $\bm W^{(i_l)}$ are not linearly independent, and hence their limit $\mathrm{det}\left({\bm W^*}^T{\bm W^*}\right)$ must also be 0. But this is a contradiction since the columns of $\bm W^*$ are linearly independent. This proves Claim 1.

    Now, let for a vector $\bm v\in \mathbb R^{n'}$ and a set $S\subseteq [1:n']$, $\bm v_S$ denote the sub-vector formed by the entries of $\bm v$ corresponding to the indices in $S$, while for a matrix $\bm A$, let $\bm A_{S,\cdot}$ denote the sub-matrix formed by the rows corresponding to the indices in $S$. Now the cross-validated error for a general $(\bm y, \bm X)$ is given by:
    \[
        \psi(\bm y, \bm X,\lambda) = \sum_{i=1}^m \left\|\left(\bm P_{-j}\bm y +\sigma \bm V\tilde{\bm z}\right)_{S_i} - \bm X_{S_i, \cdot}\hat{\bm \beta}^\lambda_{[1:n']\setminus S_i}\right\|^2,
    \]
    where $\hat{\bm \beta}^\lambda_{[1:n']\setminus S_i}$ is \textit{any} LASSO estimate obtained using the response $\left(\bm P_{-j}\bm y +\sigma \bm V\tilde{\bm z}\right)_{[1:n']\setminus S_i}$ and design matrix $\bm X_{[1:n']\setminus S_i, \cdot}$ with regularizer $\lambda$.

    Note that because $\bm X^*$ satisfies item 2 of Assumption~\ref{assm:cv}, \citet[Lemma 3]{lassouniqueness} implies that LASSO estimates obtained on all the folds of cross-validation using $\bm X^*$ are all uniquely defined, which in turn implies that the cross validated errors are uniquely defined. Now, let us define $\psi_j(\bm y, \bm X):=\psi(\bm y, \bm X,  l_j\left(\bm P_{-j}\bm y +\sigma \bm V\tilde{\bm z},\bm X)\right)$, the cross-validated error from the $j^{\mathrm{th}}$ candidate $\lambda$. Now, because item 4(a) of Assumption~\ref{assm:cv} is satisfied and based on the theorem statement, $\psi_j(\bm y^*, \bm X^*)$ are distinct across different $j$'s, without loss of generality, assume that
    \begin{equation}
        \label{eqn:psi_ordering}
        \psi_1(\bm y^*, \bm X^*)<\cdots<\psi_t(\bm y^*, \bm X^*),
    \end{equation}
    where $t$ is the output dimension of $\bm l$.
    Then, we make the following claim: 
    
    \underline{\textit{Claim 2:}} Let $(\bm y^{(i)}, \bm X^{(i)})\mapsto (\bm y^*, \bm X^*)$. Then, for large $i$, the relation
    \[
        \psi_1(\bm y^{(i)}, \bm X^{(i)})<\cdots<\psi_t(\bm y^{(i)}, \bm X^{(i)}),
    \]
    holds along with these quantities being well-defined.\\

    \underline{\textit{Proof of Claim 2}}. By Claim 1, there is a ball about $\bm X^*$ such that any matrix in that ball satisfies item 1 of Assumption~\ref{assm:cv}. Because $(\bm y^{(i)}, \bm X^{(i)})\mapsto (\bm y^*, \bm X^*)$, there exists some $i_0$ such that for all $i\geq i_0$, $\bm X^{(i)}$ also satisfies item 1 of Assumption~\ref{assm:cv}. \citet[Lemma 3]{lassouniqueness} then implies that the LASSO estimate computed using any $s$-subset of the rows of $\bm X^{(i)}$, where $s\geq r$, is unique. In particular this implies that the LASSO estimates obtained on various folds while doing cross-validation using $\bm X^{(i)}$ are all uniquely defined, and hence, so are $\psi\left(\bm y^{(i)},\bm X^{(i)},\lambda\right),\forall i\geq i_0$. From here on, we will focus on $i\geq i_0$. 
    
    By Lemma~\ref{lem:asymp_1}~and~\ref{lem:asymp_2}, note that $(\bm y,\bm X)\mapsto \psi(\bm y, \bm X, \lambda)$ is continuous at $(\bm y^*, \bm X^*)$, which immediately implies that
    \[
        \psi_j\left(\bm y^{(i)},\bm X^{(i)}\right)\to  \psi_j\left(\bm y^{*},\bm X^{*}\right), \forall j\in [1:t],i\to\infty,
    \]
    from which Claim 2 follows.

    Now, defining
    \[
        \hat{\lambda}(\bm y) := \underset{\lambda \in \bm l(\bm y, \bm X)}{\arg\min}~\psi(\bm y, \bm X, \lambda),
    \]
    we see that the cross-validated $\lambda$ is given by $\hat \lambda(\bm P_{-j}\bm y + \sigma \bm V\tilde{\bm z})$. Let us define $\hat\lambda^*$ as the cross-validated $\lambda$ corresponding to $(\bm y^*,\bm X^*)$, and from Equation~\eqref{eqn:psi_ordering}, it follows that $\hat\lambda^* = l_1\left(\bm P_{-j}^*\bm y^* +\sigma\bm V^*\tilde{\bm z},\bm X^*\right)$. Similarly, define $\hat\lambda_i = \hat\lambda\left(\bm P^{(i)}_{-j}\bm y^{(i)} + \sigma \bm V^{(i)}\tilde{\bm z}\right)$, where as before, $\bm P_{-j}^{(i)}$ and $\bm V^{(i)}$ have been computed using $\bm X^{(i)}$. We want to show that $\hat{\lambda}_i\to \hat\lambda^*$.

    From Claim 2, we know that $\hat\lambda_i = l_1\left(\bm P^{(i)}_{-j}\bm y^{(i)} + \sigma \bm V^{(i)}\tilde{\bm z},\bm X^{(i)}\right)+o(1)$. The convergence of $\{\hat{\lambda}_i\}$ now follows from the continuity of the maps $(\bm y, \bm X)\mapsto \left(\bm P_{-j}\bm y + \sigma \bm V\tilde{\bm z},\bm X\right)\mapsto \bm l\left(\bm P_{-j}\bm y + \sigma \bm V\tilde{\bm z},\bm X\right)\mapsto l_1\left(\bm P_{-j}\bm y + \sigma \bm V\tilde{\bm z},\bm X\right)$.
\end{proof}

\begin{proof}[Proof of Lemma~\ref{lem:asymp_6}]
    First, note that because of the independence, $(W_n,U)\dto (W,U)$. Now note that
    \begin{align*}
        &\mathbb P((W,U)\in E)=1\\
        \Leftrightarrow& \mathbb E\left[\mathbb P((W,U)\in E\mid U)\right]=1\\
        \Leftrightarrow &\int \mathbb P((W,u)\in E\mid U=u)\mathbb P(du)=1\\
        \Leftrightarrow&\mathbb P\left((W,U)\in E\mid U\right) = 1\textrm{, almost surely}.\\
        \Leftrightarrow& \mathbb P\left(f_U\textrm{ is continuous at }W\mid U\right)=1\textrm{, almost surely}.
    \end{align*}
    Now for any function $h$, let us introduce the notation $\mathcal C_h$ to denote the set of continuity points of $h$. Then, the above suggests that
    \[
        \mu_W\left(\mathcal C_{f_U}\right)=1,\textrm{almost surely }[\mu_U].
    \]
    Now, define $\Omega = \left\{u:\mu_W\left(\mathcal C_{f_u}\right)=1\right\}$ and note that the above equation suggests $\mu_U\left(\Omega\right)=1$. Now because $W_n\dto W$, by the continuous mapping theorem, we have,
    \begin{align*}
        \Omega =&\left\{u:\mu_W\left(\mathcal C_{f_u}\right)=1\right\}\subseteq \left\{u:f_u(W_n)\dto f_u(W)\right\}\\
        \implies & \mu_U\left(\left\{u:f_u(W_n)\dto f_u(W)\right\}\right)=1\\
        \Longleftrightarrow & \mu_U\left(\left\{u:\mathbb E[g\circ f(W_n,u)]\to \mathbb E[g\circ f(W,u)],\forall\textrm{ bounded continuous }g\right\}\right)=1,
    \end{align*}
    where the last implication follows from Portmanteau's Theorem \citep{billingsley1999convergence}. Hence, for one particular bounded, continuous function $g$,
    \begin{align*}
        &\mu_U\left(\left\{u:\mathbb E[g\circ f(W_n,u)]\to \mathbb E[g\circ f(W,u)]\right\}\right)=1\\
        \implies & \mathbb E[g\circ f(W_n,U)\mid U]\to \mathbb E[g\circ f(W,U)\mid U],\textrm{ almost surely }[\mu_U].
    \end{align*}
    Because $g$ is bounded, the bounded convergence theorem implies that,
    \[
        \mathbb E[g\circ f(W_n,U)]\to \mathbb E[g\circ f(W,U)].
    \]
    Because the above holds for any arbitrary bounded, continuous $g$, Portmanteau Theorem suggests that $f(W_n,U)\dto f(W,U)$.
\end{proof}

Finally, as an intermediate result, we have also established that the known-$\sigma$-$\ell$-test $p$-value evaluated using a fixed, known $\lambda$ is also asymptotically valid for a class of locally asymptotically Gaussian linear models. We formally summarize this as a corollary below.
\begin{corollary}
    \label{cor:validity_known_sigma_ell_test}
   Let $p^{\lambda}_{j,n}$ denote the known-$\sigma$-$\ell$-test $p$-value for testing $H_j$ using response $\sqrt{n}\hat {\bm \Sigma}_n^{-1/2}\hat{\bm \theta}_n$($=:\bm y^{(n)}$), design matrix $\hat{\bm \Sigma}_n^{-1/2}$($=:\bm X^{(n)}$), and $\lambda$ as the regularizer for the LASSO. Then
     $p^{\lambda}_{j,n}$ is asymptotically valid.
\end{corollary}
\begin{proof}
    Note that as a consequence of Lemma~\ref{lem:asymp_4}, $(\bm y, \bm X)\mapsto p_j^\lambda(\bm y, \bm X)$ is continuous at $(\bm y^*, \bm X^*)$ almost surely, under the law of $\bm y^*$. Because $(\bm y^{(n)}, \bm X^{(n)})\dto (\bm y^*, \bm X^*)$, and $\bm y^*\mid \bm X^*\sim \mathcal N(\bm X^*\bm \theta, \bm I)$, the continuous mapping theorem (and the fact that $p_j^\lambda(\bm y, \bm X)$ is known-$\sigma$-$\ell$-test $p$-value) yields that $\{p_{j,n}^{\lambda}\}_n$ is an asymptotically valid sequence of $p$-values.
\end{proof}

\subsection{Theory from Section~\ref{sec:power_analysis}}
\label{sec:proof_thm_power}
\begin{proof}[Proof of Theorem~\ref{thm:power}]
    The $p$-value for the re-centered test is given by (for an independently drawn $\tilde z_1\sim N(0,1)$):
\begin{align*}
\begin{aligned}
    \repval &= \mathbb P(|\tilde z_1 - \hat m_j'|\geq |z_1 - \hat m_j'|\mid \bm S^{(j)}, \bm X)\\
    &= \begin{cases}
        1-\Xi(z_1) + \Xi(2\hat m_j' - z_1), &\textrm{if }z_1>\hat m_j'\\
        \Xi(z_1) + 1-\Phi(2\hat m_j' - z_1), &\textrm{if }z_1\leq \hat m_j'
    \end{cases}\\
    &= 1 - \left[\left(\Phi(z_1) - \Phi(2\hat m_j' - z_1)\right)\right]\textrm{sign}(z_1-\hat m_j'),
\end{aligned}
\end{align*}
where, $z_1 = \frac{\bm X_j^T(\bm I - \bm P_{-j})\bm y}{\sigma\|(\bm I - \bm P_{-j})\bm X_j\|}$ and $\hat m_j' = \frac{-\bm X_j^T(\bm  P_{-j}\bm y - \bm X_{-j}\hat{\bm \beta}_{-j}^{\lambda}(0))}{\sigma\|(\bm I - \bm P_{-j})\bm X_j\|}$. 

Also, recall that in the known-$\sigma$ case, we abuse notations and also use $\bm S^{(j)}=\bm X_{-j}^T\bm y$ to denote the sufficient statistic under $H_j$. Next, we have,
\begin{align*}
    \|(\bm I - \bm P_{-j})\bm X_j\|^2
    = \bm X_j^T(\bm I - \bm P_{-j})\bm X_j
    = \bm X_j^T\bm \Gamma^T\bm D \bm \Gamma\bm X_j
    = \bm W_j^T \bm D \bm W_j,
\end{align*}
where $(\bm I - \bm P_{-j}) = \bm \Gamma^T\bm D \bm \Gamma$ is its spectral decomposition with an orthogonal $\bm \Gamma$ and a diagonal matrix $\bm \Sigma$ of eigen-values, and we define $\bm W_j:=\bm \Gamma\bm X_j$. Under Setting 1, the entries of $\bm X$ are drawn i.i.d. from $\mathcal N(0,1/n)$, and hence due to the independence between $\bm X_j$ and $\bm X_{-j}$, the orthogonal matrix $\bm \Gamma$ is independent of $\bm X_j$ and hence, $\bm W_j$ has entries i.i.d from $N(0,1/n)$. Note that $\bm D$ has exactly $n-d$ many 1's in its diagonal entries, and the rest are 0, so that $\bm W_j^T\bm D \bm W_j$ is the sum of squares of $n-d$ i.i.d $N(0,1/n)$'s. Hence, we have that 
\begin{align}\label{eqn:1-kappalimit}
    \|(\bm I - \bm P_{-j})\bm X_j\|^2 = \frac{1}{n}\sum_{i=1}^{n-d} Q_i^2 = \frac{n-d}{n}\frac{1}{n-d}\sum_{i=1}^{n-d} Q_i^2\Pto (1-\kappa) \cdot1 = 1-\kappa,
\end{align}
where, $Q_i\iid N(0,1)$. Now, let $F_n:=\frac{\bm X_j^T(\bm I - \bm P_{-j})\bm y}{\sigma\sqrt{1-\kappa}}$, while $G_n:=\frac{\bm X_j^T(\bm y - \bm X_{-j}\hat{\bm \beta}_{-j}^\lambda(0))}{\sigma\sqrt{1-\kappa}}$ (curiously, $G_n$ is the \textit{distilled-LASSO statistic} discussed in \citet{crt_power}), then note that we have the following relations:
\begin{align}
\label{eqn:decompositions}
\begin{aligned}
    z_1 &= \frac{F_n\sqrt{1-\kappa}}{ \|(\bm I - \bm P_{-j})\bm X_j\|};\\
    2\hat m_j' - z_1&=\frac{-2\bm X_j^T(\bm P_{-j}\bm y  - \bm X_{-j}\hat{\bm \beta}^\lambda_{-j}(0)) - \bm X_j^T(\bm I - \bm P_{-j})\bm y}{\sigma \|(\bm I - \bm P_{-j})\bm X_j\|}\\
    &=\frac{-2\bm X_j^T(\bm P_{-j}\bm y  - \bm X_{-j}\hat{\bm \beta}^\lambda_{-j}(0)) -2 \bm X_j^T(\bm I - \bm P_{-j})\bm y + \bm X_j^T(\bm I - \bm P_{-j})\bm y}{\sigma \|(\bm I - \bm P_{-j})\bm X_j\|}\\
    &= \frac{-2\left[\bm X_j^T(\bm P_{-j}\bm y  - \bm X_{-j}\hat{\bm \beta}^\lambda_{-j}(0)) + \bm X_j^T(\bm I - \bm P_{-j})\bm y \right] +  \bm X_j^T(\bm I - \bm P_{-j})\bm y}{\sigma \|(\bm I - \bm P_{-j})\bm X_j\|}\\
    &=\frac{-2\bm X_j^T(\bm y  - \bm X_{-j}\hat{\bm \beta}^\lambda_{-j}(0)) + \bm X_j^T(\bm I - \bm P_{-j})\bm y}{\sigma \|(\bm I - \bm P_{-j})\bm X_j\|}\\
    &= \frac{(F_n - 2G_n)\sqrt{1-\kappa}}{\|(\bm I - \bm P_{-j})\bm X_j\|};\\
    z_1-\hat m'_j &= \frac{\bm X_j^T(\bm I - \bm P_{-j})\bm y + \bm X_j^T(\bm P_{-j}\bm y  - \bm X_{-j}\hat{\bm \beta}^\lambda_{-j}(0))}{\sigma \|(\bm I - \bm P_{-j})\bm X_j\|}\\
    &= \frac{\bm X_j^T(\bm y  - \bm X_{-j}\hat{\bm \beta}^\lambda_{-j}(0))}{\sigma \|(\bm I - \bm P_{-j})\bm X_j\|}\\
    &= \frac{G_n\sqrt{1-\kappa}}{\|(\bm I - \bm P_{-j})\bm X_j\|}.
\end{aligned}
\end{align}
Thus, $\repval$ can be written as
\begin{align}
\label{eqn:r_pval}
    1-\left[\Phi\left(\frac{F_n\sqrt{1-\kappa}}{\|(\bm I - \bm P_{-j})\bm X_j\|} \right) - \Phi\left(\frac{(F_n - 2G_n)\sqrt{1-\kappa}}{ \|(\bm I - \bm P_{-j})\bm X_j\|}\right)\right]\mathrm{sign}(G_n).
\end{align}

Thus, to characterize the asymptotic distribution of $\repval$, we first characterize the asymptotic distribution of $(F_n,G_n)$. To simplify notation, we borrow notation from \citet{crt_power} and re-define as follows (note that this notation is only valid for the rest of this proof and is not to be confused with the notation used elsewhere in the paper): $\bm X:=\bm X_j, \bm Z:=\bm X_{-j}, \hat{\bm \theta} := \hat{\bm \beta}_{-j}^\lambda(0),\bm \theta:=\bm \beta_{-j},\bm \epsilon':=\bm y - \bm Z\bm \theta, \beta := h$. Then we have that $(F_n,G_n) = (\bm X^T(\bm I - \bm P_{\bm Z})\bm y, \bm X^T(\bm y - \bm Z\hat{\bm \theta}))/(\sigma\sqrt{1-\kappa})$. Now note that following the third equation from the bottom of the proof of \citet[Lemma 9]{crt_power}, we have that
\begin{align}
\label{eqn:x_normal_dist}
    \bm X\mid \bm y, \bm Z\sim \mathcal N\left(\frac{\beta}{n\sigma^2+\beta^2}\bm \epsilon', \frac{\sigma^2}{n\sigma^2 + \beta^2}\bm I\right).
\end{align}

Observe that we can write $F_n = \bm X^T\tilde{\bm F}$ and $G_n = \bm X^T\tilde{\bm G}$, where $(\tilde{\bm F},\tilde{\bm G})$ are functions of $(\bm y, \bm Z)$. Then, we have that 
\begin{align*}
    (F_n,G_n)\mid \bm y, \bm Z\sim \mathcal N\left(\bm \mu_n, \bm \Sigma_n\right);\bm \mu_n = \begin{pmatrix}
        \mu_{1,n}\\
        \mu_{2,n}
    \end{pmatrix}\textrm{ and }\bm \Sigma_n = \begin{bmatrix}
        \Sigma_{11,n} & \Sigma_{12,n}\\
        \Sigma_{12,n} & \Sigma_{22,n},
    \end{bmatrix},
\end{align*}
where
\begin{align*}
    \mu_{1,n}= \frac{\beta}{(n\sigma^2 + \beta^2)}\frac{\bm y^T(\bm I - \bm P_{Z})\bm \epsilon'}{\sigma\sqrt{1-\kappa}}=\frac{\beta}{(\sigma^2+\beta^2/n)}\frac{\bm y^T(\bm I - \bm P_{\bm Z})\bm \epsilon'}{n\sigma\sqrt{1-\kappa}}= \frac{\beta}{(\sigma^2+\beta^2/n)}\frac{\bm \epsilon'^T(\bm I - \bm P_{\bm Z})\bm \epsilon'}{n\sigma\sqrt{1-\kappa}},
\end{align*}
where the last equality follows because $\bm y = \bm Z\bm \theta + \bm \epsilon'$. Note that $\bm \epsilon'$ has entries i.i.d. from $\mathcal N(0, \sigma^2+\beta^2/n)$, and 
since $\frac{\bm \epsilon'^T(\bm I - \bm P_{\bm Z})\bm \epsilon'}{n} \stackrel{d}{=} \left(\sigma^2+\beta^2/n\right)\|(\bm I - \bm P_{\bm Z})\bm X\|^2$,
we have that $\frac{\bm \epsilon'^T(\bm I - \bm P_{\bm Z})\bm \epsilon'}{n}\Pto (1-\kappa)\sigma^2 $ by Equation~\eqref{eqn:1-kappalimit}. Then we have that
\[
    \mu_{1,n}\Pto \frac{\beta\sigma^2(1-\kappa)}{\sigma^3\sqrt{1-\kappa}} = \frac{\beta\sqrt{1-\kappa }}{\sigma}.
\]
Next, note that we have
\begin{align*}
    \mu_{2,n} &= \mathbb E\left[G_n\middle|\bm y, \bm Z\right]\\
    &= \mathbb E\left[\bm X^T\left(\bm y - \bm Z\hat{\bm \theta}\right)\middle|\bm y, \bm Z\right]/\left(\sigma\sqrt{1-\kappa}\right)\\
    &= \frac{\beta}{n\sigma^2 + \beta^2}\frac{(\bm y - \bm Z\hat{\bm \theta})^T\bm \epsilon'}{\sigma\sqrt{1-\kappa}}~\textrm{[Equation~\eqref{eqn:x_normal_dist}]}\\
    &= \frac{\beta}{n\sigma^2 + \beta^2}\frac{(\bm y - \bm Z\hat{\bm \theta})^T(\bm y - \bm Z{\bm \theta})}{\sigma\sqrt{1-\kappa}}\\
    &=\frac{\beta}{\sigma^2 + \beta^2/n}\frac{(\bm y - \bm Z\hat{\bm \theta})^T(\bm y - \bm Z{\bm \theta})}{n\sigma \sqrt{1-\kappa}}.
\end{align*}

Using Equation (14) of the statement of \citet[Lemma 9]{crt_power}, we have that
\[
    \mu_{2,n}\asto \frac{\beta\lambda}{\alpha_\lambda \tau_\lambda\sigma\sqrt{1-\kappa}},
\]
where $\alpha_{\lambda}$ and $\tau_\lambda$ are solutions to the AMP equations. Next,
\begin{align*}
    \Sigma_{11,n} &= \frac{\sigma^2}{\sigma^2 + \beta^2/n}\frac{\|\tilde{\bm F}\|^2/n}{\sigma^2(1-\kappa)} =  \frac{\sigma^2}{\sigma^2 + \beta^2/n}\frac{\|(\bm I - \bm P_{\bm Z})\bm y\|^2/n}{\sigma^2(1-\kappa)} =  \frac{\sigma^2}{\sigma^2 + \beta^2/n}\frac{\|(\bm I - \bm P_{\bm Z})\bm \epsilon'\|^2/n}{\sigma^2(1-\kappa)}\\
    &\Pto 1 .
\end{align*}
\begin{align*}
    \Sigma_{12,n} &= \frac{\sigma^2}{\sigma^2 + \beta^2/n}\tilde{\bm F}^T\tilde{\bm H}/n =  \frac{\sigma^2}{\sigma^2 + \beta^2/n}\frac{\bm y^T(\bm I - \bm P_{\bm Z})(\bm y - \bm Z\hat{\bm \theta})/n}{\sigma^2(1-\kappa)} = \frac{\sigma^2}{\sigma^2 + \beta^2/n}\frac{\|(\bm I - \bm P_{\bm Z})\bm y\|^2/n}{\sigma^2(1-\kappa)}\\
    &\Pto 1.
\end{align*}
\begin{align*}
    \Sigma_{22,n} &= \mathrm{Var}\left(G_n\middle|\bm y, \bm Z\right)\\
    &= \frac{\sigma^2}{n\sigma^2+\beta^2}\frac{\|\bm y - \bm Z\hat{\bm \theta}\|^2}{\sigma^2(1-\kappa)}~\textrm{[Equation~\eqref{eqn:x_normal_dist}]}\\
    &= \frac{\sigma^2}{\sigma^2+\beta^2/n}\frac{\|\bm y - \bm Z\hat{\bm \theta}\|^2}{n\sigma^2(1-\kappa)}.
\end{align*}

Then, from Equation (13) of the statement of \citet[Lemma 9]{crt_power}, we have that $\Sigma_{22,n}\asto \frac{\lambda^2}{\alpha_\lambda^2\sigma^2(1-\kappa)}$.

Now, let $\bm C_n := (F_n,G_n)^T$ and let $\psi_{\bm C_n}$ denote its characteristic function (CHF) given by
\[
    \psi_{\bm C_n}(\bm t) = \mathbb E\left[\exp\left(i\bm t^T\bm C_n\right)\right],
\]
where $i = \sqrt{-1}$. Then we have, using the expression for the CHF of a normally distributed random variable,
\begin{align*}
    \psi_{\bm C_n}(\bm t) &= \mathbb E\mathbb E\left[\exp\left(i\bm t^T\bm C_n\right)\middle| \bm y, \bm Z\right] = \mathbb E\left[\exp\left(i\bm t^T\bm \mu_n - \frac{1}{2}\bm t^T\bm \Sigma_n\bm t\right)\right]\\
    &=\mathbb E\left[\exp\left(-\frac{1}{2}\bm t^T\bm \Sigma_n\bm t\right)\left(\cos(\bm t^T\bm \mu_n) + i\sin(\bm t^T\bm \mu_n)\right)\right]\\
    &=\mathbb E\left[\exp\left(-\frac{1}{2}\bm t^T\bm \Sigma_n\bm t\right)\cos(\bm t^T\bm \mu_n)\right] + i \mathbb E\left[\exp\left(-\frac{1}{2}\bm t^T\bm \Sigma_n\bm t\right)\sin(\bm t^T\bm \mu_n)\right].
\end{align*}
Now let $\bm\mu\in \mathbb R^2$ and $\bm \Sigma\in \mathbb R^{2\times 2}$ denote the in-probability limits of $\bm \mu_n$ and $\bm \Sigma_n$, respectively. Then, an application of bounded convergence theorem implies that
\begin{align*}
    &(0,1)\ni \mathbb E\left[\exp\left(-\frac{1}{2}\bm t^T\bm \Sigma_n\bm t\right)\cos(\bm t^T\bm \mu_n)\right]\to \mathbb E\left[\exp\left(-\frac{1}{2}\bm t^T\bm \Sigma\bm t\right)\cos(\bm t^T\bm \mu)\right]\\
    &(0,1)\ni \mathbb E\left[\exp\left(-\frac{1}{2}\bm t^T\bm \Sigma_n\bm t\right)\sin(\bm t^T\bm \mu_n)\right]\to \mathbb E\left[\exp\left(-\frac{1}{2}\bm t^T\bm \Sigma\bm t\right)\sin(\bm t^T\bm \mu)\right],
\end{align*}
thereby showing that for all $\bm t \in \mathbb R^2$,
\begin{align*}
    \psi_{\bm C_n}(\bm t)\to \exp\left(i\bm t^T\bm \mu -\frac{1}{2}\bm t^T\bm \Sigma \bm t\right),\textrm{ as }n\to \infty.
\end{align*}
As pointwise convergence of CHF implies convergence in distribution, the above implies that
\begin{align*}
    \begin{pmatrix}
        F_n\\
        G_n
    \end{pmatrix}\dto  \begin{pmatrix}
        F\\
        G
    \end{pmatrix}\sim \mathcal N \left(\begin{pmatrix}
            \frac{h\sqrt{1-\kappa}}{\sigma}\\
            \frac{h\lambda}{ \alpha_\lambda \tau_\lambda\sigma\sqrt{1-\kappa}}
        \end{pmatrix}, 
        \begin{bmatrix}
            1 & 1\\
           1 & \frac{\lambda^2}{\alpha_\lambda^2\sigma^2(1-\kappa)}
        \end{bmatrix}\right).
\end{align*}
Next, observe from Equation~\eqref{eqn:r_pval}, the transformation $(F_n,G_n)^T\mapsto \repval$ is continuous except for $(F_n, G_n)\in \{(f,g):g=0\}$, which precisely arises from the discontinuity of the sign function at $b=0$. This set has measure 0 under the joint normal distribution of $(F,G)^T$. Hence, by the continuous mapping theorem, along with applying Slutsky's lemma and the facts that $\|(\bm I - \bm P_{-j})\bm X_j\|\Pto (1-\kappa)$ and $\beta = h$ establishes the in-distribution limit of $\repval$ claimed in Theorem~\ref{thm:power}. 
\end{proof}

Similar ideas can be used to derive the asymptotic distributions of  $p_{j,n}^{z,1}$ and $p_{j,n}^{z,2}$ as well. Observe from the relations derived in Equation~\eqref{eqn:decompositions} that
\begin{align}
    \label{eqn:z_pval}
     p_{j,n}^{z,1} = 1-\Phi\left(z_1\right) = 1-\Phi\left(\frac{F_n\sqrt{1-\kappa}}{\|(\bm I - \bm P_{-j})\bm X_j\|}\right).
\end{align}
Using the facts that $F_n\dto F$ and $\|(\boldsymbol I - \boldsymbol P_{-j})\boldsymbol X_j\|^2\Pto 1-\kappa$ (Equation~\eqref{eqn:1-kappalimit}), an application of Slutsky's theorem yields, $\frac{F_n\sqrt{1-\kappa}}{\|(\bm I - \bm P_{-j})\bm X_j\|}\dto F$. The continuity of $\Phi(\cdot)$ and the continuous mapping theorem then yields
\[
    p_{j,n}^{z,1}\dto 1-\Phi(F).
\]
Similarly, noting that the two-sided $z$-test $p$-value can be written as
\[
    p_{j,n}^{z,2} = 1-[\Phi(z_1) - \Phi(-z_1)]\mathrm{sign}(z_1)
\]
Similarly as above (and noting that the sign function is only discontinuous at 0, a zero measure set under $N(0,1)$), we can argue that $p_{j,n}^{z,2}\dto 1-[\Phi(F) - \Phi(-F)]\mathrm{sign}(F)$. This completes the proof of the asymptotic distributions of all the $p$-values introduced in Section~\ref{sec:power_analysis}.
\color{black}

\color{black}
\section{Further investigation of power results from Section~\ref{sec:power_analysis}}
\label{sec:app_power_analysis}

\subsection{Achieving one-sided behavior}
\label{sec:power_behavior}
Recall that in Section~\ref{sec:power_analysis}, we derived the asymptotic distribution of the $p$-value of the re-centered test as
\[
    1-\left[\Phi\left(F \right) - \Phi\left(F-2G\right)\right]\mathrm{sign}(G),
\]
while that of the one-sided $z$-test as $1-\Phi(F)$, with the distribution of $(F,G)$ given in Equation~\eqref{eqn:ab_dist}. Observe that the closeness between these two $p$-values depends on $\mathrm{sign}(G)$ taking +1 with high probability and $J = 2G-F$ being highly positive with high probability. To study the first phenomenon, observe that
\begin{align*}
    &\mathbb P(\mathrm{sign}(G) = 1) = \mathbb P(G>0)=\mathbb P\left(\frac{G - \frac{h\lambda}{\sigma\alpha_\lambda \tau_\lambda\sqrt{1-\kappa}}}{\frac{\lambda}{\sigma\alpha_\lambda\sqrt{1-\kappa}}}>-\frac{h}{\tau_\lambda}\right)=\Phi\left(\frac{h}{\tau_\lambda}\right).
\end{align*}
This shows that the probability of $\mathrm{sign}(G)$ attaining $+1$ increases with the standardized mean of $G$, $h/\tau_\lambda$, which increases with signal size ($h$) and LASSO performance ($1/\tau_\lambda$). To understand this behavior further, we will consider a distribution $B_0 = (1-\pi)\delta_{\{0\}} + \pi\delta_{\{h\}}$ and plot this standardized mean of $G$ across  varying values of $h$, $\kappa$ and $\pi$, each time keeping the other two parameters fixed. 
To understand the behavior of $J$, we look at its standardized and un-standardized means for different varying values of $h,\kappa$ and $\pi$, as described above. The plots of the standardized means of $G$ and $J$ are in Figure~\ref{fig:stand_mean}, while the un-standardized mean of $J$ is in Figure~\ref{fig:unstand_mean}.

\begin{figure}[h]
    \centering
\includegraphics{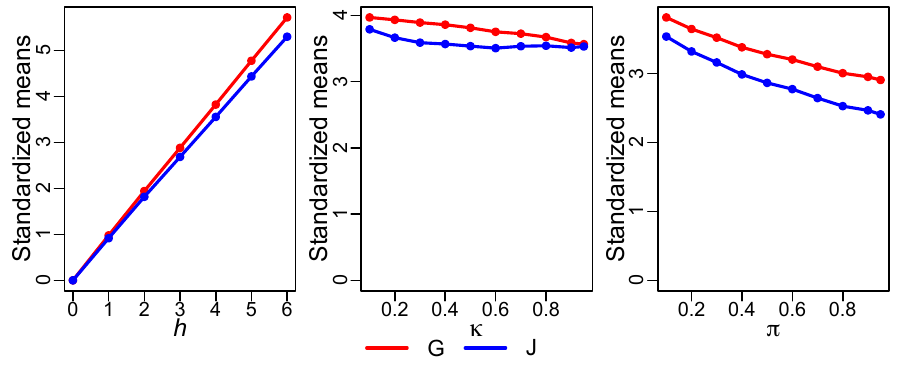}
    \caption{Plot of the standardized means of $G$ and $J$. For the left, we fix $\kappa = 0.5, \pi = 0.1$ and vary $h$. For the center, we fix $h = 3, \pi = 0.1$ and vary $\kappa$. For the right, we fix $h = 3, \kappa = 0.5$ and vary $\pi$.}
    \label{fig:stand_mean}
\end{figure}

\begin{figure}[h]
    \centering
\includegraphics{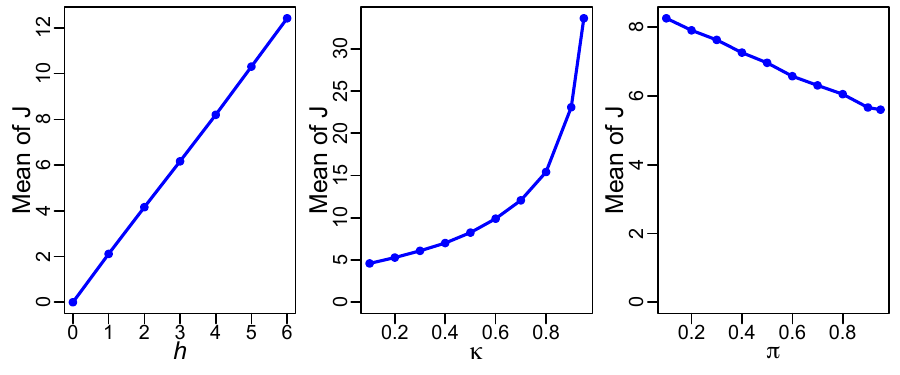}
    \caption{Plot of the (un-standardized) mean of $J$. The parameter specifications are exactly the same as in Figure~\ref{fig:stand_mean}.}
    \label{fig:unstand_mean}
\end{figure}

From Figure~\ref{fig:stand_mean}, as expected, with increasing $h$, the standardized mean of $G$ increases, showing a greater probability of $\mathrm{sign}(G)$ attaining $+1$. We see that an increase in $\kappa$ is associated with a decrease in this value, which can be explained by the poor LASSO performance (resulting in larger $\tau_\lambda$) as the dimension increases. A similar decrease is observed when the sparsity goes down (i.e., $\pi$ increases). However, note that in the latter two cases (and even when $\pi = 0.95$), the standardized mean lies above 2, which indicates less than $3\%$ chance of $\mathrm{sign}(G)$ being $-1$. Similar trends are also observed for the standardized mean of $J$.

The un-standardized mean of $J$ plotted in Figure~\ref{fig:unstand_mean} also shows similar trends, except for the plot where $\kappa$ varies. In this case, we see that the un-standardized mean increases with $\kappa$, which can be explained by the $\sqrt{1-\kappa}$ in the denominator of its expression (which gets removed once $J$ is standardized, as $\sqrt{1-\kappa}$ also appears in the denominator of the expression for the standard deviation of $J$), which gets smaller as $\kappa$ increases. Note that the mean of $J$ is approximately always above $5$ in the middle and right plot. Thus, for an appreciable signal size ($h\geq 3$) and across all the settings, the fact that the mean of $J$ lies above $5$ as well as its standardized mean is above $2$ (follows from Figure~\ref{fig:stand_mean}) shows that $J$ mostly puts its mass far into the positive half of the real line. This, combined with our understanding of $\mathrm{sign}(G)$'s behavior above, shows that with appreciable signal size, and across a large range of choices of $\kappa$ and $\pi$, the re-centered test performs close to the one-sided $z$-test.

\subsection{Power curves of the re-centered $z_1$-based-test}

\label{sec:power_curves}
In this section, we compare the limiting power curve of the re-centered test, obtained in Section~\ref{sec:power_analysis} under Setting~\ref{setting:amp}, along with the one- and two-sided $z$-tests for the following two candidates for $B_0$ (that is, the distribution from which the rest of the coefficients are drawn):
\begin{enumerate}
    \item $B_0 = \mu_1(\pi,h): = \pi\delta_{(0)} + (1-\pi)\delta_{(h)}$, where $h$ is the signal of the coefficient we are testing; here, $\delta_{(x)}$ represents a distribution putting all its mass at $x$.
    \item $B_0 = \mu_2 := \mathcal N(0,1)$.
\end{enumerate}
Observe that the distribution $\mu_1(\pi,h)$ on average sets a fraction $1-\pi$ of the coefficients to 0, while the rest have the same amplitude as the coefficient for which we are testing. So $1-\pi$ plays the role of sparsity. However, $\mu_2$ does not introduce any true sparsity in the coefficients. We compare the power curves for the three methods using $B_0=\mu_1(\pi,h)$ for different values of $\pi$ and $h$ in Figure \ref{fig:power_ber}, while those for $B_0=\mu_2$ are presented in Figure \ref{fig:power_normal}. For empirical validation, we also plot the average empirical power curves of the re-centered test and the known-$\sigma$-$\ell$-test using $n=1000$. To make these two curves (which use cross-validation to choose $\lambda$) comparable with our theoretical power curve (which takes $\lambda$ as fixed), we set $\lambda$ for our each point on our theoretical power curve by averaging the cross-validated $\lambda$ choice over 1000 data sets generated from the appropriate distribution.


\begin{figure}[h]
    \centering      \includegraphics{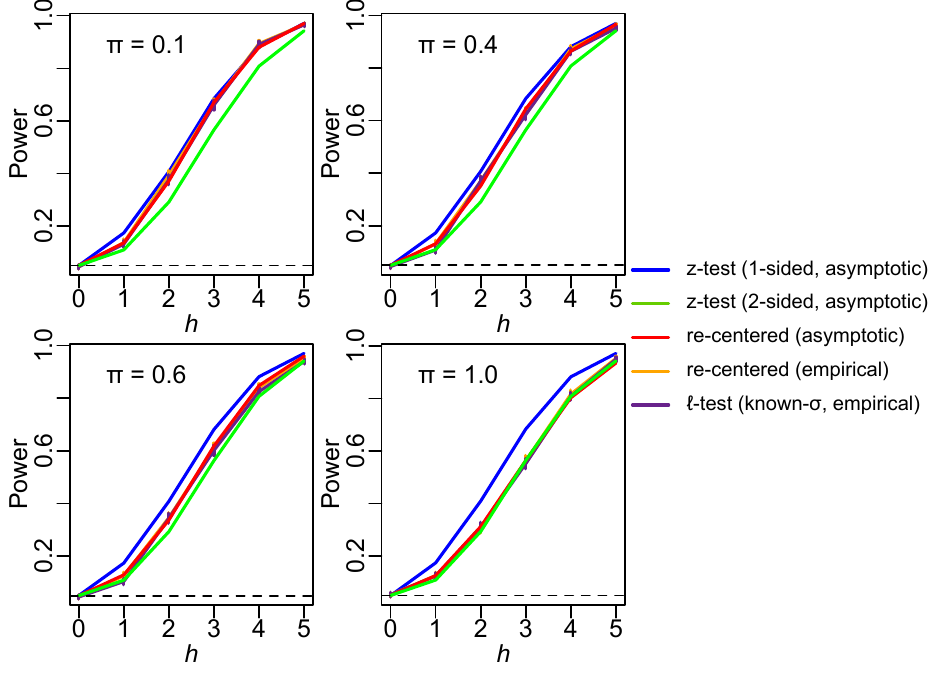}
    \caption{\color{black}Various power curves for $B_0 = \mu(\pi,h)$, with $\pi$ varying in $\{0.1,0.4,0.6,0.9\}$ across the four panels and $h$ varying in the x-axis in each of them. We used $\kappa = 0.5$ and $\sigma = 1$. 
    The error bars represent plus or minus 2 standard errors.\color{black}}
    \label{fig:power_ber}
\end{figure}

\begin{figure}[h]
    \centering      \includegraphics{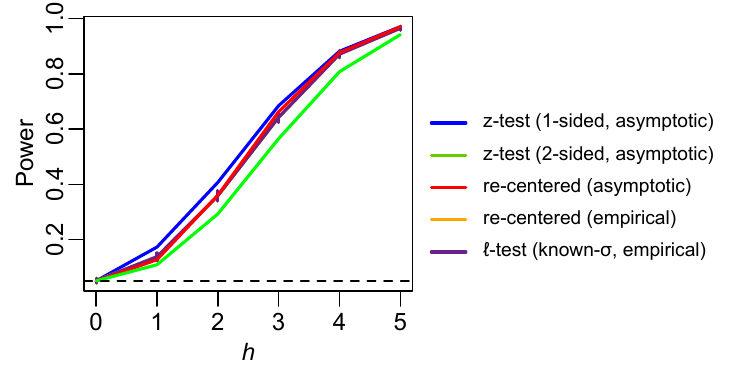}
    \caption{\color{black}Various power curves for $B_0 = \mathcal N(0,1)$, with $h$ varying in the x-axis. We used $\kappa = 0.5$ and $\sigma = 1$. 
    The error bars represent plus or minus two standard errors.\color{black}}
    \label{fig:power_normal}
\end{figure}

From both Figures \ref{fig:power_ber} and \ref{fig:power_normal}, we observe excellent agreement between the empirical and theoretical power-curves for the re-centered test. We also see that the empirical power curve for the (known-$\sigma$) $\ell$-test is extremely close to that of the re-centered test, highlighting that the two tests perform very similarly. For Figure \ref{fig:power_ber}, we see that between $\pi = 0.1$ to $\pi = 1.0$, the power of the $\ell$- and re-centered tests starts to move from being closer to the one-sided $z$-test power curve to that of its two-sided counterpart, demonstrating the importance of sparsity for the power benefit of re-centering. Notably, even with $100\%$ of the $\bm \beta$ vector populated (with signals equally strong as the coefficient we are trying to test for), the $\ell$- and re-centered tests perform indistinguishably from the two-sided $z$-test, showing that even without sparsity, at least asymptotically, our approach gives up essentially nothing compared with the standard two-sided test.

This resilience to lack of sparsity is further corroborated by Figure \ref{fig:power_normal}, where the entries of $\bm \beta_{-j}$ are drawn from $\mathcal N(0,1)$. Here, there is no exact sparsity (i.e., none of the coefficients are exactly 0), even though one can treat these coefficients as approximately or `weakly' sparse, as some of them are closer to zero than others. Even in this case, the $\ell$- and re-centered tests exhibit power significantly higher than that of the one-sided $z$-test.
\color{black}

\section{Characterization of $\hat{\bm \beta}^{\lambda}_{-j}(b)$ as a function of $b$}
\label{sec:app_beta_path}

\begin{algorithm}[h]
	\caption{Piecewise linear characterization of $\hat{\bm \beta}^{\lambda}(b)$ (as defined in \eqref{eqn:app_beta_x})}
	\label{alg:app_betapath}
    \textbf{Input:} Data: $(\bm y, \bm v, \bm X)$, Regularization parameter: $\lambda$.\\
    \textbf{Output: }Piecewise linear characterization of $\hat{\bm \beta}^{\lambda}(b)$.\\
    Find a point a point $x_0$ such that $\hat{\bm \beta}^{\lambda}$ is differentiable at $x_0$ and set $b=x_0$ \\
    Calculate $\bm \gamma(x_0)$ and set $\bm \gamma = \bm\gamma(x_0)$.\\
    \While{Not stopped}{
        Find $d_1:=\min\left\{d>0:\left\lvert\bm Z_i^T\left(\bm y - \bm v b - \bm Z (\hat{\bm \beta}^{\lambda}(b) + d \bm \gamma )\right)\right\rvert = n\lambda\right\}$.\\
        Find $d_2 := \min\left\{d>0: \textrm{ at least one coordinate of }\hat{\bm \beta}^{\lambda}_{S(b)}+d\bm \gamma_{S(b)}\textrm{ is }0\right\}$.\\
        \If {the minima could not be found in the above two steps}{
            Label $b$ as a terminal knot\\
            Record $\bm \gamma(b)$ and tag it with $b$.\\
            Break the While loop.\\
        } 
        Set $d \leftarrow \min\{d_1,d_2\}$.\\
        Set $\hat{\bm \beta}^{\lambda}(b+d) = \hat {\bm \beta}^{\lambda}(b)+d\bm\gamma$.\\
        Set $b = b+d$\\
        Set $\bm \gamma = \bm \gamma (b)$
        Label $b$ a non-terminal knot.\\
        Record $(b,\hat{\bm \beta}^\lambda(b))$.
    }
    Repeat from Step 4, but with $\bm \gamma = -\bm \gamma(x_0)$ instead.\\
    Generate piecewise linear paths by linearly interpolating $\hat{\bm \beta}^\lambda(b)$ between two non-terminal knots or a terminal and non-terminal knot. For a terminal knot, $b_t$, towards the side where there is no other knot, generate the linear path with slope $\bm \gamma(b_t)$ originating at $(b_t,\hat{\bm \beta}^\lambda(b_t))$.
\end{algorithm}

In Sections \ref{sec:ciforbeta} and \ref{sec:postselection}, we saw that for obtaining the $\ell$-test confidence intervals for $\beta_j$, one needs to evaluate the function $\hat{\bm \beta}^{\lambda}_{-j}(b)$ for different values of $b$. Recall from \eqref{eqn:betax} that $\hat{\bm \beta}^{\lambda}_{-j}(b)$ is defined as
\[
     \hat{\bm \beta}_{-j}^{\lambda}(b):= \underset{\bm{\beta}_{-j}\in \mathbb R^{d-1}}{\arg\min}\left(\frac{1}{2n}\|\bm y -b \bm X_j - \bm X_{-j}\bm {\beta}_{-j}\|^2 + \lambda \|\bm {\beta}_{-j}\|_{1}\right).
\]
Certainly, evaluating $\hat{\bm \beta}^{\lambda}_{-j}(b)$ for different values of $b$ is computationally expensive as each evaluation requires a LASSO run. In this section, we show that this computation burden can be relieved significantly by providing a characterization of $\hat{\bm \beta}^{\lambda}_{-j}(b)$. 

We will ease notations and for this section we will assume that we have data of the form $(\bm y, \bm v, \bm Z)$ and define,
\begin{equation}
\label{eqn:app_beta_x}
    \hat{\bm \beta}^{\lambda}(b):= \underset{\bm{\beta}}{\arg\min}\left(\frac{1}{2n}\|\bm y -b \bm v - \bm Z\bm {\beta}\|^2 + \lambda \|\bm {\beta}\|_{1}\right).
\end{equation}
The above notation is valid in this section only and any mention of $\hat{\bm \beta}^\lambda(b)$ would always imply the above. We will characterize $\hat{\bm \beta}^{\lambda}(b)$ as a function of $b\in \mathbb R$. Note that one can obtain $\hat{\bm \beta}^\lambda_{-j}(b)$ (as defined in \eqref{eqn:betax}) by substituting $\bm Z = \bm X_{-j}$ and $\bm v = \bm X_{j}$.

In the following proposition, we first show that $\hat{\bm \beta}^{\lambda}(b)$ is a piecewise linear function of $b$, in which we take heavy inspiration from \citet{piecewiselinear} where the authors establish that the optimizers of $\ell_1$-penalized, twice-differentiable likelihoods are piecewise linear in the regularization parameter $\lambda$ and provide algorithm for generating these paths.
\begin{prop}
    \label{prop:app_betapath}
    Assume that $\bm y,\bm v \in \mathbb R^n$ and $\bm Z \in \mathbb R^{n\times d}$, for $n>d$. Then, for $\hat{\bm \beta}^{\lambda}(b)$ defined as in \eqref{eqn:app_beta_x}, we have that 
    $\hat{\bm \beta}^{\lambda}(b)$ is a continuous, piecewise linear function of $b$.
\end{prop}
\begin{proof}[Proof sketch]
    \textcolor{black}{First note that we can use arguments similar to the ones presented in Lemma~\ref{lem:asymp_3} to conclude that $b\mapsto \hat{\bm \beta}^{\lambda}(b)$ is a continuous transformation}. Now for any $b\in \mathbb R$, define,
    \begin{equation}
    \label{eqn:app_sx}
        S(b) = \left\{j\in [1:d]: \left(\hat{\bm \beta}^{\lambda}(b)\right)_j \neq 0\right\}.
    \end{equation}
    Note that $S(b)$ is constant in a neighborhood $b\in \mathcal N$, if and only if $\hat{\bm \beta}_{S(b)}^{\lambda}(b)$ is differentiable in $b$. Pick such a $b$, then it follows that
    \begin{equation}
    \label{eqn:app_beta_der}
        \frac{\partial }{\partial b} \hat{\bm \beta}_{S(b)}^\lambda (b)= -\left(\bm Z_{S(b)}^T\bm Z_{S(b)}\right)^{-1}\bm Z_{S(b)}^T \bm v
    \end{equation}
    Thus this shows that  $\hat{\bm \beta}^{\lambda}_{S(b)}(b)$ , and hence,  $\hat{\bm \beta}^{\lambda}(b)$ is linear in that neighborhood. And hence, $\hat{\bm \beta}^{\lambda}(b)$ is piecewise linear for $b\in \mathbb R$.
\end{proof}

We now turn our attention to methods of exactly generating these piecewise linear paths, following an approach analogous to \citet{piecewiselinear}. To begin with, denote,
\[
     L(\bm \beta):= \frac{1}{2n}\|\bm y - \bm v b - \bm Z\bm \beta\|^2
\]
Then, note that we can decompose $\bm \beta = \bm \beta^+ - \bm \beta^-$ to write our LASSO minimization problem as,
\begin{align*}
    &\textrm{Minimize: }L(\bm \beta^+ - \bm \beta ^-) + \lambda \sum (\beta_i^+ - \beta_i^-)\\
    &\textrm{Subject to: }\beta_i^+, \beta_i^-\geq 0 , \forall i.
\end{align*}
We introduce Lagarange multipliers for each of these $2p$ elements to get the Lagarangian
\[
    \mathcal L = L(\bm \beta^+ - \bm \beta ^-) + \lambda \sum ( \beta_i^+ + \beta_i^-) - \sum \lambda_i^+\beta_i^+ -\sum \lambda_i^-\beta_i^-.
\]
KKT conditions on the primal show that the following relations are suggested at the optimum for all $1\leq i\leq d$:
\begin{align*}
    &(\partial_{\bm \beta} L(\bm \beta))_i + \lambda - \lambda_i^+ = 0\\
    -&(\partial_{\bm \beta} L(\bm \beta))_i + \lambda - \lambda_i^- = 0\\
    &\lambda_i^+\beta_i^+ = 0\\
    & \lambda_i^-\beta_i^- = 0
\end{align*}
The above set of equations suggest that,
\begin{equation}
\label{eqn:app_beta0cond}
\hat{\beta}^\lambda_i(b) \neq 0 \Leftrightarrow \left\lvert\left(\partial_{\bm \beta} L\big\vert_{\bm \beta = \hat{\bm \beta}^\lambda(b)}\right)_i\right\rvert = \lambda \Leftrightarrow \left\lvert\bm Z_i^T\left(\bm y - \bm v b - \bm Z \hat{\bm \beta}^\lambda(b)\right)\right\rvert = n\lambda.
\end{equation}
We use Equation \eqref{eqn:app_beta_der}, \eqref{eqn:app_beta0cond}, along with Proposition \ref{prop:app_betapath} and the fact that the set $S(b)$ (from \eqref{eqn:app_sx}) changes only at non-differentiability points of $\hat{\bm \beta}(b)$ to devise an algorithm to exactly generate the paths of $\hat {\bm \beta}(b)$ as a function of $b$ in  Algorithm \ref{alg:app_betapath}. The algorithm uses a notation defined below for the `derivative' of $\hat{\bm \beta}^\lambda(b)$:
\begin{align*}
    \bm \gamma(b) :=& \begin{cases}
    &\bm \gamma_{S(b)}(b) = -\left(\bm Z_{S(b)}^T\bm Z_{S(b)}\right)^{-1}\bm Z_{S(b)}^T \bm v\\
    &\bm \gamma_{S(b)^c}(b) = \bm 0
    \end{cases}, \textrm{ if }S(b) \neq \phi\\
    :=& \bm 0, \textrm{ otherwise}.
\end{align*}


\section{Choice of the tuning parameter, $\lambda$}
\label{sec:app_lambda_choice}
\subsection{The min rule vs. The 1se rule}
\label{sec:app_min_vs_1se}

In Section \ref{sec:lambda_choice_singletest}, we justified that a reasonable way of choosing $\lambda$ for testing $H_j$ so that it does not invalidate our theory surrounding the $\ell$-test can be to cross validate on $(\tilde{\bm y}, \bm X_{-j})$, where $\tilde{\bm y}$ is drawn from the conditional distribution of $\bm y\mid \bm S^{(j)}$, under $H_j$. For computational convenience, we have used cross-validation on $(\bm X_{-j},\bm{\tilde y})$ instead of $(\bm X,\bm{\tilde y})$, as the two choices result in almost identical performances of the resulting methods and because the LASSO estimates obtained using $(\bm X_{-j},\bm{ y})$ is the same as that using $(\bm X_{-j},\bm{\tilde y})$, this enables us to recycle some common information from the latter dataset, thereby saving computation time while obtaining the conditional distribution of $\hat\beta_j\mid \bm S^{(j)}$, under $H_{j}$.

Throughout this paper, we recommend using the min rule for choosing $\lambda$ using cross-validation---that is, choosing the $\lambda$ that results in the smallest cross-validated error. However another popular choice with cross-validation can be to pick the largest $\lambda$ resulting in a cross-validated error within one standard deviation of the minimum cross-validated error, also known as the 1se rule. The latter rule results in stricter selection, but more severe multiplicity correction post-selection, and hence, it is not entirely clear whether the trade-off that the 1se rule presents can be any better than the min rule. 

In Section \ref{sec:app_ci_with_1se}, we provide the empirical coverage and lengths of the resulting confidence intervals when using the 1se rule for selecting $\lambda$. To summarize our findings, we observe that the 1se rule and the min rule have similar performances for constructing $\ell$-test confidence intervals except for the case when the $\bm \beta$ is very dense, in which case the min rule results in intervals of shorter length. Hence, we recommend using the min rule for choosing $\lambda$.

\subsection{The randomness in the choice of $\lambda$} 
\label{sec:app_single_test_randomness_lambda}

The rule for the choice of $\lambda$ we discussed involves sampling, $\tilde{\bm y}\sim \bm y\mid \bm S^{(j)}$, under $H_{j}$ and then running cross-validation on $(\bm X_{-j}, \tilde{\bm y})$. This suggests towards some inherent sources of variability in the method---the random sampling of $\tilde{\bm y}$ and the random splitting of the dataset to perform cross-validation. In order to understand its effect, we perform $m=100$ replications of an experiment under the setting of the left panel of Figure \ref{fig:unconditional_test} where for each replicate we form a linear model with this design matrix and obtain $p$-value for testing $H_{j}:\beta_j = 0$ for $m=100$ samples of $\tilde{\bm y}$. Let $p_{ij}$ denote the $p$-value after sampling $\tilde{\bm y}$ for the $j^{\mathrm{th}}$ time for the dataset $(\bm y^{(i)}, \bm X^{(i)})$ (corresponding to the $i^{\mathrm{th}}$ replication of drawing $(\bm y, \bm X)$ from a pre-determined distribution). For these $m^2 = 10^4$ $p$-values, we plot the empirical estimate of the overall standard deviation of the $p$-values, $\sqrt{\mathrm{Var}(p_{ij})}$ against the standard deviation conditioned on a replicate, $\sqrt{\mathbb E\mathrm{Var}\left(p_{ij}\mid (\bm y^{(i)}, \bm X^{(i)})\right)}$ in Figure \ref{fig:variability_single} in the log-scale. These quantities are estimated using $\sqrt{\frac{1}{m^2-1}\left(\sum_{i,j=1}^m (p_{ij}-\bar p_{\cdot\cdot})^2\right)}$ and $\sqrt{\frac{1}{m^2-m}\left(\sum_{i,j=1}^m(p_{ij}-\bar p_{i\cdot})^2\right)}$, respectively, where, $\bar p_{i\cdot}$ represents the mean of the entries in $\{p_{ij}:1\leq j\leq m\}$, while $\bar p_{\cdot\cdot}$ represents the mean of all the $p$-values, $\{p_{ij}:1\leq i,j\leq m\}$.
\begin{figure}[H]
    \centering
    \includegraphics{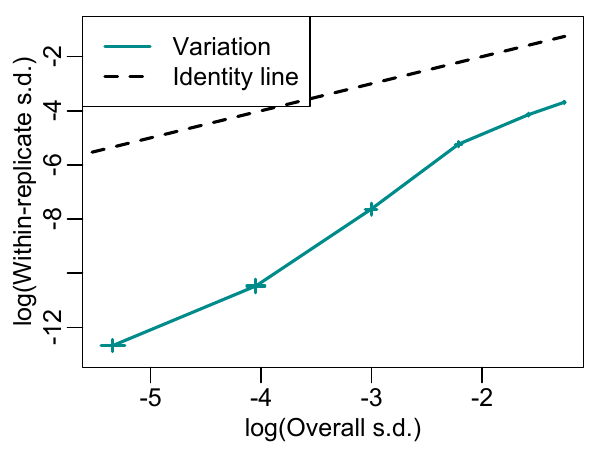}
    \caption{Relative variability in the $\ell$-test $p$-values due to the randomness in $\hat{\lambda}_j$. The error bars represent plus or minus two units of standard errors.}
    \label{fig:variability_single}
\end{figure}
This figure suggests that the variability due to sampling of $\tilde{\bm y}$ is negligible compared to the monte-carlo variability in the $p$-values due to the replication of the procedure---we see from Figure \ref{fig:variability_single} that sampling of $\tilde{\bm y}$ never accounts for more than about $9\%$ of the overall standard deviation (or about $0.8\%$ of the overall variation). This suggests that even though randomized due to sampling of $\tilde {\bm y}$, the $p$-values we obtain are stable.

Finally, note that an alternate way of doing cross-validation that does not introduce the randomness in the procedure due to sampling $\tilde{\bm y}$ is by cross-validating on $(\hat{\bm y}_j, \bm X_{-j})$ instead, where $\hat{\bm y}_j$ is the projection of $\bm y$ on the columnspace of $\bm X_{-j}$ (and is the non-zero-mean component of $\tilde{\bm y}$, when sampled from $\bm y\mid \bm S^{(j)}$ under $H_j$). Note that computation of $\ell$-distribution based on either of $(\hat{\bm y}_j, \bm X_{-j})$ or $(\tilde{\bm y}, \bm X_{-j})$ would exactly be the same. Even though we have established that the sampling of $\tilde{\bm y}$ introduces negligible randomness in the $\ell$-test $p$-value, one might wonder why introduce \emph{any} randomness at all in the first place and not cross-validate on $(\hat{\bm y}_j, \bm X_{-j})$? In Figure \ref{fig:yhat_ytilde}, we compare three possible datasets we can cross-validate on to choose $\lambda$: $(\tilde{\bm y},\bm X_{-j})$ (our default choice), $(\hat{\bm y}_j, \bm X_{-j})$ and $(\bm y, \bm X)$. Note that, as described in Section \ref{sec:lambda_choice_singletest}, the last choice is not valid as the chosen $\lambda$ will not be a function of the sufficient statistic, however as this is cross-validating on the full dataset, we can expect this chosen $\lambda$ to have the `optimal performance' and the resulting power curve can be used as a benchmark. Indeed, Figure \ref{fig:yhat_ytilde} shows that choosing $\lambda$ based on $(\hat{\bm y}_j, \bm X_{-j})$ suffers a detriment as compared to our recommended choice (which also performs almost similarly to cross-validating on the full $(\bm y, \bm X)$), providing further justification for it.

\begin{figure}[h]
    \centering
\includegraphics[height=2.7in, width = 3.6in]{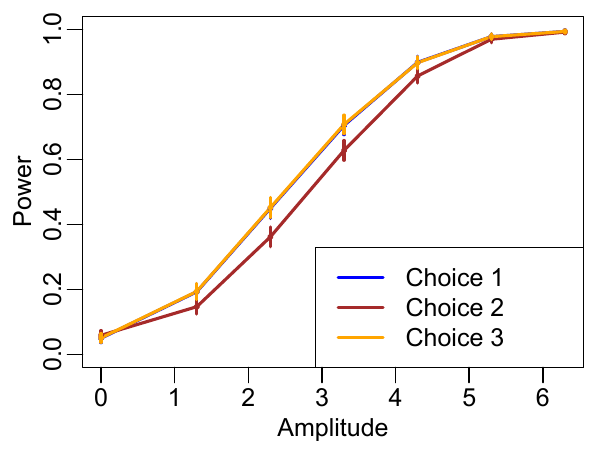}
    \caption{We use exactly the same setting as in the left panel of Figure \ref{fig:unconditional_test}. Choice 1,2 and 3 denote choosing $\lambda$ by cross-validating on $(\tilde{\bm y}, \bm X_{-j})$, $(\hat{\bm y}_j, \bm X_{-j})$ and $(\bm y, \bm X)$, respectively. The curves for Choice 1 and 3 are on top of each other. The error bars represent plus or minus two standard errors.}
    \label{fig:yhat_ytilde}
\end{figure}

\section{Further Experiments}

\subsection{Performance of the $\ell$-test under different settings}
\label{sec:app_ltest_further_expt}
To explore further aspects of the performance of the $\ell$-test, we test its performance under an additional setting, with the results are summarized in Figure \ref{fig:unconditional_test_extra}. We see that in this case, similar to Figure \ref{fig:unconditional_test}, the power of the $\ell$-test increases with increasing amplitude, but almost overlapping with that of the one-sided $t$-test. Note that, as follows from Appendix \ref{sec:app_improved_performance}, in this case where our particular choice of the design matrix is closer to un-identifiability, the $\ell$-test is more sure about its guess of the sign of the alternate $\beta_j$ as compared to the $d=50$ case. 

\begin{figure}[H]
    \centering
\includegraphics{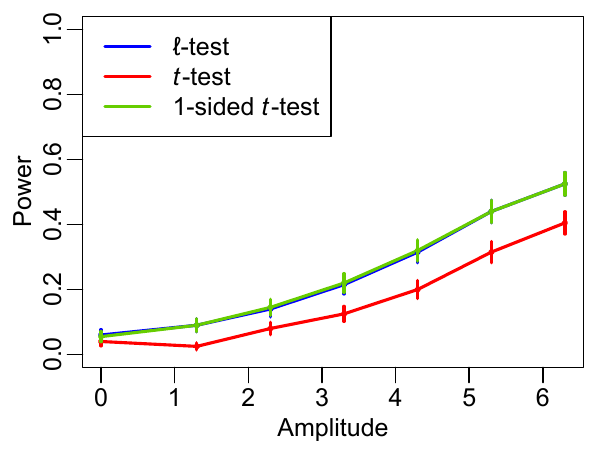}
     \caption{Exactly the same setting as in the left panel of Figure \ref{fig:unconditional_test} but with $d=90$. The error bars represent plus or minus two standard errors.}
     \label{fig:unconditional_test_extra}
\end{figure}

\subsection{Robustness of the $\ell$-test to violations of the linear model assumptions}
\label{sec:app_robustness}
In this section, we extend the results in Section \ref{sec:robust} by performing more extensive experiments to empirically evaluate the robustness of the validity of the $\ell$-test to the violations in the assumptions of the Gaussian linear model. \textcolor{black}{Aside from the details of the model violations described below, the setup is identical to that of Section \ref{sec:robust}.}
\begin{itemize}
    \item \textbf{Setting 1 (Violation of the Gaussianity of errors---effect of heavy tails): }For each specific value of $(n,d)$, we draw i.i.d. errors from a $t$-distribution with $\nu$ degrees of freedom. We vary $\nu$ between 30 and 2. For $\nu>2$, we also standardize the mean zero errors with the standard deviation of the $t_{\nu}$ distribution. We do not do this for $t_{2}$ as it does not have a finite second moment. As $\nu$ varies from 30 to 2, the tails of $t_{\nu}$ gets fatter as compared to the normal distribution. \color{black}We summarize our results in Figure \ref{fig:t_robustness} (corresponding to $\rho = 0$) and Figure \ref{fig:t_robustness_corr} (corresponding to $\rho = 0.5$).\color{black}
    \item \textbf{Setting 2 (Violation of the Gaussianity of errors---effect of skewness): }We consider a setup similar to that in Setting 1 but instead consider Gamma distributed error with scale parameter 1 and shape parameter, $\alpha$. We vary $\alpha$ between 1 and 10 and for each error draw, we standardize the error with the mean and standard deviation of the Gamma$(1,\alpha)$ distribution. This error distribution is asymmetric for smaller values of $\alpha$ and moves towards symmetry as the value of $\alpha$ increases. \color{black}We summarize our results in Figure \ref{fig:gamma_robustness} (corresponding to $\rho = 0$) and Figure \ref{fig:gamma_robustness_corr} (corresponding to $\rho = 0.5$).\color{black}
    \item \textbf{Setting 3 (Violation of homoskedasticity of error): }\color{black}We again consider a similar setup as above, but change the error distribution as follows: For design matrix, $\bm X$, let $r_i$ denote mean of the $i^{\mathrm{th}}$ row of the design matrix. We generate the error vector $\bm \epsilon = (\epsilon_1,\dots, \epsilon_n)^T$, based on the following rule:
    \begin{align}
    \label{eqn:hetero_scheme1}
         \epsilon_i \iid \begin{cases}
            \mathcal{N}(0,1),&\textrm{ if }r_i\leq0\\
            \mathcal{N}(0,\eta^2),&\textrm{ if }r_i> 0
        \end{cases}.
    \end{align}
    where $\eta^2>0$ is a quantity, specified by us, that controls the heteroskedasticity in the error term. The scheme \eqref{eqn:hetero_scheme1} represents a simple way of imparting heteroskedasticity, where the error variance of a row depends on the mean value of its entries. 
    \color{black}For $\eta = 1$, the distribution is homoskedastic while becomes heteroskedastic for larger and smaller values of $\eta$. We vary $\eta^2$ in the set $\{0.01,0.25,0.5,1,4,8\}$ and compare the performance of the two tests for each of these values. \color{black}The results are summarized in Figure \ref{fig:hetero2_robustness} (corresponding to $\rho = 0$) and Figure \ref{fig:hetero2_robustness_corr} (corresponding to $\rho = 0.5$).\color{black}
    \item \textbf{Setting 4 (Violation of the linearity assumption): }In this setting we test the robustness of the two tests to non-linearity in the model. We consider settings with similar specifications as the above three cases but with i.i.d. homoskedastic, normal errors with variance $\sigma^2$. In this case, for design matrix, $\bm X$, and error term, $\bm \epsilon$, we define,
    \[
        y_i = (\bm X_i^\delta)^T\bm \beta +\epsilon_i,
    \]
    where $\delta$ is variable we will control, and $\bm X_i^\delta$ denotes a vector whose $j^{\mathrm{th}}$ entry is the \color{black} absolute value of the $j^{\mathrm{th}}$ entry of $\bm X_i$ raised to the exponent, $\delta$, and then multiplied back by its sign\color{black}. $\delta = 1$ recovers the usual linear model and larger departure of $\delta$ from 1 imparts higher degree of non-linearity to the model. We vary $\delta$ in the set, $\{0.3,0.5,1,2,3,4\}$ and \color{black}summarize the results in Figure \ref{fig:non_linear_robustness} (corresponding to $\rho = 0$) and in Figure \ref{fig:non_linear_robustness_corr} (corresponding to $\rho = 0.5$).\color{black}
\end{itemize}
The results from all the \textcolor{black}{four} simulations suggest that the $\ell$-test and the $t$-test have similar degree of tolerance against violations of the linear model assumptions. The $t$-test exhibits robustness against the violation of the Gaussianity assumption in the error term, as Figures \ref{fig:t_robustness} and \ref{fig:gamma_robustness} indicate and we see a similar behavior for the $\ell$-test. Notably, the performance of $\ell$-test is almost similar to that of the $t$-test in the extreme cases, such as when the degrees of freedom for the error $t$-distribution is 2 in Figure \ref{fig:t_robustness} (indicating fat tails of the error distribution) or when the shape parameter of the centerd gamma distributed error is 1 (which essentially is a centerd exponential distribution with rate 1, indicating a high degree of skeweness in the error term). For Setting 3 with heteroskedastic errors, we see from Figure \ref{fig:hetero2_robustness} that the $\ell$-test's size does depart from its nominal target of $0.05$ as $\eta$ moves away from 1 (which is the homoskedastic case). However, this departure is of a similar degree as that of the $t$-test, so that both the tests exhibit similar robustness properties under this setting. From Figure \ref{fig:non_linear_robustness} as well, we see that both the $t$-test and the $\ell$-test are robust to the violation of non-linearity. \color{black}Furthermore, Figures \ref{fig:t_robustness_corr}, \ref{fig:gamma_robustness_corr}, \ref{fig:hetero2_robustness_corr} and  \ref{fig:non_linear_robustness_corr} show that the conclusions remain valid even when inter-variable dependencies are introduced with $\rho = 0.5$\color{black}. 
\begin{figure}[H]
        \centering        \includegraphics{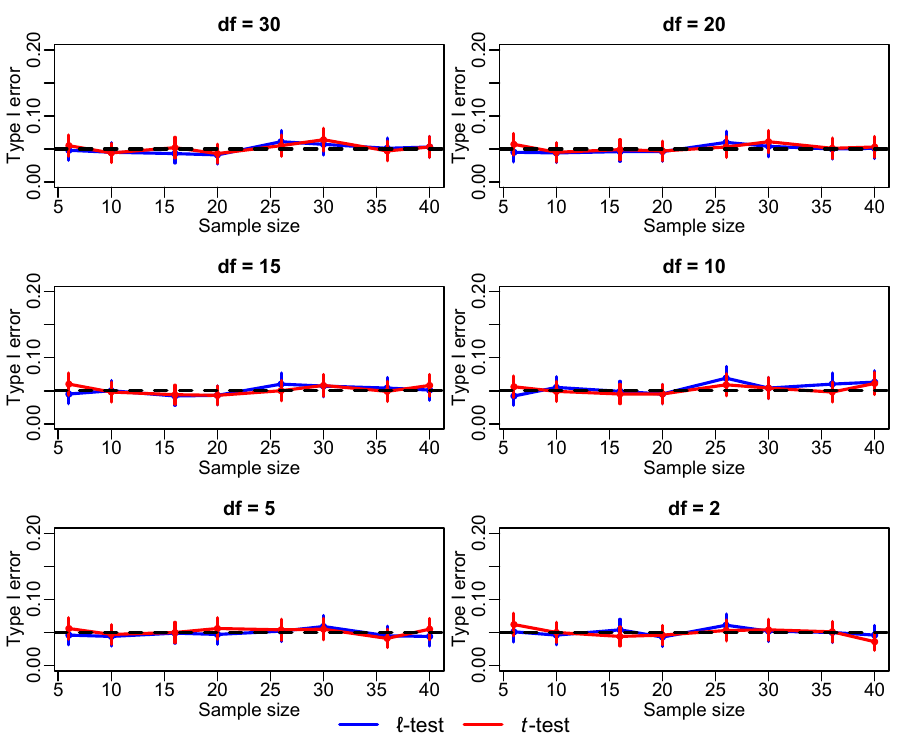}
        \caption{\color{black}$\rho = 0$. \color{black}Effect of $t$-distributed errors on the size of the $\ell$-test and the $t$-test, for different degrees of freedom. The error bars represent plus or minus two standard errors. }
        \label{fig:t_robustness}
    \end{figure}
    \begin{figure}[H]
        \centering        \includegraphics{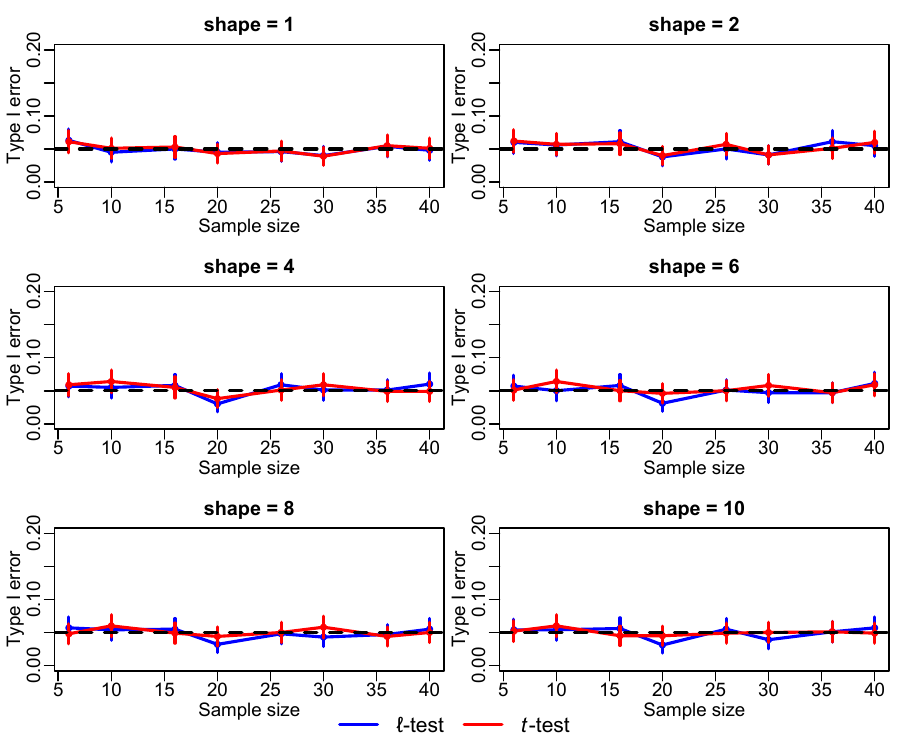}
        \caption{\color{black}$\rho = 0$. \color{black}Effect of Gamma distributed errors (with scale parameter 1) on the size of the $\ell$-test and the $t$-test, for different values of the shape parameter. The error bars represent plus or minus two standard errors.}
        \label{fig:gamma_robustness}
    \end{figure}

    \begin{figure}[H]
        \centering        \includegraphics{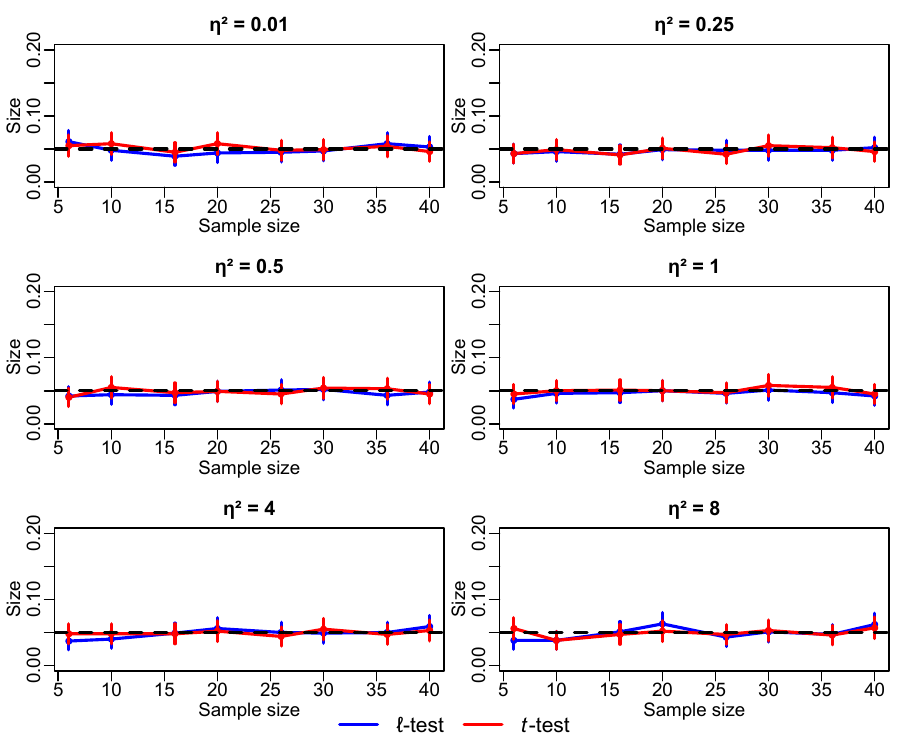}
        \caption{\color{black}$\rho = 0$. \color{black}Effect of heteroskedasticity on the size of the $\ell$-test and the $t$-test. Values of $\eta$ away from 1 indicate higher degree of heteroskedasticity. The error bars represent plus or minus two standard errors.}
        \label{fig:hetero2_robustness}
    \end{figure}
    
    \begin{figure}[H]
        \centering        \includegraphics{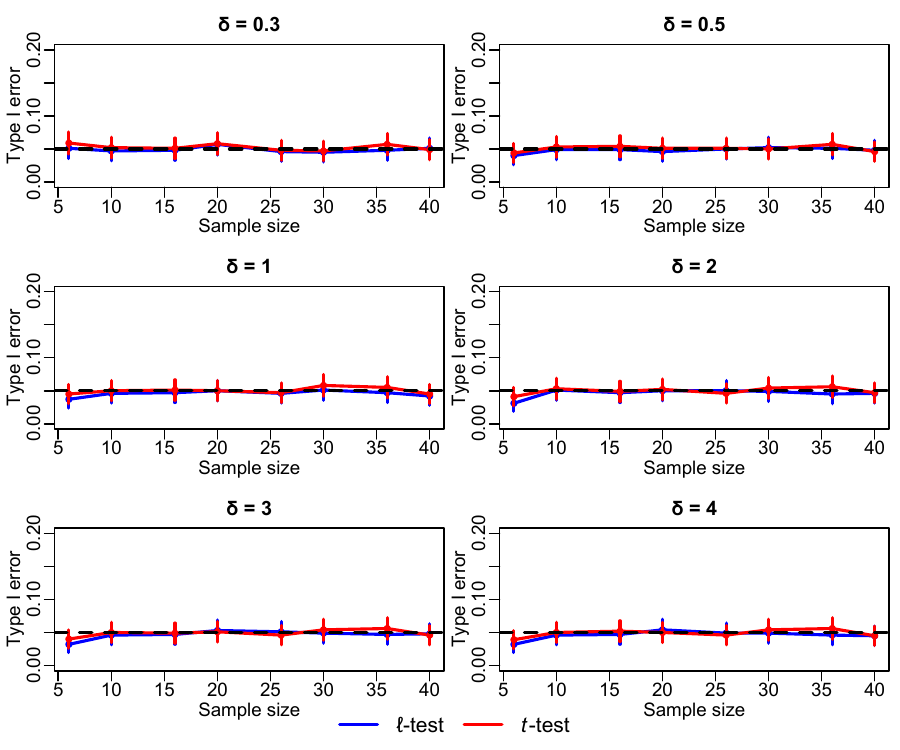}
        \caption{\color{black}$\rho = 0$. \color{black}Effect of non-linearity on the size of the $\ell$- and the $t$-test. Departure of the exponent, $\delta$ away from 1 indicates higher degree of non-linearity in the model. The error bars represent plus or minus two standard errors.}
        \label{fig:non_linear_robustness}
    \end{figure}

    \begin{figure}[H]
        \centering        \includegraphics{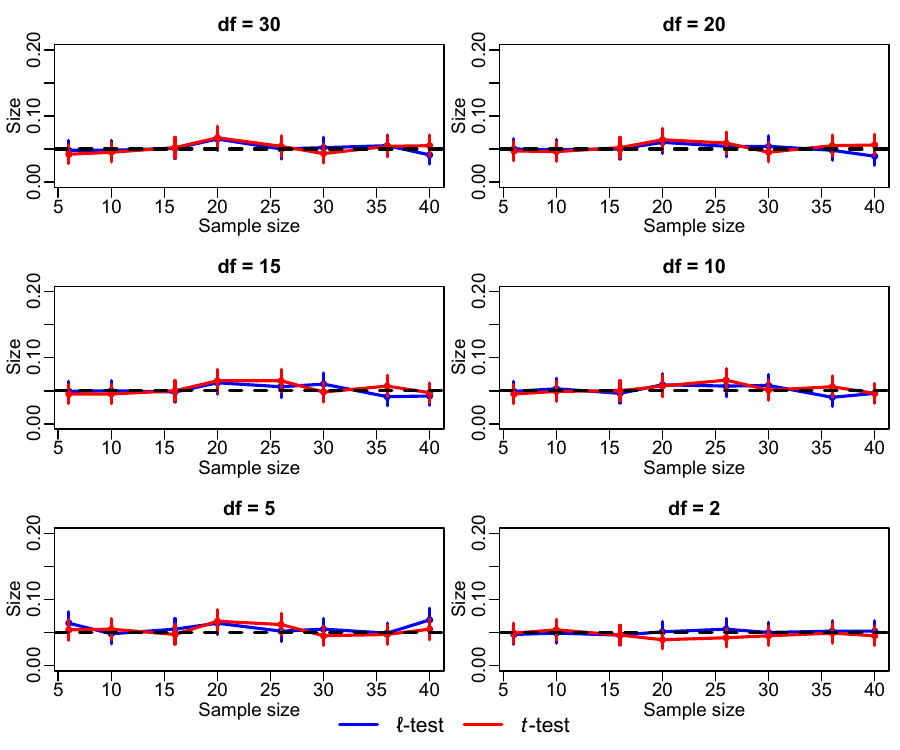}
        \caption{Same as Figure~\ref{fig:t_robustness} but with $\rho = 0.5$.}
        \label{fig:t_robustness_corr}
    \end{figure}
    \begin{figure}[H]
        \centering        \includegraphics{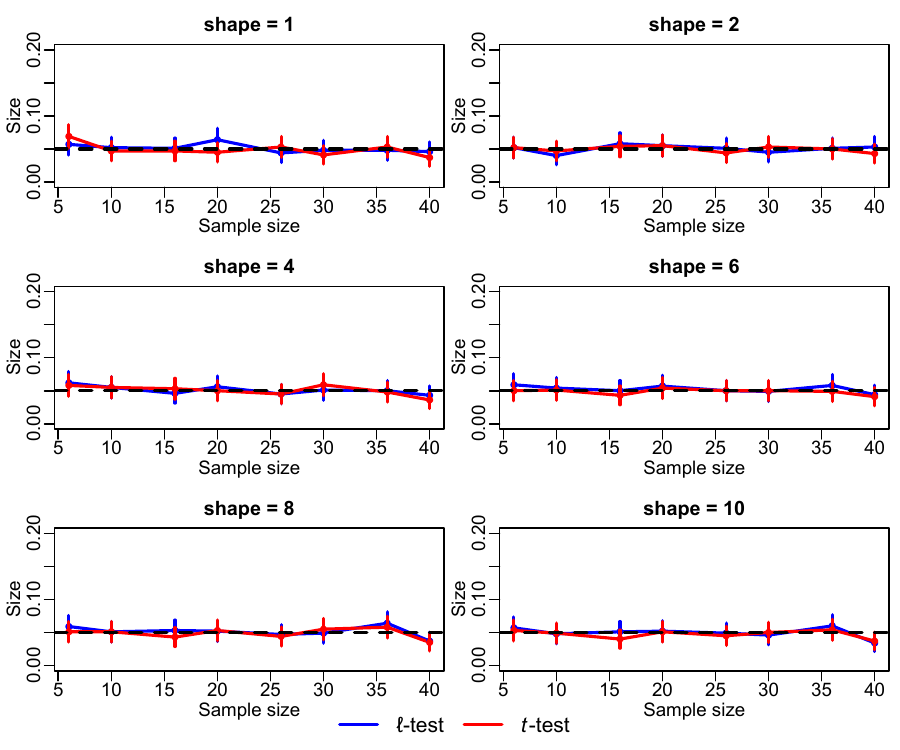}
        \caption{\textcolor{black}{Same as Figure~\ref{fig:gamma_robustness} but with $\rho = 0.5$.}}
        \label{fig:gamma_robustness_corr}
    \end{figure}
    \begin{figure}[H]
        \centering        \includegraphics{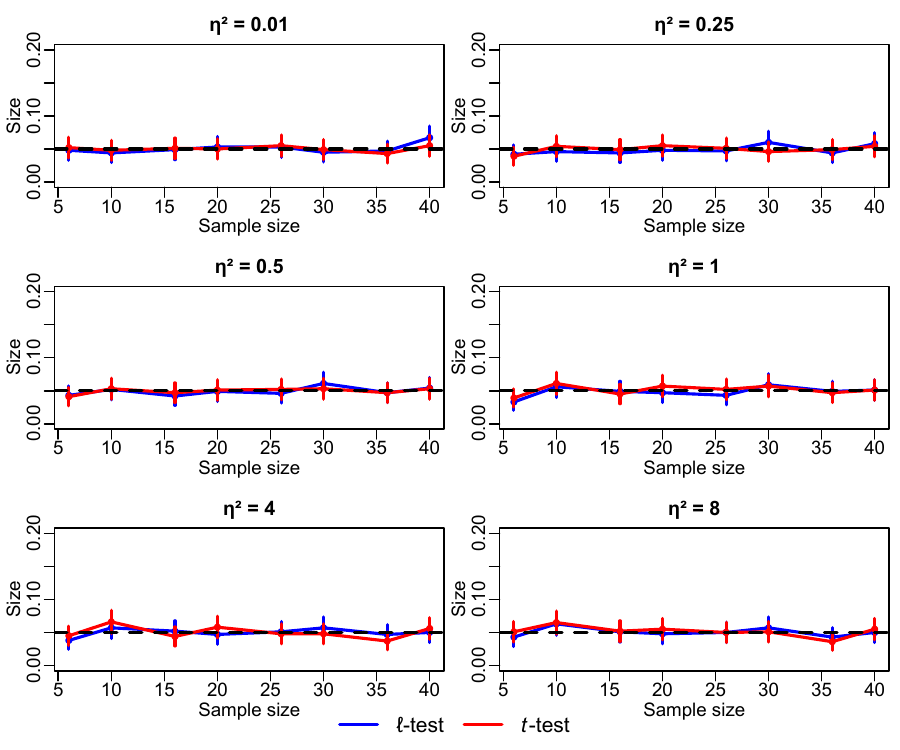}
        \caption{\textcolor{black}{Same as Figure~\ref{fig:hetero2_robustness} but with $\rho = 0.5$.}}
        \label{fig:hetero2_robustness_corr}
    \end{figure}
    \begin{figure}[H]
        \centering        \includegraphics{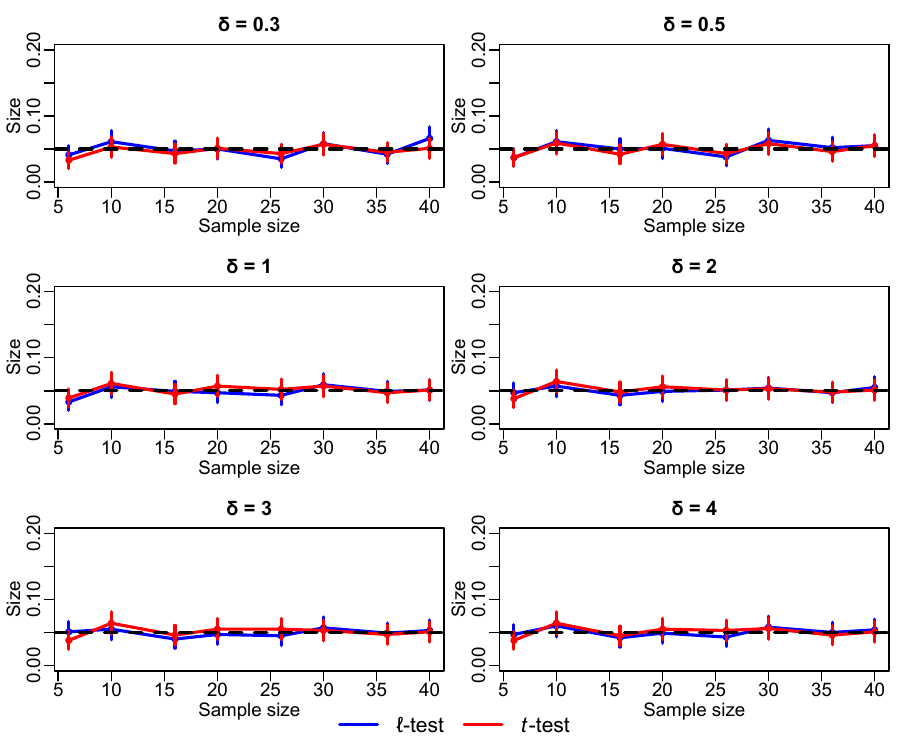}
        \caption{\textcolor{black}{Same as Figure~\ref{fig:non_linear_robustness} but with $\rho = 0.5$.}}
        \label{fig:non_linear_robustness_corr}
    \end{figure}

\subsection{Comparison of the various confidence intervals in the linear model}
\label{sec:app_ci_with_1se}
In this section, we summarize the full results of our simulations comparing the various procedures of obtaining confidence intervals for $\beta_j$. We consider four different experimental settings in Figure \ref{fig:setting_1_unadjusted_1se} (studying the effect of varying the amplitude with $d$ much smaller than $n$ ), Figure \ref{fig:setting_2_unadjusted_1se} (studying the effect of varying the amplitude with $d$ and $n$ close), Figure \ref{fig:setting_sparsity_unadjusted_1se} (studying the effect of varying the sparsity in $\bm \beta$) and Figure \ref{fig:setting_corr_unadjusted_1se} (studying the effect of varying the inter-variable correlations). The linear model we consider has exactly the same specifications as in Section \ref{sec:expt_ci_unconditional} and the captions of the respective figures contain the further details. For the $\ell$-test confidence interval, we consider the performance of the procedure when both the min rule and the 1se rule is used to choose $\lambda$. In addition to our observations in Section \ref{sec:expt_ci_unconditional} and similar to those in Section \ref{sec:app_ltest_further_expt}, we observe from Figure \ref{fig:setting_2_unadjusted_1se} that the gap between the average length of the confidence interval obtained by inverting the one-sided $t$-test and the $\ell$-test confidence interval is even smaller when $d=90$, and can again be justified by our observations in Appendix \ref{sec:app_improved_performance}. We also see that, except for the cases when the vector $\bm \beta$ is very dense, the performance of the 1se rule and the min rule is almost identical, despite the former being a stricter selection rule. Even though the former chops-off a larger mass from the distribution of $u_1$, it is less likely that this region significantly overlaps with the corresponding $5\%$ rejection region and hence the corresponding smoothed out statistic (as described in Section \ref{sec:smoothing}) under both the rules have almost similar distribution in the 5$\%$ rejection region. However, as is evident from Figure \ref{fig:setting_sparsity_unadjusted_1se}, the 1se rule does perform worse than the min rule when the most of the entries of $\bm \beta$ are signals and can be attributed to the fact that in this case the LASSO is not a good estimator of $\bm \beta$ and most of our intuitions break down.

\begin{figure}[H]
    \centering
    \includegraphics{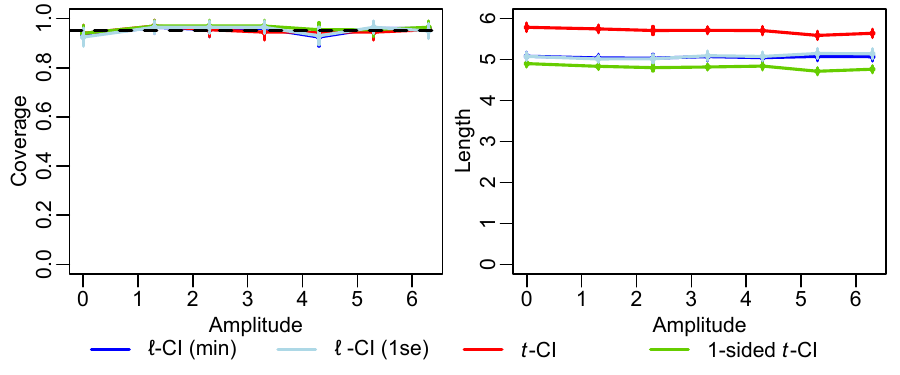}
    \caption{Length and coverage of various $95\%$ confidence intervals. We set $n=100,d=50,k=5,\bm \Sigma = \bm I, \sigma = 1$ and vary $A$. The error bars represent plus or minus two standard errors.}
    \label{fig:setting_1_unadjusted_1se}
\end{figure}

\begin{figure}[H]
    \centering
    \includegraphics{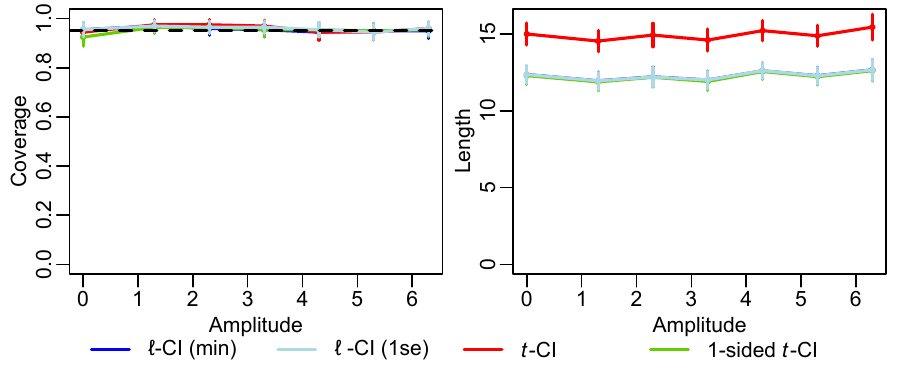}
    \caption{Length and coverage of various $95\%$ confidence intervals. We set $n=100,d=90,k=5,\bm \Sigma = \bm I, \sigma = 1$ and vary $A$. The error bars represent plus or minus two standard errors.}
    \label{fig:setting_2_unadjusted_1se}
\end{figure}

\begin{figure}[H]
    \centering
    \includegraphics{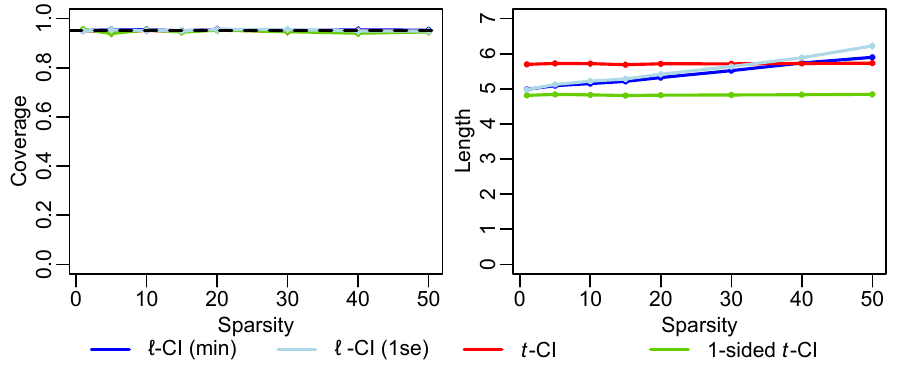}
    \caption{Length and coverage of various $95\%$ confidence intervals while varying the sparsity. We set $n=100,d=90,A=4.3,\bm \Sigma = \bm I, \sigma = 1$ and vary $k$. The error bars represent plus or minus two standard errors.}
    \label{fig:setting_sparsity_unadjusted_1se}
\end{figure}

\begin{figure}[H]
    \centering
    \includegraphics{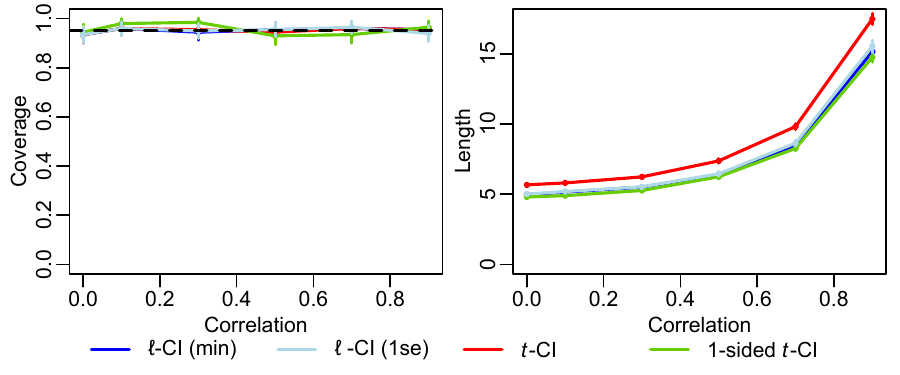}
    \caption{Length and coverage of various $95\%$ confidence intervals while varying the inter-variable correlation. We set $n=100,d=90,A=4.3,k=4,\bm \Sigma_{ij} = \rho^{|i-j|}, \sigma = 1$ and vary $\rho$. The error bars represent plus or minus two standard errors.}
    \label{fig:setting_corr_unadjusted_1se}
\end{figure}

\subsection{LASSO selection-adjusted confidence sets for linear model coefficients}
\label{sec:app_expt_lasso_adjusted_ci}
Expanding on Section \ref{sec:expt_lasso_adjusted_ci}, we compare the various conditional-on-LASSO-selection inference methods under another setting where $d=90$, with the results plotted in Figure \ref{fig:setting_2_adjusted}. As in Figure~\ref{fig:setting_1_adjusted}, there is practically no difference between the performance of $\underline{\hat C}_j^\lambda$ and $\underline{\hat C}_j^{*\lambda}$ but, unlike in Figure~\ref{fig:setting_1_adjusted}, \citet{liu2018powerful} performs quite similarly to them as well. The reason for this is that, as noted in Section~\ref{sec:relatedwork}, \citet{liu2018powerful}'s method assumes known $\sigma^2$ (and hence we provide it the true value of $\sigma^2$ in both this simulation and that of Figure~\ref{fig:setting_1_adjusted}) while the $\ell$-test does not. Unknown $\sigma^2$ is what necessitates the second component of the sufficient statistic $\bm S^{(j)}$, namely, $\bm y^T \bm y$, which the $\ell$-test conditions on but \citet{liu2018powerful}'s method does not. This difference becomes information-theoretically more pronounced as $d$ approaches $n$, as it does in the simulation in Figure~\ref{fig:setting_2_adjusted}, because the residual degrees of freedom decreases, resulting in a relative power loss for the $\ell$-test-based confidence intervals that approximately offsets the benefit of the $\ell$-test. To confirm this explanation, we can easily derive a version of the $\ell$-test and corresponding conditional confidence intervals that assumes $\sigma^2$ is known and hence does not condition on $\bm y^T \bm y$ (see footnote in Section \ref{sec:discussions}). This is plotted as dashed curves in Figure~\ref{fig:setting_2_adjusted}, and as expected these methods again provide consistently shorter confidence intervals than the method in \citet{liu2018powerful}.

\begin{figure}[H]
    \centering
    \includegraphics{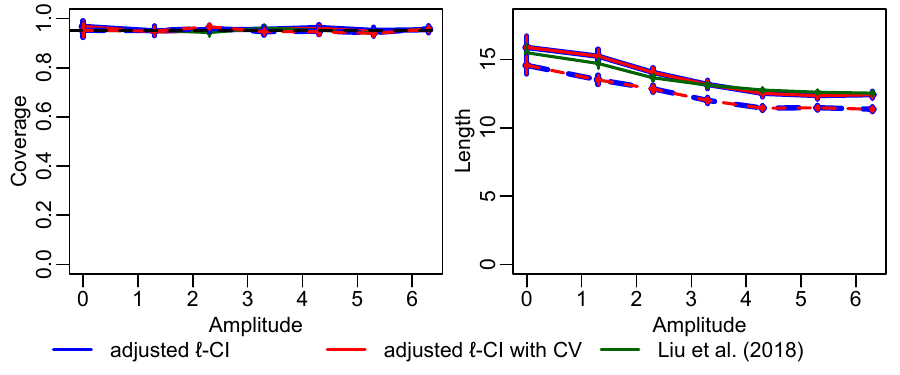}
    \caption{We consider the exact same setting as Figure \ref{fig:setting_1_adjusted} but with $d=90$. The dashed curves represent the counterparts of the corresponding $\ell$-test procedures with known $\sigma^2$. The error bars represent plus or minus two standard errors.}
    \label{fig:setting_2_adjusted}
\end{figure}

\section{Inverting the one-sided $t$-test}
\label{sec:app_1-sided-interval}
In this section, we describe the confidence interval obtained by inverting the one-sided $t$-test, as mentioned in Section \ref{sec:expt_ci_unconditional}. For the linear model, $\bm y \sim N_n(\bm X\bm \beta, \sigma^2 \bm I_n)$, the one-sided $t$-test tests $H_{j}(\gamma): \beta_j = \gamma$ by rejecting for large values of the $t$-test statistic if $\beta_j>\gamma$ and and for small values, if $\beta_j<\gamma$. In case $\beta_j = \gamma$, we, without any loss of generality, fix the convention of rejecting for small values of the $t$-test statistic, however, the validity of the test is not affected if rejecting for large values as well. Thus, with this convention, the one-sided $t$-test rejects $H_{j}(\gamma)$ when $\frac{\hat \beta_{j, \mathrm{OLS}}-\gamma}{\hat{\mathrm{SE}}\left(\hat \beta_{j, \mathrm{OLS}}\right)}> t_{\alpha;n-d}$, if $\beta_j>\gamma$ and when $\frac{\hat \beta_{j, \mathrm{OLS}}-\gamma}{ \hat{\mathrm{SE}}\left(\hat \beta_{j, \mathrm{OLS}}\right)}< -t_{\alpha;n-d}$, if $\beta_j\leq \gamma$. Here $t_{\alpha;n-d}$ is the quantile of the $t_{n-d}$ distribution putting mass $\alpha$ on its upper tail. Inverting this test gives the following one-sided $100(1-\alpha)\%$ confidence interval,
\begin{align*}
    C^t_{\mathrm{1-sided}} = \begin{cases}
        \left[\hat{\beta}_{j,\mathrm{OLS}}\pm t_{\alpha;n-d}\hat{\mathrm{SE}}\left(\hat \beta_{j, \mathrm{OLS}}\right)\right], &\textrm{ if } \beta_j \in  \left[\hat{\beta}_{j,\mathrm{OLS}}\pm t_{\alpha;n-d}\hat{\mathrm{SE}}\left(\hat \beta_{j, \mathrm{OLS}}\right)\right]\\
        \left[\beta_j, \hat{\beta}_{j,\mathrm{OLS}}+ t_{\alpha;n-d}\hat{\mathrm{SE}}\left(\hat \beta_{j, \mathrm{OLS}}\right)\right), &\textrm{ if } \beta_j < \hat{\beta}_{j,\mathrm{OLS}}- t_{\alpha;n-d}\hat{\mathrm{SE}}\left(\hat \beta_{j, \mathrm{OLS}}\right)\\
        \left(\hat{\beta}_{j,\mathrm{OLS}}- t_{\alpha;n-d}\hat{\mathrm{SE}}\left(\hat \beta_{j, \mathrm{OLS}}\right), \beta_j\right), &\textrm{ if } \beta_j > \hat{\beta}_{j,\mathrm{OLS}} + t_{\alpha;n-d}\hat{\mathrm{SE}}\left(\hat \beta_{j, \mathrm{OLS}}\right).
    \end{cases}
\end{align*}
Note that the above interval is indeed an oracle interval as one needs to know the exact value of $\beta_j$ to construct it. Furthermore note that when $\beta_j > \hat{\beta}_{j,\mathrm{OLS}} + t_{\alpha;n-d}\hat{\mathrm{SE}}\left(\hat \beta_{j, \mathrm{OLS}}\right)$, the interval surely misses the target parameter, $\beta_j$, and this is because of the one-sided nature of the test when we are testing at the true parameter, $\gamma = \beta_j$. In case we were rejecting for large values of the $t$-test statistic on observing $\beta_j$, we would get an open interval in the case when $\beta_j < \hat{\beta}_{j,\mathrm{OLS}} - t_{\alpha;n-d}\hat{\mathrm{SE}}\left(\hat \beta_{j, \mathrm{OLS}}\right)$ instead.

\end{document}